\documentclass[11pt]{article} 
\usepackage{geometry} 
\usepackage{graphicx, amsmath,amsfonts,cite,appendix, amssymb,xcolor,enumerate, mathrsfs, amsthm,url,hyperref,mathtools,subfig} 
\usepackage{epsfig}
\usepackage{booktabs} 
\usepackage{array} 
\usepackage{paralist} 
\usepackage{verbatim} 

\numberwithin{equation}{section}

\newcommand{\be}{\begin{equation}}
\newcommand{\ee}{\end{equation}}			
\newcommand{\ben}{\begin{equation*}}
\newcommand{\een}{\end{equation*}}
\newcommand{\mc}{\mathcal}
\newcommand{\mbf}{\mathbf}
\newcommand{\msf}{\mathsf}
\newtheorem{lemma}{Lemma}[section]

\newtheorem{thm}{Theorem}
\newtheorem{fact}{Fact}
\newtheorem{note}{Note}
\newtheorem{rem}{Remark}

\newcommand{\e}{\epsilon}

\newcommand{\abs}[1]{\left \vert#1\right \vert}
\newcommand{\norm}[1]{\lVert#1\rVert}
\newcommand{\ex}{\mathbb{E}}
\newcommand{\wh}[1]{\widehat{#1}}

\newcommand{\prb}{\mathbb{P}}

\newcommand{\by}{\mathbf{y}}
\newcommand{\bw}{\mathbf{w}}
\newcommand{\bay}{\textsf{Bayes}}
\newcommand{\bst}{\boldsymbol{\theta}}
\newcommand{\bsth}{\hat{\boldsymbol{\theta}}}
\DeclareMathOperator*{\argmin}{arg\,min}

\usepackage{sectsty}
\allsectionsfont{\sffamily\mdseries\upshape} 
\title{\vspace{-20pt} Empirical Bayes Estimators for High-Dimensional Sparse Vectors}
\author{{Pavan Srinath}\\{University of Cambridge}
  \\{\small {srinath.pavan@gmail.com}}
 \and Ramji Venkataramanan\\ {University of Cambridge}
  \\ {\small{ramji.v@eng.cam.ac.uk}}
}

\begin{document}
\maketitle
\thispagestyle{empty}	

\begin{abstract}
{The problem of estimating a high-dimensional sparse vector $\bst \in \mathbb{R}^n$ from an observation in i.i.d. Gaussian noise is considered. The performance is measured using squared-error loss. An empirical Bayes shrinkage estimator, derived using a Bernoulli-Gaussian prior, is analyzed and compared with the well-known soft-thresholding estimator.  We obtain  concentration inequalities for the Stein's unbiased risk estimate and the loss function of both  estimators. The results show that for large $n$, both the risk estimate and the loss function concentrate on deterministic values close to the true risk.

Depending on the underlying $\bst$, either the proposed empirical Bayes (eBayes) estimator or soft-thresholding may have smaller loss. We consider a hybrid estimator that attempts to pick the better of the soft-thresholding estimator and the eBayes estimator by comparing  their risk estimates. It is shown that: i) the loss of the hybrid estimator concentrates on the minimum of the losses of the two competing estimators, and ii) the risk of the hybrid estimator is within order $\frac{1}{\sqrt{n}}$ of  the minimum of the two risks. Simulation results are provided to support the theoretical results. Finally, we use the eBayes and hybrid estimators as  denoisers in the  approximate message passing (AMP) algorithm for compressed sensing,  and show that their performance is superior to the soft-thresholding denoiser in a wide range of settings.
}
{Sparse estimation, Shrinkage estimators, Stein's unbiased risk estimate, Soft-thresholding,  Large deviations, Concentration inequalities}
\end{abstract}


\section{Introduction} \label{sec_intro}
Consider the problem of estimating a sparse vector  $\bst \in \mathbb{R}^n$ from a noisy observation 
$\mbf{y}$ of the form
\be \mbf{y} = \bst + \mbf{w}. \label{eq:obs_model} \ee
The noise vector $\mbf{w} \in \mathbb{R}^n$ is  distributed as $\mc{N}(\mbf{0},  \mbf{I})$, i.e., its components are i.i.d. Gaussian  with mean zero and  unit variance.\footnote{The case where $\mbf{w} \sim \mathcal{N}(\mbf{0},  \sigma^2\mbf{I})$ with $\sigma$ known reduces to the above form by rescaling $\mbf{y}$ by $1/\sigma$, so that $\bst/\sigma$ is to be estimated.}

In this paper, the performance of an estimator $\bsth$ is measured using the squared-error loss function given by 
 $L(\boldsymbol{\theta}, \hat{\boldsymbol{\theta}}(\mbf{y}) ) \vcentcolon= \norm{\hat{\boldsymbol{\theta}}(\mbf{y})- \boldsymbol{\theta}}^2$,
 where $\norm{\cdot}$ denotes the Euclidean norm. The \emph{risk} of the estimator for a given $\bst$ is the expected value of the loss function:
 \[ R(\boldsymbol{\theta}, \hat{\boldsymbol{\theta}} ) \vcentcolon=  \mathbb{E}_{\bst}\left[ \norm{ \bsth(\mbf{y}) - \boldsymbol{\theta}}^2 \right]. \]
 We emphasize that $\bst$ is deterministic, so the expectation above is computed over $\mbf{y} \sim \mc{N}(\mbf{\bst}, \mbf{I})$. In the remainder of the paper, for brevity we drop the subscript on the expectation. 
 We assume that $\bst$ has $k$ non-zero entries out of $n$, where $k$ may not be known to the estimator. Though our results are general, they are most interesting for the case where $k = \Theta(n)$. Thus as $n$ gets large, the sparsity level $\eta:=k/n$ is bounded above and below by arbitrary constants in $(0,1]$.
 
 The sparse estimation problem has been widely studied  \cite{donoho1994ideal,donoho1994minimax,donoho_johnstoneWav,johnstone2004needles,johnstone2005empirical, LeungBarron06, leungthesis04, carvalho_horseshoe,johnstoneBook} due to its fundamental role in non-parametric function estimation (see, e.g., \cite[Sec. 1.10]{tsybakov09Book}). Indeed, if the function has a sparse representation in an orthogonal basis (e.g., a Fourier or wavelet basis),  then \eqref{eq:obs_model} models the problem of estimating the function from a noisy measurement of $n$ basis coefficients. 
 Another motivation for  constructing good sparse estimators  comes from Approximate Message Passing (AMP)  algorithms for compressed sensing. Recall that the goal in compressed sensing \cite{CandesTaoLP, donohoCS, candesRT06} is to recover a sparse vector $\bst$ from a noisy linear measurement of the form $\mbf{A}\bst  + \text{noise}$, where $\mbf{A} \in \mathbb{R}^{m \times n}$ is a measurement matrix with $m < n$. AMP \cite{DonMalMont09,BayMont11,bayMontLASSO,BayatiLM15, krz12,Rangan11} refers to  a class of low-complexity iterative algorithms that can be used to estimate $\bst$, under certain conditions on the measurement matrix $\mbf{A}$. In each iteration,  AMP produces an effective observation vector  that is well-approximated  as the sum of the desired signal $\bst$ and a Gaussian noise vector, i.e.,  the effective observation is well-represented by the model \eqref{eq:obs_model}.   Then, the AMP algorithm uses a sparse estimator  to generate an updated estimate of $\bst$  from the effective observation in each iteration. We discuss the application of the  sparse estimators proposed in this paper to compressed sensing AMP in Sec. \ref{subsec_AMP}.

Thresholding estimators are a popular class of estimators for the model \eqref{eq:obs_model} when $\bst$ is assumed to be sparse \cite{johnstoneBook, donoho1994ideal,donoho1994minimax, donoho_johnstoneWav, johnstone2004needles,johnstone2005empirical}. In these estimators, the entries of $\mbf{y}$ whose absolute value falls below a threshold $\lambda >0$ are set to zero. The remaining entries  of $\mbf{y}$ may either be retained without modification (hard-thresholding), or  shrunk towards the origin by an amount $\lambda$ (soft-thresholding). The soft thresholding estimator $\bsth_{ST}$ with threshold $\lambda$  is given by 
\begin{equation}\label{eq_ST}
 \hat{\theta}_{ST,i}(y_i;\lambda) = \left\{\begin{array}{cl}
                               y_i - \lambda & \textrm{if }y_i > \lambda \\
                               0, & \textrm{if } -\lambda \leq y_i \leq \lambda \\
                               y_i + \lambda & \textrm{if }y_i < -\lambda.\\                               
                              \end{array}\right., \quad i \in [n].
\end{equation}
Thresholding estimators have many attractive properties. For example, when $n$ is large and the sparsity level $\eta =k/n \to 0$,  the worst-case risk  over the set of $\eta$-sparse vectors is  $2 \eta \log \eta^{-1}(1 + o(1))$. This is close to minimax over the set since only the $o(1)$ term can be improved by a better estimator \cite[Chapter 8]{johnstoneBook}. However, no sharp theoretical bounds exist for the risk of thresholding estimators for moderate or large values of $\eta$.

In this paper, alongside soft-thresholding, we consider an empirical Bayes estimator derived using a Bernoulli-Gaussian prior. This estimator is motivated by the empirical Bayes derivation of James-Stein (shrinkage) estimators by Efron and Morris\cite{efron_morris}. For the observation model given by \eqref{eq:obs_model}, if we assume a Gaussian prior $ \mathcal{N}(\mu \textbf{1},  \xi^2\mbf{I})$ on $\bst$ (where $\mbf{1}$ denotes the all-ones vector), then the Bayes estimator is 
\be
\bsth_{\bay} = \mu \mbf{1} + \left(1 - \frac{1}{1+\xi^2} \right)\left( \by - \mu \mbf{1} \right).
\label{eq:Bayes_est}
\ee
 In \cite{efron_morris}, Efron and Morris use plug-in estimates for $\mu$ and  $1/(1 + \xi^2)$, based on 
\ben
\ex \left[ \frac{1}{n} \sum_{i=1}^n y_i \right] = \mu \quad \text{ and } \quad  \ex \left[\frac{n-3}{\norm{\mbf{y} - \mu \mbf{1}}^2} \right] = \frac{1}{1+\xi^2},
\een
to obtain the following shrinkage estimator:
\be 
\bsth_{\msf{L}} = \bar{y} \mbf{1} + \left(1 - \frac{n-3}{ \norm{\mbf{y} - \bar{y} \mbf{1}}^2}  \right)_+\left( \by - \bar{y} \mbf{1} \right). 
\label{eq:lindley_est}
\ee
Here $\bar{y}=\sum_i y_i/n$, and the notation $x_+$ denotes $\max (x,0)$.  The estimator $\bsth_{\msf{L}}$ in \eqref{eq:lindley_est} is the positive-part version of Lindley's estimator \cite{lindley,baranchik}, which shrinks each element of $\mbf{y}$ towards the empirical mean 
$\bar{y}$. Taking the positive-part of the shrinkage term ensures that it is always non-negative, as in the underlying Bayes estimator \eqref{eq:Bayes_est}. When there no  assumptions on the structure of $\bst$, the shrinkage estimator $\bsth_{\msf{L}}$ has several attractive properties including uniform dominance of the maximum-likelihood estimator (see, for example, \cite[Chapter $5$]{lehmannCas98}).

 In our model, it is  known that $\bst$ is sparse, though the sparsity level $\eta$ may be unknown.  To incorporate this knowledge, we consider an empirical Bayes estimator derived using a prior  for each element of $\bst$ that is a mixture of a point mass at $0$ and a continuous distribution with density $\psi(\theta; \mu,\xi)$, where $\mu$ is a location parameter (mean) and $\xi$ is a scale parameter. With a mixture weight $\e \in [0,1]$ to control the sparsity, the prior is given by
 \be
f(\theta; \e, \mu, \xi) = (1- \e) \delta(\theta) + \e \,  \psi(\theta; \mu, \xi), \quad \theta \in \mathbb{R}.
\label{eq:eb_prior}
\ee
As above, assuming $\psi$ to be the Gaussian density, we can derive an empirical Bayes estimator using plug-in estimates $\hat{\mu}, \wh{\xi^2}$ for the location and scale parameters, respectively. The resulting empirical Bayes (eBayes) estimator is given in  \eqref{eq:eBayes} in the next section. The mixture weight $\e$, which determines the sparsity of the prior, is treated as a fixed parameter that could be optimized. In particular, it need not be the true sparsity $\eta$ (which may be unknown).

\begin{figure*}[t]
    \centering
    \begin{subfloat}[\label{fig1}]{
       \includegraphics[trim= 0.9in 0 1.22in 0, clip=true,width= 2.7in,height=2.2in]{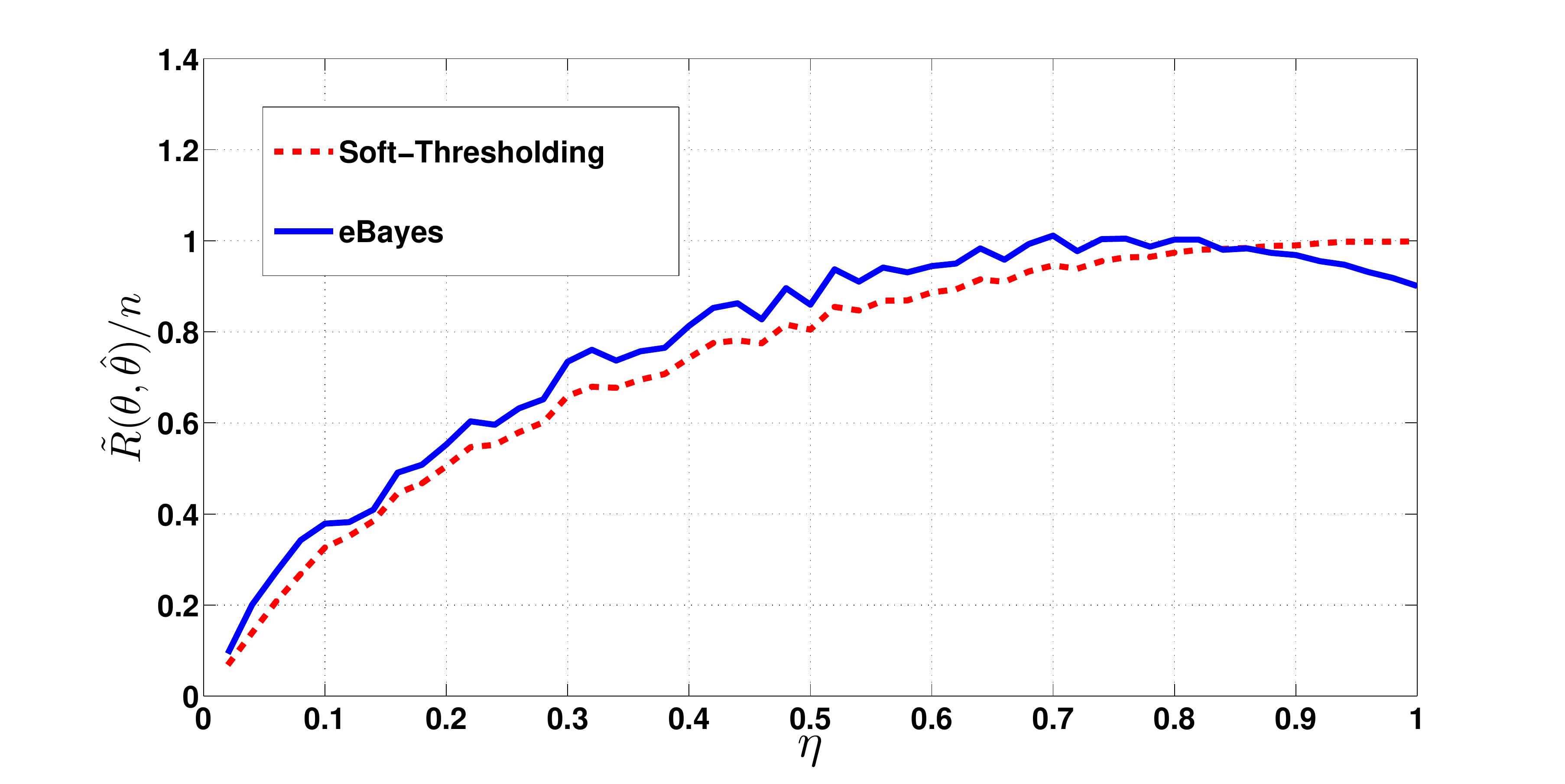}   
       }
    \end{subfloat}
  \quad
    \begin{subfloat}[\label{fig2}]{
       \includegraphics[trim= 0.9in 0 1.22in 0, clip=true,width= 2.7in,height=2.2in]{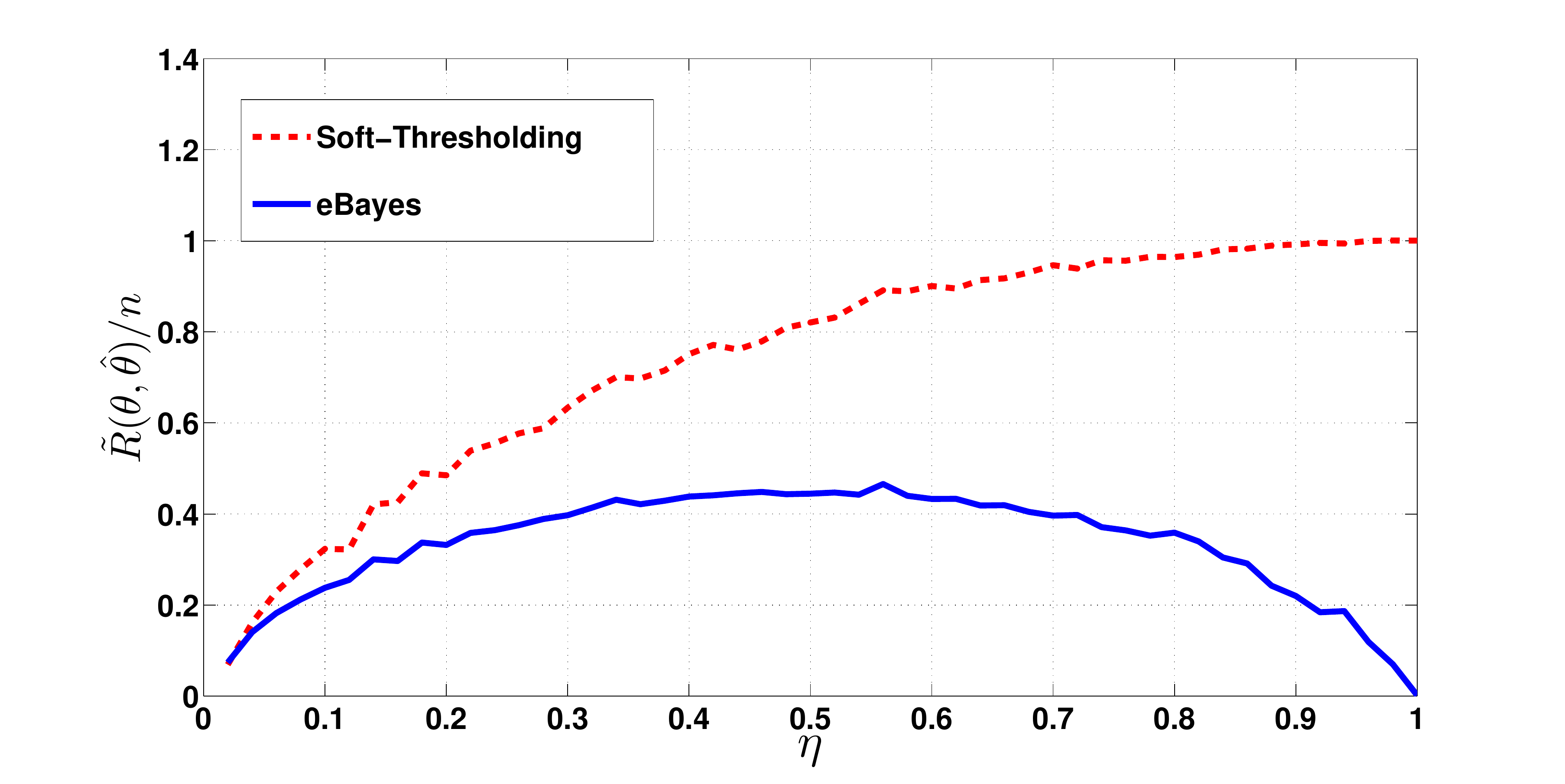}
       }
    \end{subfloat}    
    \caption{\small Average normalized loss $\tilde{R}(\bst,\bsth)/n$ with $n=1000$ for the following cases: a) Half the non-zero entries in $\bst$ equal $3$ and the other half $-3$. b) All the non-zero entries in $\bst$ equal $3$. In each case, the average normalized loss  is computed over 1000 independent realizations of the noise vector $\bw$.}
    \label{fig_12}
\end{figure*}

Depending on the underlying $\bst$ and the noise realization $\mbf{w}$, either the soft-thresholding estimator or the eBayes estimator may have the smaller loss. This is illustrated in Fig.  \ref{fig_12}, which compares the performance  of the two estimators for two different kinds of $\bst$ of length $n=1000$. The average normalized losses  for the two cases are shown Figs. \ref{fig1} and \ref{fig2} as a function of the true sparsity level $\eta = \e$, which is assumed to be known for both estimators. In the figures, the threshold $\lambda^*$ for $\bsth_{ST}$ is chosen as  \cite[Sec. 3]{montanari_graphical_models}
\begin{align*}
 \lambda^*& = \argmin_{\lambda  \geq 0}\left\{\e(1 + \lambda^2) + (1-\e)\left[2(1+\lambda^2)\Phi(-\lambda) - 2 \lambda \phi(\lambda) \right] \right \},
\end{align*}
where $\phi$ is the standard normal density, and $\Phi(x) \vcentcolon= \int_{-\infty}^{x} \phi(u) du$.  This choice $\lambda^*$ minimizes the worst-case soft-thresholding risk over the class of all  
$\bst$  with sparsity level $\e$ \cite{donoho1994minimax}. 

The plots indicate that depending on the underlying $\bst$, either $\bsth_{ST}$ or $\bsth_{EB}$ may have smaller loss.  The goal is to construct an estimator that reliably chooses the estimator with lower loss. Noting that the loss depends on the underlying $\bst$ as well as the noise realization, we propose a hybrid estimator that chooses one of the two  competing estimators by comparing their Stein's unbiased risk estimates (SURE) \cite{stein2}. These risk estimates are given in Section \ref{sec_risk_estimator}. A key result of this paper is that for any $\bst$, the loss of the hybrid estimator (which chooses one of the two estimators based on SURE) concentrates on the smaller of the two losses. In particular,  the probability of the actual normalized loss deviating from the smaller one by more than $t$  decays exponentially in $n\min\{t, t^2\}$ for $t >0$. 

The contributions of this paper are as follows:
\begin{itemize}
\item   We derive the eBayes estimator in Sec. \ref{sec_ebayes}, and  a risk function estimator based on Stein's unbiased risk estimate (SURE)  in  Sec. \ref{sec_risk_estimator}.
 
 \item  Sec. \ref{sec_loss_conc}  contains the main theoretical results of the paper. The first result (Theorem \ref{thm1}) is a concentration inequality for the SURE of eBayes, which shows that for large $n$, the risk estimate concentrates on a deterministic value which is within $\mc{O}({1}/{\sqrt{n}})$ of the true risk.  We remark that unlike the soft-thresholding estimator, the concentration of the SURE for eBayes cannot be established directly via readily available Gaussian concentration results as it does not satisfy Lipschitz or similar conditions.
 
We then show in Theorem \ref{thm1a}  that the loss of the eBayes estimator concentrates on a deterministic value that is within $\mc{O}({1}/{\sqrt{n}})$ of the eBayes risk. Theorem \ref{thm1b}  shows that soft-thresholding  loss concentrates on the true risk of soft-thresholding.  Finally, we use the above concentration results to analyze the performance of a hybrid estimator which chooses the estimator (soft-thresholding or eBayes) with the smaller risk estimate. Theorem \ref{thm2} shows that  for  the  hybrid estimator, the loss concentrates on the minimum of the losses of the two rival estimators, and its risk is within $\mc{O}({1}/{\sqrt{n}})$ of the minimum of the two risks.  Thus, the hybrid estimator uses the data to reliably choose an estimator tailored to the underlying $\bst$. 

  \item Sec. \ref{sec_simulation}  provides simulation results  to validate the theoretical results. The simulation results suggest that the proposed eBayes estimator is superior to soft-thresholding in a variety of cases, including the case where the non-zero entries come from a distribution with heavier-than-Gaussian tails, e.g., the Laplace distribution.   In Sec. \ref{subsec_AMP}, we use the eBayes and the hybrid estimators as denoisers in the AMP algorithm for compressed sensing, and compare their performance with that of the soft-thresholding denoiser. 
\end{itemize}

The approach taken in this paper of obtaining concentration inequalities for risk estimates can be used to bound the risk of a hybrid estimator that picks one among several estimators,  provided one has concentration bounds for the risk estimates of each of the competing estimators.  This suggests that an interesting direction for future work  is to obtain concentration bounds for the risk estimates and loss functions of other useful estimators whose parameters depend on the data, e.g., an empirical Bayes estimator based on a Bernoulli-Laplace prior. 

\subsection{Related Work}

In the context  of wavelets, several works have considered estimators based on a signal prior  that is a mixture of a point mass at $0$ and a Gaussian distribution,  see e.g., \cite{abramovich1998wavelet,clyde1998multiple}. In most of these works,  the hyperparameters of the prior are chosen based on some prior information about the signal.    Martin and Walker \cite{martin2014asymptotically} propose an  estimator based on a data-dependent prior, and show that the resulting empirical Bayes estimator is asymptotically minimax.  Johnstone and Silverman \cite{johnstone2004needles, johnstone2005empirical} proposed empirical Bayes estimators  based on a prior that is a mixture of a point mass at $0$ and a distribution with a heavy-tailed density. The weights of the mixture  are first determined using marginal log-likelihood; the estimator then uses  a thresholding rule based on the posterior median.  It was shown that the risk of  this estimator over the class of $\eta$-sparse vectors is within a constant factor of the minimax risk when the sparsity level $\eta$ is small enough. 

In this paper, we fix the mixture weight for the eBayes estimator and then empirically estimate the location and scale parameters of the continuous part of the prior. This allows us to obtain concentration inequalities for the risk estimates, which then lead to concentration results for the loss and bounds for the risk for both the eBayes and the hybrid estimator. 

Our previous work \cite{SrinathV16} also used concentration inequalities to characterize the performance of estimators with data-driven parameters. However, the estimators  proposed in that paper were for general $\bst$, as opposed to the sparse $\bst$ considered here.  Moreover, the loss function estimates in  \cite{SrinathV16} are not based on SURE as the estimators are not smooth. Consequently, the techniques required to obtain the concentration results in  \cite{SrinathV16} are quite different from those used here.

A recent paper by Zhang and Bhattacharya \cite{ZB17}  also considers an empirical Bayes estimator defined via the prior in \eqref{eq:eb_prior}. The parameters of the prior are estimated by maximizing the marginal likelihood using the EM algorithm,  and the properties of the posterior median and the posterior mean are studied. When the  density $\psi$ in the prior is unimodal and satisfies certain conditions, it is shown in \cite[Theorem 2.2]{ZB17} that the SURE corresponding to the posterior mean is within $\mathcal{O}(\frac{(\log n)^{3/2}}{\sqrt{n}})$ of the true risk with high probability. We  comment on the differences between this result and our SURE concentration result  (Theorem \ref{thm1}) in Note \ref{note:ZBcomparison} on p.\pageref{note:ZBcomparison}.

As an alternative to using a hybrid estimator that picks one of several  estimators based on risk estimates, George \cite{george1} and Leung and Barron \cite{LeungBarron06, leungthesis04} have proposed combining the estimators using exponential mixture weights  based on the risk estimates. We note that in high dimensions, the weight assigned to the estimator with the smallest risk estimate is exponentially larger (in $n$) than the others, so it is effectively equivalent to picking the estimator with the smallest risk estimate.

\emph{Notation}: The set $\{1,2,\cdots,n\}$ is denoted by $[n]$. Bold lowercase (uppercase) letters are used to denote vectors (matrices), and plain lowercase letters for their entries.  For example, the entries of  $\mathbf{y}$ are  $y_i$, $i=1,\cdots,n$. All vectors have length $n$ and are column vectors. The transpose of $\by$ is denoted by $\by^T$. The complement of an event $\mc{E}$ is denoted by $\mathcal{E}^c$, and its indicator function by $\mathsf{1}_{\{\mathcal{E}\}}$. For a random variable $X$, $X_+$ denotes $\max(0,X)$. For positive-valued functions $f(n)$ and $g(n)$, the notation $f(n) = \mathcal{O}(g(n))$ means that $\exists k>0$ such that $\forall n>n_{0}$, $f(n) \leq k g(n)$. Also, for a sequence of random variables $\{ X_n, n=1,2,\cdots \}$ and a sequence of deterministic numbers $\{ a_n, n=1,2,\cdots \}$, the notation $X_n = \mathcal{O}_P(a_n)$ implies that for any $\delta > 0$, there exists a finite $M > 0$ and a finite $N > 0$ such that $\prb\left(\vert X_n/a_n\vert \geq M \right) \leq \delta$, $ \forall n > N$.


\section{Empirical Bayes Estimator}\label{sec_ebayes}

If the $\{\theta_i \}, i \in [n]$ were generated i.i.d. according to the distribution $f(\theta; \e, \mu, \xi)$ in \eqref{eq:eb_prior}, then  the conditional mean  of $\theta$ given $y$  is the optimal estimator for squared-error loss. The empirical Bayes estimator for a \emph{fixed} $\e \in [0,1]$ is this conditional mean, with the values of $\mu, \xi$ estimated from the data $\mbf{y}$. Hence, $ \forall i \in [n]$,
\be
\bsth_{EB, i}(\mbf{y;\e}) = \frac{\int_{\mathbb{R}}  x f(x; \e, \hat{\mu}, \hat{\xi}) \phi(y_i-x) dx}{ \int_{\mathbb{R}} f(x; \e, \hat{\mu}, \hat{\xi}) \phi(y_i-x) dx}.
\label{eq:EB_estimator}
\ee 
In \eqref{eq:EB_estimator}, $\phi(x):=\frac{1}{\sqrt{2\pi}}e^{-x^2/2}$ is the standard normal density, and $\hat{\mu}, \hat{\xi}$ are the estimates of  $\mu, \xi$ from $\mbf{y}$. A consistent estimator for the location parameter $\mu$ (converging in probability to $\mu$) is
\be \hat{\mu}(\mbf{y}) = {\bar y}/ {\e},   \label{eq:mean_estimate} \ee
where  the empirical  mean $\bar{y} = \sum_{i} y_i /n$. The scale parameter can be estimated using the second moment 
$\overline{y^2} := \norm{\mbf{y}}^2/n$ and the mean $\bar{y}$. In this paper, we consider the Gaussian density for $\psi$ in \eqref{eq:eb_prior} so that \[\psi(\theta; \mu, \xi) = \frac{1}{\sqrt{2 \pi \xi^2}} \exp(-(\theta-\mu)^2 / 2 \xi^2).\]
The mean $\mu$ is estimated as in \eqref{eq:mean_estimate}, and $\xi^2$, being the variance, is estimated as 
\be \widehat{\xi^2}(\mbf{y}) = \frac{1}{\e}\left( \overline{y^2} - \frac{(\bar{y})^2}{\e}  - 1\right)_+. \label{eq:xi_estimate} \ee
The resulting empirical Bayes estimator is 
\be \label{eq:eBayes}
\bsth_{EB, i}(\mbf{y;\e})  =  \frac{\hat{\mu} + \left( 1 - \frac{1}{ 1 + \widehat{\xi^2}}\right)(y_i - \hat{\mu})}{1 + \frac{(1-\epsilon)}{\e} \sqrt{ 1 + \widehat{\xi^2}}\,
\exp\left(- \frac{y_i^2}{2} + \frac{(y_i-\hat{\mu})^2}{2(1+\widehat{\xi^2})}\right)}, \ i \in [n].
\ee
For $\e=1$, $\bsth_{EB}$ reduces to the  positive-part Lindley's estimator given in \eqref{eq:lindley_est}. 

Note that $\bsth_{EB}$ is a shrinkage estimator --- the numerator shrinks each $y_i$ towards a common element $\hat{\mu}$. There are two terms that determine the overall shrinkage, the first being the term $ \left[ 1 - \frac{1}{ 1 + \widehat{\xi^2}}\right]$ which is common for all the $y_i$. The second term influencing the shrinkage  is the exponential in the denominator which depends on $y_i$. To get intuition about the role of these terms,   assume that the location parameter is zero, i.e., $\hat{\mu} = 0$ in \eqref{eq:eBayes}. Then the estimator is given by
\be \label{eq:eBayes_zero_location}
\bsth_{EB, i}(\mbf{y};\e)  =  \frac{ \left(  \frac{ \widehat{\xi^2}}{ 1 + \widehat{\xi^2}}\right)y_i }{1 + \frac{(1-\epsilon)}{\e} \sqrt{ 1 + \widehat{\xi^2}}
\exp\left(- \frac{\widehat{\xi^2}y_i^2 }{2(1+\widehat{\xi^2})} \right) }
\ee
with $\widehat{\xi^2} = (1/\e)(\overline{y^2}-1)_+$. When the magnitude of $\theta_i$ is large ($\gg 1$), $y_i$ is likely to have  large magnitude as well; hence, the amount of shrinkage due to the denominator is smaller. On the other hand, for $\theta_i$ with smaller magnitude, $y_i$ is also likelier to have smaller magnitude and the amount of shrinkage  is correspondingly larger. 

To further understand the role of the shrinkage factor in the numerator, suppose that an oracle provided us with the values  $\{ \theta_i^2 \},  i\in[n]$. Then, the ideal linear minimax estimator is \cite{candes}
\begin{align*}
 \hat{\theta}_i = \frac{\theta_i^2} {(1 + \theta_i^2)} y_i, ~~~~i \in [n].
\end{align*}
 Noting that $\norm{\bst}^2/(n\e)$ is the mean of the $\{ \theta_i^2 \}$ for $\theta_i \neq 0$, in the absence of the oracle, the estimator attempts to approximate the term ${\theta_i^2}/{(1 + \theta_i^2)}$ via the ratio $\frac{\norm{\bst}^2/(n\e)}{1 + \norm{\bst}^2/(n\e)}$.  This ratio in  turn is well-approximated for large $n$ by $\widehat{\xi^2}/(1+\widehat{\xi^2})$ ---  this  can be seen from \eqref{eq:xi_estimate} by observing that $\overline{y^2} = \norm{\by}^2/n$ is close to its mean $\norm{\bst}^2/n + 1$ (when $\bar{y} = 0$). This is the significance of the common shrinkage factor in the numerator.  The denominator further shrinks the estimate if it believes that the $\theta_i$ has a small magnitude. 

To summarize,  in \eqref{eq:eBayes_zero_location}, the $y_i$ corresponding to the large non-zero components of $\bst$ are shrunk by approximately $\frac{\norm{\bst}^2/(n\e)}{1+ \norm{\bst}^2/(n\e)}$, while those corresponding to the zero components of $\bst$ are made even closer to $0$.



\section{Risk Estimates and the Hybrid Estimator}\label{sec_risk_estimator}

Recall from Fig. \ref{fig_12} that  depending on the underlying $\bst$, either $\bsth_{ST}$ or $\bsth_{EB}$ may have smaller loss. To construct a hybrid estimator that reliably chooses the better estimator, we use Stein's unbiased risk estimate (SURE) \cite{stein2} to estimate the losses of each estimator.

\begin{fact}
  \cite{stein2} If an estimator $\bsth(\by)$ is almost everywhere differentiable, then
 \begin{equation*}
  \hat{R}(\bst,\bsth(\by)) \vcentcolon = -n + \norm{\by - \bsth}^2 + 2 \sum_{i=1}^n \frac{\partial \hat{\theta}_i}{\partial y_i}
 \end{equation*}
is an unbiased estimate of the risk $R(\bst,\bsth)$, i.e., $\ex \left[\hat{R}(\bst,\bsth(\by))\right] = R(\bst,\bsth)$ where the expectation is again over $\by \sim \mathcal{N}(\bst,\mathbf{I})$. $\hat{R}(\bst,\bsth(\by))$ is called the SURE of $\bsth$. 
\end{fact}

Using SURE, the normalized risk estimate for $\bsth_{ST}$ with threshold $\lambda$ is given by
\begin{equation}\label{sure_st}
  \hat{R}(\bst,\bsth_{ST};\lambda) = -1 + \frac{\norm{\by - \bsth_{ST}}^2}{n} + \frac{2}{n} \sum_{i=1}^n \mathsf{1}_{\{ y_i^2 > \lambda^2 \}}. 
\end{equation}

 To keep the  exposition simple, for our concentration results we assume that the location parameter $\hat{\mu}$ in $\bsth_{EB}$ is zero, so that  $\bsth_{EB}$ is given by \eqref{eq:eBayes_zero_location}. Extending the results to the case with a general $\hat{\mu}$ is straightforward, though a bit cumbersome. Let  
 \begin{equation}
 \label{eq:aydy_def}
 \begin{split}
 a_\by \vcentcolon &=  \frac{\wh{\xi}^2}{1 + \wh{\xi}^2} =\left[ 1 - \frac{\e}{(\norm{\by}^2/n -1)_+ + \e}\right],\\
 d_\by \vcentcolon &=  1 + \wh{\xi}^2 = 1+  \frac{1}{\e} \left(\frac{\norm{\mathbf{y}}^2}{n} - 1\right)_+, \\
 c_\by \vcentcolon & = \frac{1-\e}{\e} \sqrt{1 + \wh{\xi}^2} =  \frac{1-\e}{\e}\sqrt{d_\by}, \\
 b_i(\by) \vcentcolon&= 1 + c_\by e^{-\frac{a_\by y_i^2}{2}}. 
 \end{split}
\end{equation} 

\noindent Using SURE, the normalized risk estimate  for $\bsth_{EB}$ with $\hat{\mu}=0$ is  
 \begin{align*}
&   \frac{\hat{R}(\bst,\hat{\bst}_{EB}(\by); \e)}{n} 
 = -1 + \frac{1}{n}\norm{\mathbf{y} - \bsth_{EB}}^2 + \frac{2}{n} \sum_{i=1}^n \frac{\partial\hat{\theta}_i}{\partial y_i} \\
 &  = -1 + \frac{\norm{\by}^2}{n} + \frac{a^2_\by}{n}\sum_{i=1}^n\frac{y_i^2}{b^2_i(\by)} - \frac{2a_\by}{n}\sum_{i=1}^n\frac{y_i^2}{b_i(\by)} \\ 
 &\quad + \frac{2}{n}  \sum_{i=1}^n\left[\frac{a_\by}{b_i(\by)} + \frac{a_\by'(i) y_i}{b_i(\by)} + \frac{\left[ a_\by'(i)  c_\by  (y_i^2/2) + a_\by c_\by y_i - c_\by'(i)  \right]a_\by y_ie^{-\frac{a_\by y_i^2}{2}}}{b_i^2(\by)}  \right] \\
 & =  \left(\frac{\norm{\by}^2}{n} - 1 \right)+ \frac{a^2_\by}{n}\sum_{i=1}^n\frac{y_i^2}{b^2_i(\by)} - \frac{2a_\by}{n}\sum_{i=1}^n\frac{y_i^2}{b_i(\by)}\\
 &\quad + \frac{2}{n}  \sum_{i=1}^n\left[\frac{a_\by}{b_i(\by)} + \frac{2 y_i^2}{n\e d_\by^2 b_i(\by)}\mathsf{1}_{\{\norm{\by}^2 > n\}} + \left(\frac{(1-\e)a_\by}{n \e^2 d_\by^{3/2}}\right)\frac{  y_i^4  e^{-\frac{a_\by y_i^2}{2}}}{b_i^2(\by)}  \right] \\
 & \quad +\frac{2}{n}  \sum_{i=1}^n\left[ \frac{  a_\by^2 c_\by y_i^2  e^{-\frac{a_\by y_i^2}{2}}}{b_i^2(\by)} - \left( \frac{(1-\e)a_\by}{\e^{2} \sqrt{d_\by}} \right)\frac{  y_i^2  e^{-\frac{a_\by y_i^2}{2}}}{n b_i^2(\by)}  \right] 
 \end{align*}
where $a_\by'(i) = \frac{\partial a_\by}{\partial y_i}$, $c_\by'(i) =  \frac{\partial c_\by}{\partial y_i}$. Rearranging terms and simplifying, we obtain
\begin{equation}
\label{sure_eb}
\begin{split}
 &\frac{\hat{R}(\bst,\hat{\bst}_{EB}(\by); \e)}{n} =  \left(\frac{\norm{\by}^2}{n} - 1 \right)+ \frac{a^2_\by}{n}\sum_{i=1}^n\frac{y_i^2\big(1+  2c_\by  e^{-\frac{a_\by y_i^2}{2}}\big)}{b^2_i(\by)}- \frac{2a_\by}{n}\sum_{i=1}^n\frac{y_i^2-1}{b_i(\by)}\\ 
  &  +\frac{4 }{d_\by^2 \e n^2}  \sum_{i=1}^n\frac{ y_i^2}{b_i(\by)}  \mathsf{1}_{\{\norm{\by}^2 > n\}} 
   +   \frac{2(1-\e)a_\by}{ d_\by^{3/2}  \e^2 n^2}  \sum_{i=1}^n \frac{y_i^4  e^{-\frac{a_\by y_i^2}{2}}}{b_i^2(\by)}     
  -   \frac{2(1-\e)a_\by}{ \sqrt{d_\by} \e^{2} n^2} \sum_{i=1}^n \frac{y_i^2  e^{-\frac{a_\by y_i^2}{2}}}{b_i^2(\by)} .  
\end{split}
\end{equation}

For large $n$, the last three terms of \eqref{sure_eb} with $n^2$ in the denominator are very small and  can be neglected in a practical application of the risk estimate. More precisely, the proof of Theorem \ref{thm1} in the next section shows that the last three terms  concentrate around deterministic constants of order $\frac{1}{n}$.

We use the risk estimates in  \eqref{sure_st} and \eqref{eq:aydy_def} to define a hybrid estimator that aims to select the estimator with smaller loss for the $\boldsymbol{\theta}$ in context.  The hybrid estimator is defined as
\begin{align}\label{comb_estimator}
  \hat{\bst}_{H} & = 
\gamma_{\mbf{y}}  \hat{\bst}_{EB} + (1-\gamma_{\mbf{y}}) 
 \hat{\bst}_{ST}, \\
\label{eq_gamma}
\gamma_{\mbf{y}} & = \left\{ \begin{array}{ccc}
                    1 & \textrm{if } & \hat{R}(\bst,\hat{\bst}_{EB}(\by)) \leq 
                    \hat{R}(\bst,\hat{\bst}_{ST}(\by)), \\
                    0 & \textrm{otherwise.}&\\
                  \end{array} \right.
\end{align}

In the next section, we obtain concentration results for the risk estimates and loss functions of $\bsth_{ST}$ and $\bsth_{EB}$, and use these to show that the loss of the hybrid estimator concentrates on the minimum of the losses of the two estimators.


\section{Main Results} \label{sec_loss_conc}

\subsection{Concentration Results for the Empirical Bayes Estimator} \label{subsec:eb_conc}

The constants in our concentration  results for the eBayes estimator depend on $\bst$ via $\frac{1}{n}\sum_{i=1}^n \theta_i^4$. In order to make these constants universal, we assume that the fourth moment of $\bst$ is bounded.

\textbf{Assumption A}: There exists a finite constant $\Lambda > 0$ such that $\frac{1}{n}\sum_{i=1}^n \theta_i^4 \leq \Lambda$.

When Assumption A is satisfied, the constants in the concentration results depend only on $\Lambda$ (and not on the underlying $\bst$ or $n)$.  For brevity, we henceforth do not explicitly indicate the dependence on $\lambda$ and $\e$ in the notation for the risk estimates on the LHS of \eqref{sure_st} and \eqref{sure_eb}, respectively.

\begin{thm}\label{thm1}
Consider a sequence of $\bst$ with increasing dimension $n$ and satisfying Assumption A. Then the risk estimate  $\hat{R}(\bst,\hat{\bst}_{EB}(\by))$ satisfies the following for any $t > 0$:
\begin{equation}\label{eq_conc_thm1}
 \mathbb{P}\left( \frac{1}{n}\left \vert \hat{R}(\bst,\hat{\bst}_{EB}(\by)) - R_1(\bst,\hat{\bst}_{EB}) \right \vert \geq t \right) \leq Ke^{-nk\min(t,t^2)}
\end{equation}
where $0 < K \leq 24$ and $k>0$ are absolute constants, and $R_1(\bst,\hat{\bst}_{EB})$ is a deterministic quantity such that  
\begin{align}\label{eq_convergence_thm1}
\left \vert \frac{R_1(\bst,\hat{\bst}_{EB})}{n} - \frac{{R}(\bst,\hat{\bst}_{EB})}{n} \right \vert = \mathcal{O}\left(\frac{1}{\sqrt{n}} \right).
\end{align}
\end{thm}
\begin{proof}
The $i^{th}$ element of $\bsth_{EB}$ in \eqref{eq:eBayes_zero_location} is $\hat{\theta}_i = \frac{a_\by y_i}{b_i(\by)}$. The SURE of $\bsth_{EB}$ is as given in \eqref{sure_eb}.  We need to show concentration for each term on the RHS of \eqref{sure_eb}. In the following, $K, k, k_{0}, \ldots, k_{10}$ are universal positive constants that do not depend on $t$ or $n$.
 
Since  $\norm{\mbf{y}}^2$ is a non-central chi-squared random variable with mean $\norm{\bst}^2+n$,   we have the following large deviations bound \cite{birge}. For any $t > 0$ 
\begin{align}\label{eq4_thm1}
 \prb\left( \abs{\frac{\norm{\by}^2}{n} -1  - \frac{\norm{\bst}^2}{n}  } \geq t\right) \leq 2e^{-nk_0\min(t,t^2)}.
\end{align}
The concentration for the remaining terms of \eqref{sure_eb} is shown using two lemmas stated below. The proofs of the lemmas are given  Sec. \ref{subsec_proof_2}. The first lemma shows that the last three terms in \eqref{sure_eb} concentrate around their expectations.
\begin{lemma}
\label{lem:un}
Let
\begin{align*}
 u_n \vcentcolon  = \frac{4}{\e d_\by^2n^2} \sum_{i=1}^n \frac{y_i^2}{b_i(\by)} \mathsf{1}_{\{\norm{\by}^2 > n\}}, 
   v_n \vcentcolon =  \frac{2(1-\e )a_\by}{d_\by^{3/2}n^2\e^2} \sum_{i=1}^n\frac{  y_i^4  e^{-\frac{a_\by y_i^2}{2}}}{b_i^2(\by)}  , 
  x_n  \vcentcolon=  \frac{2(1-\e)a_\by}{n^2\e^2\sqrt{d_\by}}  \sum_{i=1}^n \frac{  y_i^2  e^{-\frac{a_\by y_i^2}{2}}}{  b_i^2(\by)}.
\end{align*}
Then for any $t > 0$,
\begin{equation} 
\label{eq5a_thm1}
\begin{split}
 \prb \left( \vert u_n - \ex u_n \vert \geq t \right) & \leq 2 e^{-n^2k_1t^2},  \qquad \prb \left( \vert v_n  -\ex v_n \vert \geq t \right)  \leq 2 e^{-n^2k_2t^2}, \\
 \prb \left( \vert x_n  - \ex x_n \vert \geq t \right) &  \leq 2 e^{-n^2k_3t^2}.
 \end{split}
\end{equation}
\end{lemma}

Establishing concentration inequalities for the second and third terms of \eqref{sure_eb} around their respective means is more challenging. This is because the summands are dependent random variables and it is not straightforward to prove that their sum satisfies Lipschitz or similar conditions for which Gaussian concentration results are readily available.  Hence,  in the following lemma, we  prove concentration of these terms around certain deterministic values, and then show that these deterministic values are close to the required means.
\begin{lemma}
\label{lem:fn}
Let
\begin{equation*}
\begin{split}
f_n \vcentcolon = & \frac{a^2_\by}{n}\sum_{i=1}^n\frac{y_i^2}{b^2_i(\by)} - \frac{a^2}{n}\sum_{i=1}^n\ex \left[\frac{y_i^2}{\left(1 + ce^{-ay_i^2/2}\right)^2}\right], \\
 g_n \vcentcolon  =&  \frac{2a_\by}{n}\sum_{i=1}^n\frac{y_i^2} {b_i(\by)} - \frac{2a}{n}\sum_{i=1}^n\ex \left[\frac{y_i^2}{1 + c e^{-ay_i^2/2}}\right], \\
  h_n \vcentcolon  =& \frac{2a_\by}{n}  \sum_{i=1}^n\frac{1}{b_i(\by)} -  \frac{2a}{n}  \sum_{i=1}^n\ex \left[\frac{1}{1 + c e^{-ay_i^2/2}}\right],\\  
 w_n \vcentcolon= & \frac{2a_\by^2 c_\by}{n}  \sum_{i=1}^n \frac{   y_i^2  e^{-\frac{a_\by y_i^2}{2}}}{b_i^2(\by)}  - \frac{2a^2 c}{n}  \sum_{i=1}^n \ex \left[  \frac{y_i^2  e^{-\frac{a y_i^2}{2}}}{\left(1 + c e^{-ay_i^2/2}\right)^2}\right],
\end{split}
\end{equation*}
where
\begin{equation}\label{eq:bc_defs}
 a \vcentcolon= \frac{\norm{\bst}^2/n}{\e + \norm{\bst}^2/n}, \quad c \vcentcolon = \frac{1- \e}{\e^{3/2}}\sqrt{\e  + \norm{\bst}^2/n }.
\end{equation}
Then, for any $t > 0$,
\begin{align}\label{eq6d_thm1}
\prb \left( \left \vert f_n \right \vert \geq t \right) & \leq 4e^{-nk_4\min(t,t^2)},\\ \label{eq6e_thm1}
 \prb \left( \vert g_n  \vert \geq t \right ) &\leq 4e^{-nk_5\min(t,t^2)},\\ \label{eq6f_thm1}
\prb \left( \vert h_n  \vert \geq t \right ) &\leq 4e^{-nk_6\min(t,t^2)},\\ \label{eq6g_thm1}
\prb \left( \vert w_n  \vert \geq t \right ) &\leq 4e^{-nk_7\min(t,t^2)}.
\end{align}
\end{lemma}
Using the results of Lemmas \ref{lem:un} and \ref{lem:fn}, we obtain, for any $t > 0$, 
\begin{align*}
 \mathbb{P}\left( \left \vert \frac{\hat{R}(\bst,\hat{\bst}_{EB}(\by))}{n} - \frac{R_1(\bst,\hat{\bst}_{EB})}{n} \right \vert \geq t \right) \leq 24e^{-nk\min(t,t^2)}
\end{align*}
where $k$ is an absolute positive constant and 
\begin{align*}
&  \frac{R_1(\bst,\hat{\bst}_{EB})}{n} \\
 & = b  + \frac{a^2}{n}\sum_{i=1}^n\ex \left[\frac{y_i^2}{\left(1 + ce^{-ay_i^2/2}\right)^2}\right]  - \frac{2a}{n}\sum_{i=1}^n\ex \left[\frac{y_i^2 - 1}{1 + c e^{-ay_i^2/2}}\right]  + \frac{2a^2 c}{n}  \sum_{i=1}^n \ex \left[\frac{  y_i^2  e^{-\frac{a y_i^2}{2}}}{ \left(1 + c e^{-ay_i^2/2}\right)^2}\right]  \\
& + \frac{4}{\e n^2}   \sum_{i=1}^n\ex \left[\frac{y_i^2}{d_\by^2b_i(\by)} \mathsf{1}_{\{\norm{\by}^2 > n\}}\right] 
 +    \frac{2(1-\e)}{n^2 \e^2}\sum_{i=1}^n\ex \left[\frac{a_\by y_i^4  e^{-\frac{a_\by y_i^2}{2}}}{ d_\by^{3/2}b_i^2(\by)} \right]  - \frac{2(1-\e)}{n^2\e^2 }  \sum_{i=1}^n \ex \left[\frac{ a_\by y_i^2  e^{-\frac{a_\by y_i^2}{2}}}{ \sqrt{d_\by} b_i^2(\by)}\right],
\end{align*}

\noindent with the constants $a,c$ as defined in \eqref{eq:bc_defs}.

Finally, to prove \eqref{eq_convergence_thm1}, we use Lemma \ref{lem3} to get 
\begin{align}\label{eq_risk_bound}
  \ex \left\vert  \frac{\hat{R}(\bst,\hat{\bst}_{EB}(\by))}{n} - \frac{R_1(\bst,\hat{\bst}_{EB})}{n} \right \vert \leq \frac{C}{\sqrt{n}}\left(1+ \frac{1}{\sqrt{n}}\right)
\end{align}
for some positive constant $C$. Since $\ex \vert X \vert  \geq \vert \ex X \vert$ and $\ex \left[ \hat{R}(\bst,\hat{\bst}_{EB}(\by))\right] ={R}(\bst,\hat{\bst}_{EB})$, \eqref{eq_convergence_thm1} follows.

\end{proof}

\begin{note}
 Theorem \ref{thm1} implies that $\frac{1}{n}\left [ \hat{R}(\bst,\hat{\bst}_{EB}(\by)) - R(\bst,\hat{\bst}_{EB}) \right ] = \mathcal{O}_P\left(\frac{1}{\sqrt{n}} \right)$. This is a slightly stronger result than \cite[Theorem 2.2]{ZB17} which states that 
$\frac{1}{n}\left [ \hat{R}(\bst,\hat{\bst}_{EB}(\by)) - R(\bst,\hat{\bst}_{EB}) \right ] = \mathcal{O}_P\left(\frac{(\log n)^{3/2}}{\sqrt{n}} \right)$. The result in \cite[Theorem 2.2]{ZB17} applies to a class of densities  $\psi$ in \eqref{eq:eb_prior} that are unimodal with $\log \psi$ satisfying certain Lipschitz conditions. Though this class is more general than the Gaussian, the result is derived assuming that the parameters defining $\psi$ are fixed and do not depend on the data. (In particular, the parameters can take on any fixed value in a specified range which grows logarithmically with $n$.) In contrast, our parameter estimates $\hat{\mu}(\mbf{y})$ and $\widehat{\xi^2}(\mbf{y})$ depend on the data. Obtaining concentration results for terms with these data-dependent parameters (e.g., those in Lemma \ref{lem:fn}) is the key technical challenge in proving Theorem \ref{thm1}.
\label{note:ZBcomparison}
 \end{note}

The next result shows that the normalized loss of the eBayes estimator concentrates on a deterministic value close to the true risk.

\begin{thm}\label{thm1a}
Consider a sequence of $\bst$ with increasing dimension $n$ and satisfying Assumption A. Then the loss function 
$L(\bst,\hat{\bst}_{EB}) = \norm{\bst - \hat{\bst}_{EB}}^2$ satisfies the following for any $t > 0$:
\begin{equation}\label{eq_conc_thm1a}
 \mathbb{P}\left( \frac{1}{n}\left \vert L(\bst,\hat{\bst}_{EB}(\by)) - R_2(\bst,\hat{\bst}_{EB}) \right \vert \geq t \right) \leq Ke^{-nk\min(t,t^2)}
\end{equation}
where $K \leq 10 $ and $k$ are absolute positive constants, and $R_2(\bst,\hat{\bst}_{EB})$ is a deterministic quantity such that  
\begin{align}\label{eq_convergence_thm1a}
\left \vert \frac{R_2(\bst,\hat{\bst}_{EB})}{n} - \frac{{R}(\bst,\hat{\bst}_{EB})}{n} \right \vert = \mathcal{O}\left(\frac{1}{\sqrt{n}} \right).
\end{align}
\end{thm}
\begin{proof}
We have
\begin{align}
 \frac{L(\bst,\hat{\bst}_{EB}(\by))}{n}  = \frac{\norm{\bst - \hat{\bst}_{EB}}^2}{n} = \frac{\norm{\bst}^2}{n} + \frac{\norm{\hat{\bst}_{EB}(\by)}^2}{n} - \frac{2a_\by}{n}\sum_{i=1}^n \frac{\theta_i y_i}{b_i(\by)}.
 \label{eq:ebayes_lf}
\end{align}
We have already shown in   \eqref{eq6d_thm1} that 
\[ \frac{\norm{\hat{\bst}_{EB}(\by)}^2}{n} =  \frac{a^2_\by}{n}\sum_{i=1}^n\frac{y_i^2}{b^2_i(\by)}   \]  
concentrates around $\frac{a^2}{n}\sum_{i=1}^n\ex \left[ y_i^2/{(1 + ce^{-ay_i^2/2})^2}\right]$.  The concentration for the last term in \eqref{eq:ebayes_lf} around its mean is complicated to prove due to the absence of any Lipschitz behaviour. We instead show in Sec. \ref{sec_proof_thm1a} that for any $t > 0$,
\begin{align}\label{eq_conc_loss_thm1a}
 \prb\left( \frac{1}{n}\left\vert\sum_{i=1}^n \frac{ \theta_i a_\by y_i}{b_i(\by)} - \sum_{i=1}^n \ex \left[\frac{ a\theta_i  y_i}{1 + ce^{-ay_i^2/2}}\right] \right \vert \geq t\right) \leq 6e^{-nk\min(t,t^2)}.
\end{align}
Thus, using the concentration inequalities in \eqref{eq6d_thm1}  and \eqref{eq_conc_loss_thm1a}, from \eqref{eq:ebayes_lf} we obtain that for any $t > 0$, 
\begin{align*}
 \mathbb{P}\left( \left \vert \frac{L(\bst,\hat{\bst}_{EB}(\by))}{n} - \frac{R_2(\bst,\hat{\bst}_{EB})}{n} \right \vert \geq t \right) \leq 10e^{-nk\min(t,t^2)}
\end{align*}
where
\begin{align}
\label{eq:R2def}
& \frac{R_2(\bst,\hat{\bst}_{EB})}{n}= \frac{\norm{\bst}^2}{n}  + \frac{a^2}{n}\sum_{i=1}^n\ex \left[\frac{y_i^2}{\left(1 + ce^{-ay_i^2/2}\right)^2}\right]  - \frac{2a}{n}\sum_{i=1}^n\ex \left[\frac{ \theta_i y_i}{1 + ce^{-ay_i^2/2}}\right]  
\end{align}
with the constants $a,c$ as defined in \eqref{eq:bc_defs}. We note that due to Assumption A,  the RHS of \eqref{eq:R2def} is bounded by a universal constant not depending on $n$.

To prove \eqref{eq_convergence_thm1a}, we apply Lemma \ref{lem3} which shows that the concentration result \eqref{eq_conc_thm1a} implies the following bound on the expected value:
\begin{align}\label{eq_risk_bound_thm1a}
  \ex \left\vert  \frac{L(\bst,\hat{\bst}_{EB}(\by))}{n} - \frac{R_2(\bst,\hat{\bst}_{EB})}{n} \right \vert \leq \frac{C}{\sqrt{n}}\left(1+ \frac{1}{\sqrt{n}}\right)
\end{align}
where  $C$ is a universal positive constant. Since $\ex \vert X \vert  \geq \vert \ex X \vert$ and $\ex \left[ L(\bst,\hat{\bst}_{EB}(\by))\right] ={R}(\bst,\hat{\bst}_{EB})$, \eqref{eq_convergence_thm1a} follows.

\end{proof}

\subsection{Concentration Results for the Soft-Thresholding Estimator} \label{subsec:st_conc}

The concentration result for the risk estimate of soft-thresholding was obtained by Donoho and Johnstone \cite{donoho_johnstoneWav}. In contrast to the eBayes estimator, the normalized risk estimate for soft-thresholding given in \eqref{sure_st} is bounded. Therefore a concentration result can be directly obtained using Hoeffding's inequality \cite{boucheron2}. 

\begin{thm}\label{fact1}\cite{donoho_johnstoneWav}
 The risk estimate $\hat{R}(\bst,\hat{\bst}_{ST})$ for the soft-thresholding estimator satisfies the following for any $t > 0$,
\begin{align}
\label{eq_RST_conc}
 \mathbb{P}\left( \frac{1}{n}\left \vert \hat{R}(\bst,\hat{\bst}_{ST}(\by)) - {R}(\bst,\hat{\bst}_{ST}) \right \vert \geq t \right) \leq 2e^{-\frac{2t^2}{9(1+\lambda^2)^2}}.
\end{align}
\end{thm}

We can also show that the normalized loss of the soft-thresholding estimator concentrates on the true risk.
\begin{thm}\label{thm1b}
The loss function $L(\bst,\hat{\bst}_{ST})  = \norm{\bst - \hat{\bst}_{ST}}^2$ of the soft-thresholding estimator satisfies the following for any $t > 0$:
\begin{equation}\label{eq_conc_thm1b}
 \mathbb{P}\left( \frac{1}{n}\left \vert L(\bst,\hat{\bst}_{ST}(\by)) - {R}(\bst,\hat{\bst}_{ST}) \right \vert \geq t \right) \leq 2e^{-nk\min(t,t^2)}
\end{equation}
where $k$ is an absolute positive constant.
\end{thm}
\begin{proof}
See Sec. \ref{sec_proof_thm1b}.
\end{proof}

\subsection{Concentration and Risk Bound for the Hybrid Estimator} \label{subsec:hyb_conc}

For a given  $\bst$, let
\begin{align}
&L_{min}( \boldsymbol{\theta},\by) \vcentcolon=  \min\left\{L( \boldsymbol{\theta}, \hat{\bst}_{EB}(\by)), L( \boldsymbol{\theta}, \hat{\bst}_{ST}(\by))\right\}, \nonumber \\
&L_{max}( \boldsymbol{\theta},\by) \vcentcolon=  \max\left\{L( \boldsymbol{\theta}, \hat{\bst}_{EB}(\by)), L( \boldsymbol{\theta}, \hat{\bst}_{ST}(\by))\right\}, \nonumber \\
&L_{sep}( \boldsymbol{\theta}, \by) \vcentcolon= L_{max}( \boldsymbol{\theta},\by) - L_{min}( \boldsymbol{\theta},\by), \label{eq:lsep_def} \\
& \kappa_n \vcentcolon= \frac{\abs{R_1(\bst,\hat{\bst}_{EB}) - R_2(\bst,\hat{\bst}_{EB})}}{n}, \label{eq:kappa_n}
\end{align}
where $R_1(\bst,\hat{\bst}_{EB})$ and $R_2(\bst,\hat{\bst}_{EB})$ are the deterministic concentrating values in Theorems \ref{thm1} and \ref{thm1a}, respectively.
Note that $ \kappa_n $ is an $\mathcal{O}(1/\sqrt{n})$ quantity since both $R_1(\bst,\hat{\bst}_{EB})/n$ and $R_2(\bst,\hat{\bst}_{EB})/n$ are within $\mathcal{O}(1/\sqrt{n})$ from ${R}(\bst,\hat{\bst}_{EB})/n$. The following theorem characterizes the loss $L(\bst,\hat{\bst}_H(\by))$ and the risk $R(\bst,\hat{\bst}_H)$ of the hybrid estimator.
\begin{thm}\label{thm2}
Consider a sequence of $\bst$ with increasing dimension $n$ and satisfying Assumption A. Then, for any $t >0$, we have
\begin{align} \label{eq:hybrid_loss}
&\prb\left( \frac{1}{n} L(\bst,\hat{\bst}_H(\by))  \geq \frac{1}{n}L_{min}( \boldsymbol{\theta},\by) + t + \kappa_n \right) \leq K e^{-nk\min(t,t^2)}, 
\end{align}
for some absolute positive constants $K$ and $k$. The risk of the hybrid estimator can be bounded as
\begin{align}
   \frac{R( \boldsymbol{\theta}, \hat{\bst}_{H})}{n} & \leq  \frac{\ex \left[L_{min}( \boldsymbol{\theta},\by)\right]}{n} +  \mathcal{O}\left(\frac{1}{\sqrt{n}}\right) \\
 & \leq \frac{1}{n} \min \left \{ R( \boldsymbol{\theta}, \hat{\bst}_{EB}), \, R( \boldsymbol{\theta}, \hat{\bst}_{ST})  \right \}  +  \mathcal{O}\left(\frac{1}{\sqrt{n}}\right). \label{eq:hybrid_risk}
\end{align}
 \end{thm}
\begin{proof} See Sec. \ref{sec_hybrid_proof}. \end{proof}



\section{Simulation Results}\label{sec_simulation}

\begin{figure*}[t]
    \centering
    \begin{subfloat}[\label{fig3}]{
       \includegraphics[trim= 0.9in 0 1.22in 0, clip=true,width= 2.6in,height=2.2in]{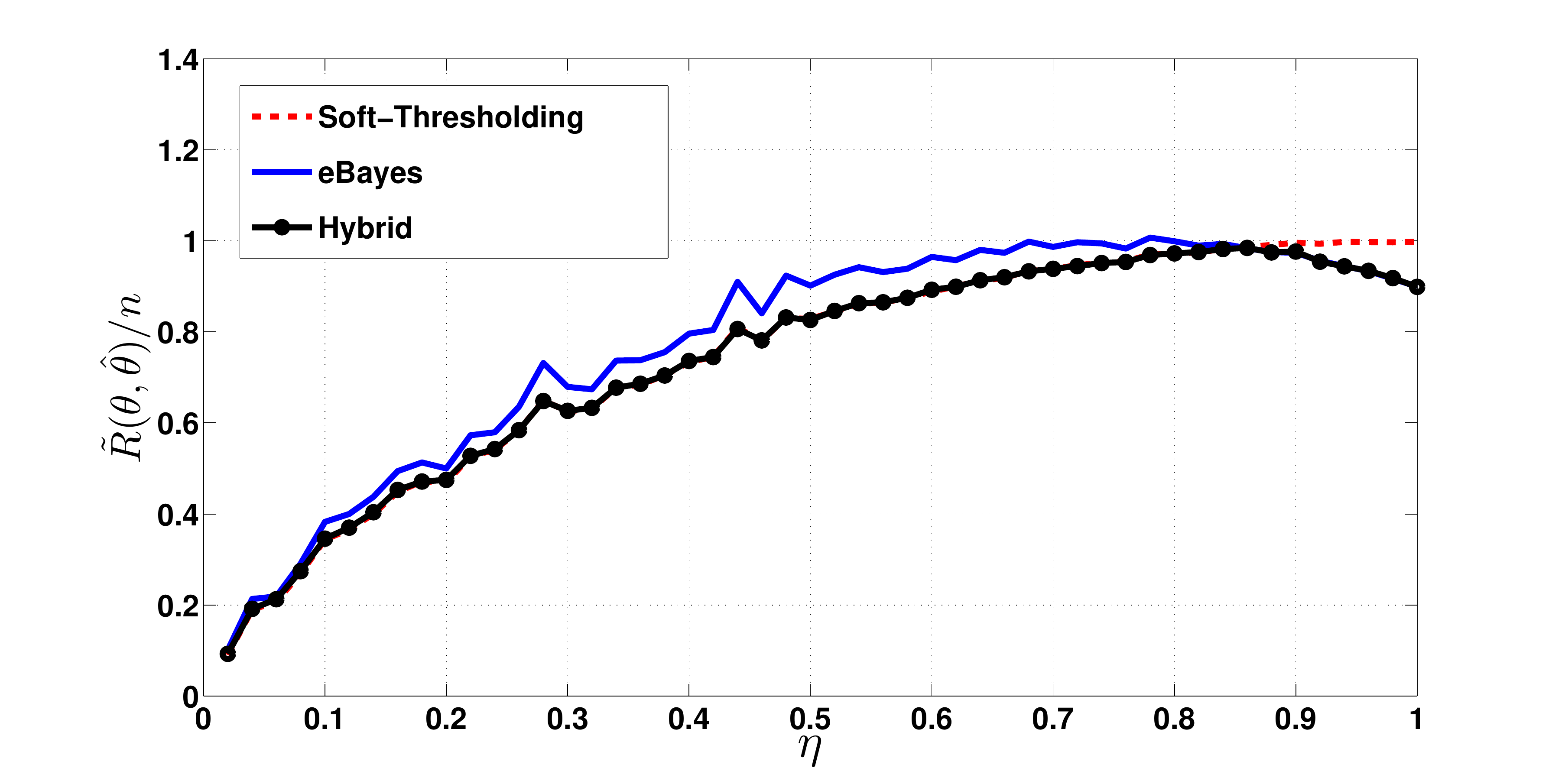}   
       }
    \end{subfloat}
  \quad
    \begin{subfloat}[\label{fig4}]{
       \includegraphics[trim= 0.9in 0 1.22in 0, clip=true,width= 2.6in,height=2.2in]{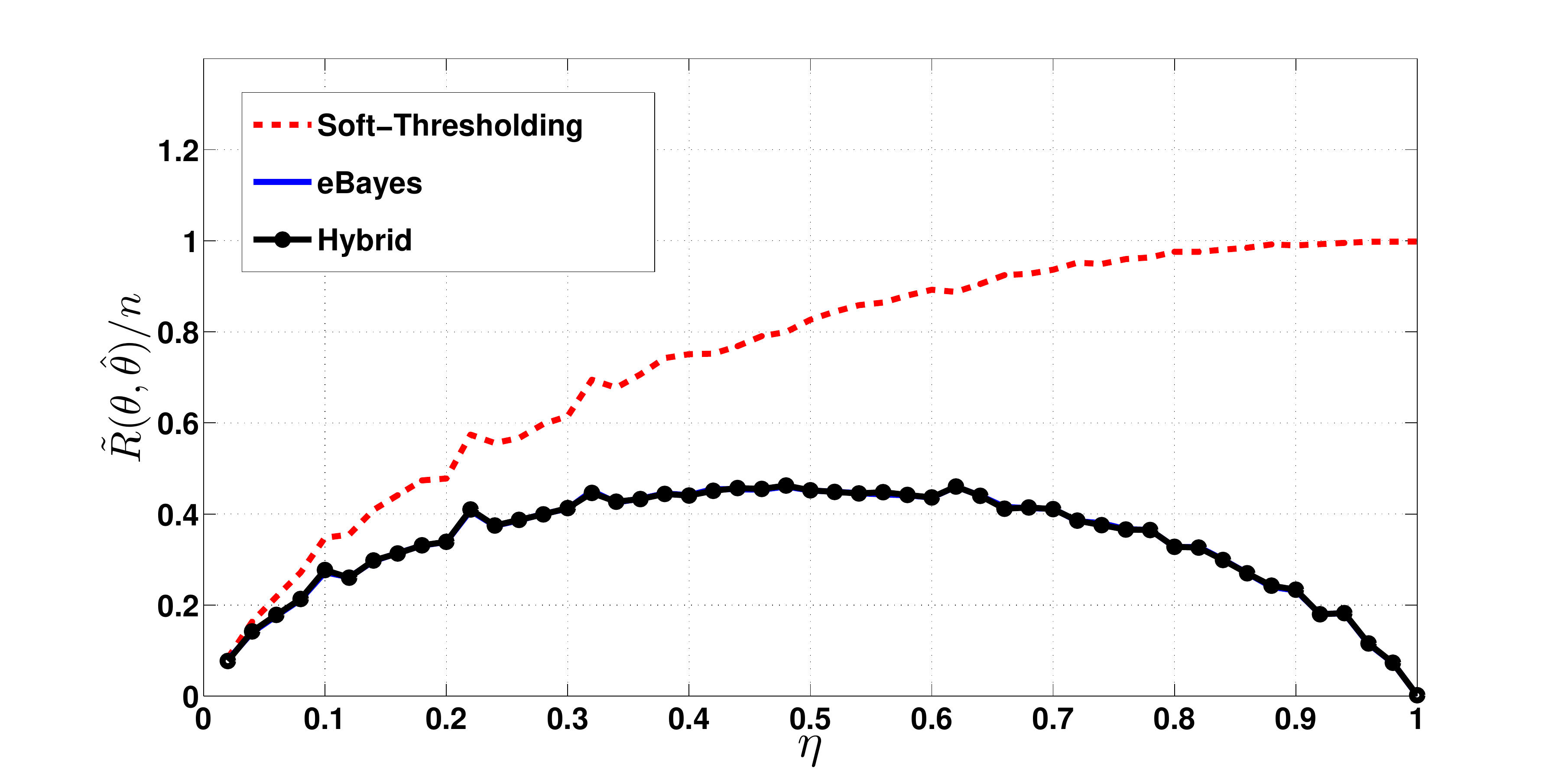}
       }
    \end{subfloat}    
    \caption{\small  Average normalized loss $\tilde{R}(\bst,\bsth)/n$ with $n=1000$ for the following cases: a) Half the non-zero entries in $\bst$ equal $3$ and the other half $-3$. b) All the non-zero entries in $\bst$ equal $3$. }
    \label{fig_34}
\end{figure*}

The performance of the hybrid estimator for the two kinds of $\bst$ considered in Fig. \ref{fig_12} is highlighted in Fig. \ref{fig_34}.  Clearly, $n =1000$ is large enough for the hybrid estimator to accurately pick the better of $\bsth_{ST}$ and $\bsth_{EB}$. In both Figs. \ref{fig_12} and \ref{fig_34}, $\e$ was chosen equal to the true sparsity level $\eta$ for both estimators.

When the true sparsity level $\eta$ is unknown, one can optimize SURE to find the best choice of $\e$ for both $\bsth_{ST}$ and $\bsth_{EB}$. The concentration results (Theorem \ref{fact1} and Theorem \ref{thm1}) imply that the SURE for either  estimator does not deviate much from the actual risk for large $n$. Donoho and Johnstone \cite{donoho_johnstoneWav} have proposed SureShrink which chooses the thresholding parameter $\lambda^*$ from the interval $(0,\sqrt{2\log n}]$ as follows. The interval  $(0,\sqrt{2\log n}]$ is discretized to define a discrete set $\mathcal{S}$. Then
\begin{equation}\label{eq_opt_th}
 \lambda^* = \argmin_{\lambda \in \mathcal{S}} {\hat{R}(\bst,\bsth_{ST}(\by);\lambda)}/ n
\end{equation}
where $ \hat{R}(\bst,\bsth_{ST}(\by);\lambda)$ is as defined in \eqref{sure_st}.

\begin{figure*}[!htp]
    \centering
    \begin{subfloat}[\label{fig5}]{
       \includegraphics[trim= 0.9in 0 1.22in 0, clip=true,width=2.6in,height=2.2in]{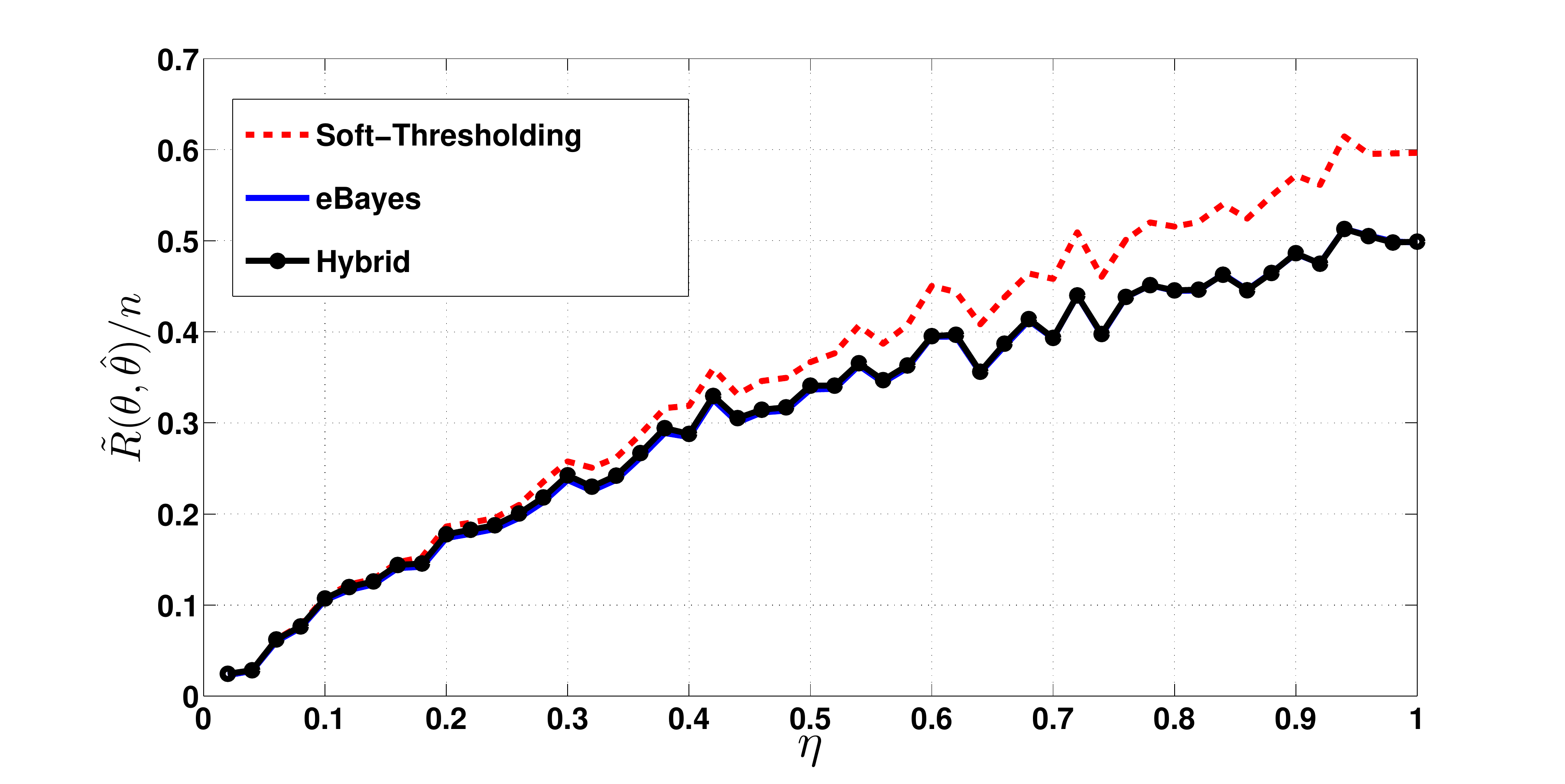}   
       }
    \end{subfloat}
  \quad
    \begin{subfloat}[\label{fig6}]{
       \includegraphics[trim= 0.9in 0 1.22in 0, clip=true,width=2.6in,height=2.2in]{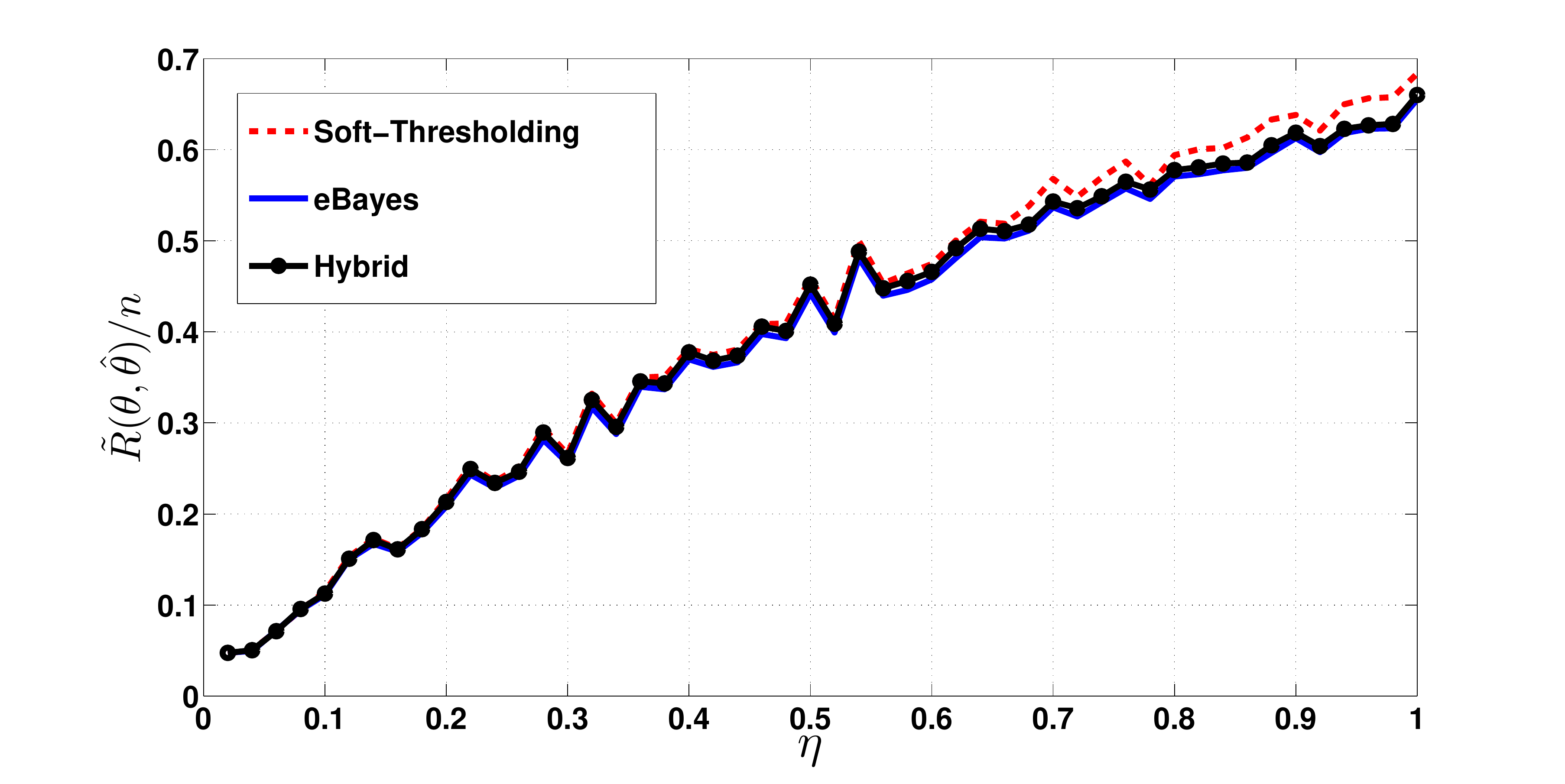}
       }
    \end{subfloat} 
    \quad
    \begin{subfloat}[\label{fig7}]{
       \includegraphics[trim= 0.9in 0 1.22in 0, clip=true,width=2.6in,height=2.2in]{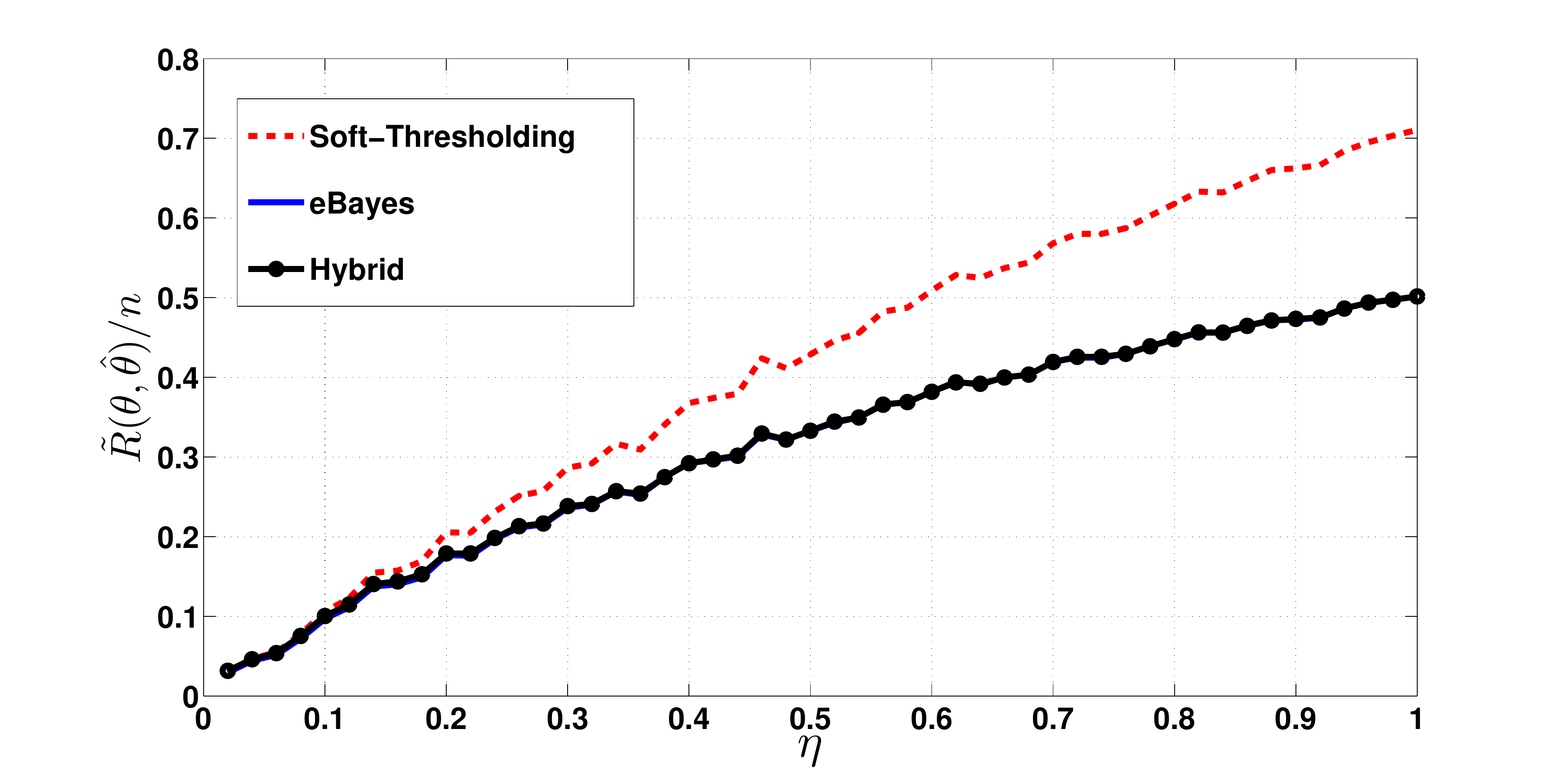}
       }
    \end{subfloat} 
    \quad
    \begin{subfloat}[\label{fig8}]{
       \includegraphics[trim= 0.7in 0 1.22in 0, clip=true,width=2.6in,height=2.2in]{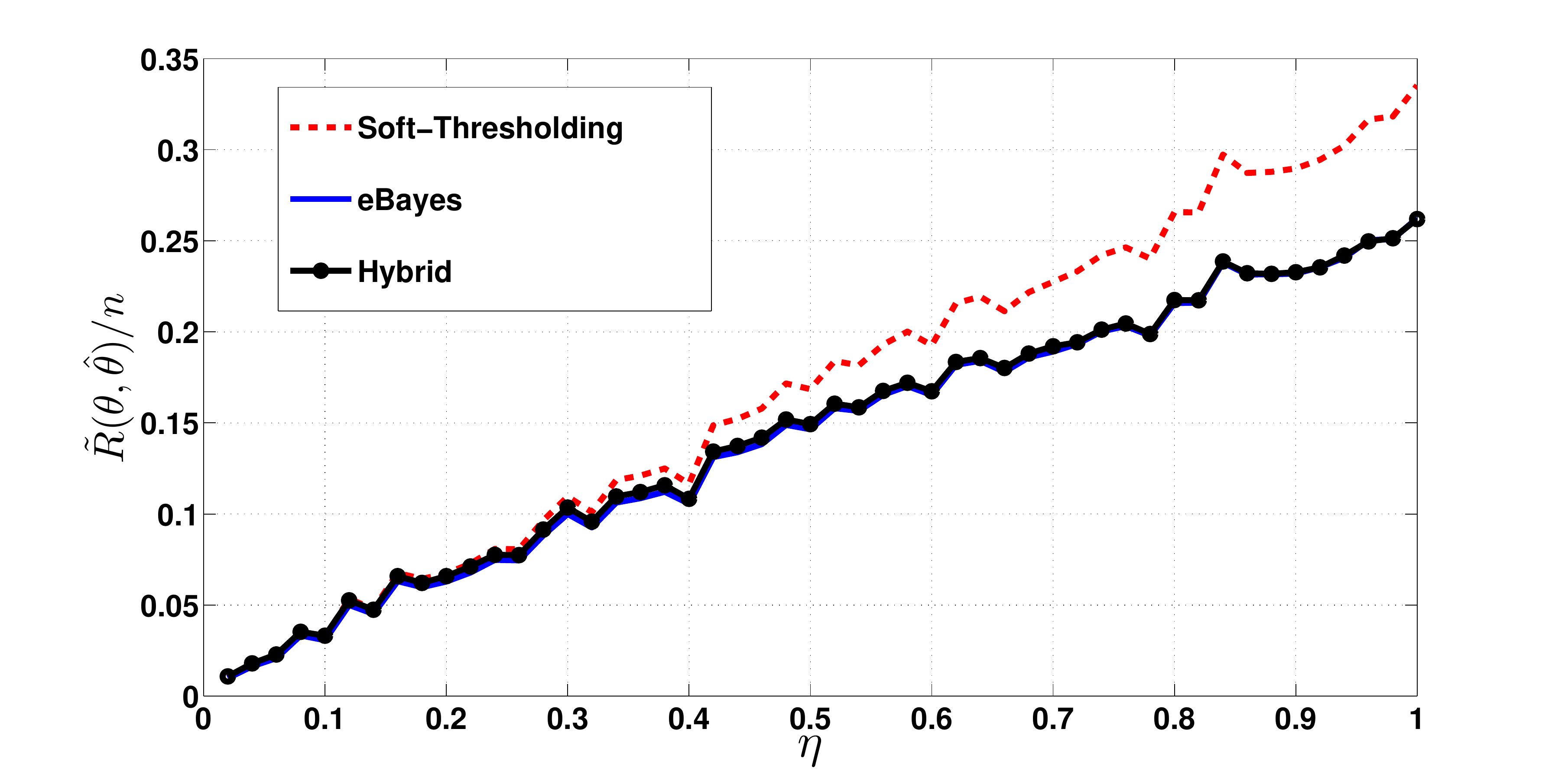}
       }
    \end{subfloat}   
    \caption{\small Average normalized loss $\tilde{R}(\bst,\bsth)/n$ with $n=1000$ for the following cases: a) The non-zero entries of $\bst$ are drawn from $\mathcal{N}(0,1)$. b) The non-zero entries are drawn from the Laplace distribution with mean $0$ and variance $2$. c) The non-zero entries are drawn from the Rademacher (equiprobable $\pm 1$) distribution. d) The non-zero entries are drawn uniformly from $[-1,1]$.}
    \label{fig_5678}
\end{figure*}

\begin{figure*}[!htp]
    \centering
    \begin{subfloat}[\label{fig5_1}]{
       \includegraphics[trim= 0.9in 0 1.22in 0, clip=true,width=2.6in,height=2.2in]{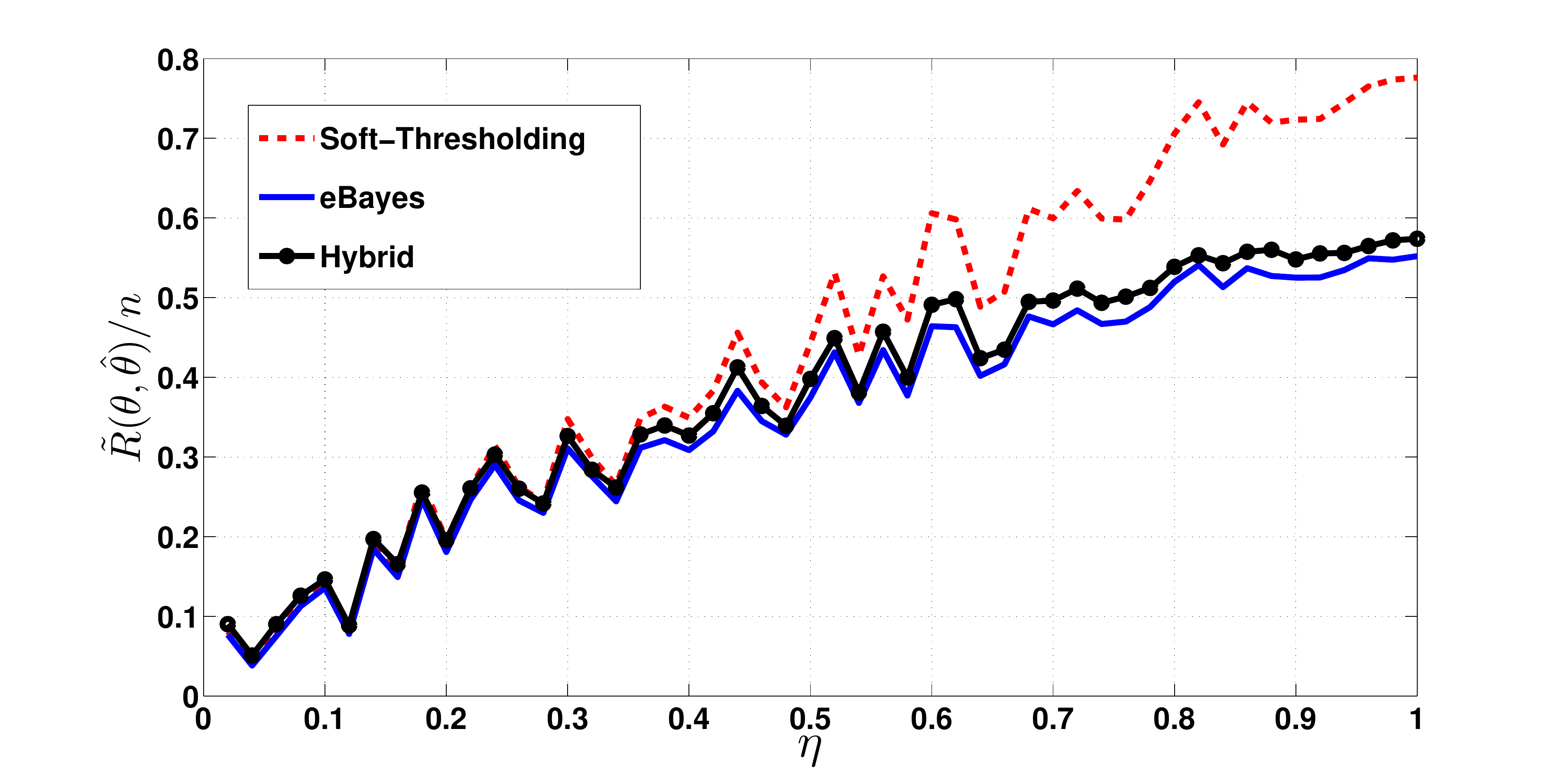}   
       }
    \end{subfloat}
  \quad
    \begin{subfloat}[\label{fig6_1}]{
       \includegraphics[trim= 0.9in 0 1.22in 0, clip=true,width=2.6in,height=2.2in]{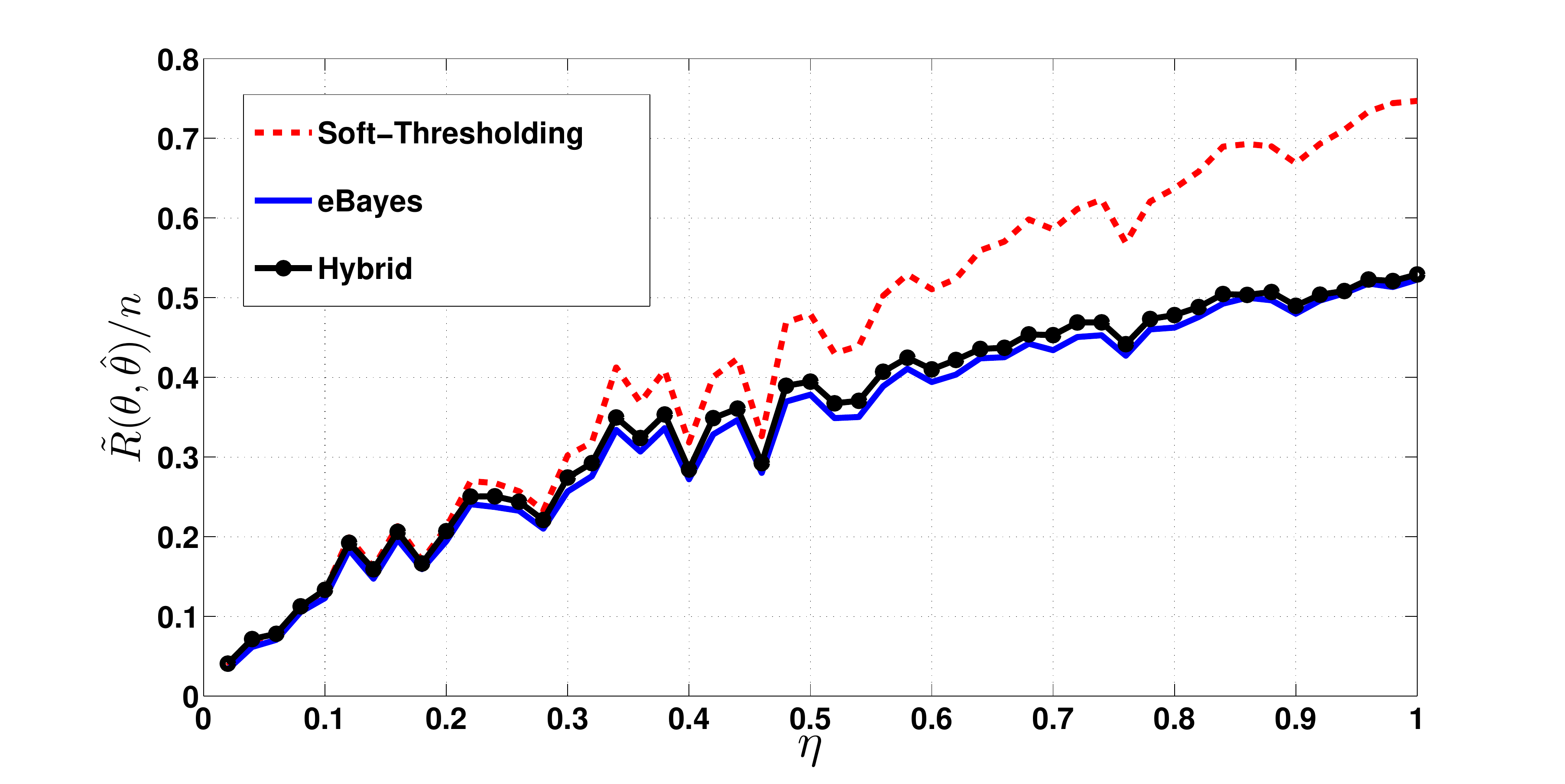}
       }
    \end{subfloat} 
    \quad
    \begin{subfloat}[\label{fig7_1}]{
       \includegraphics[trim= 0.9in 0 1.22in 0, clip=true,width=2.6in,height=2.2in]{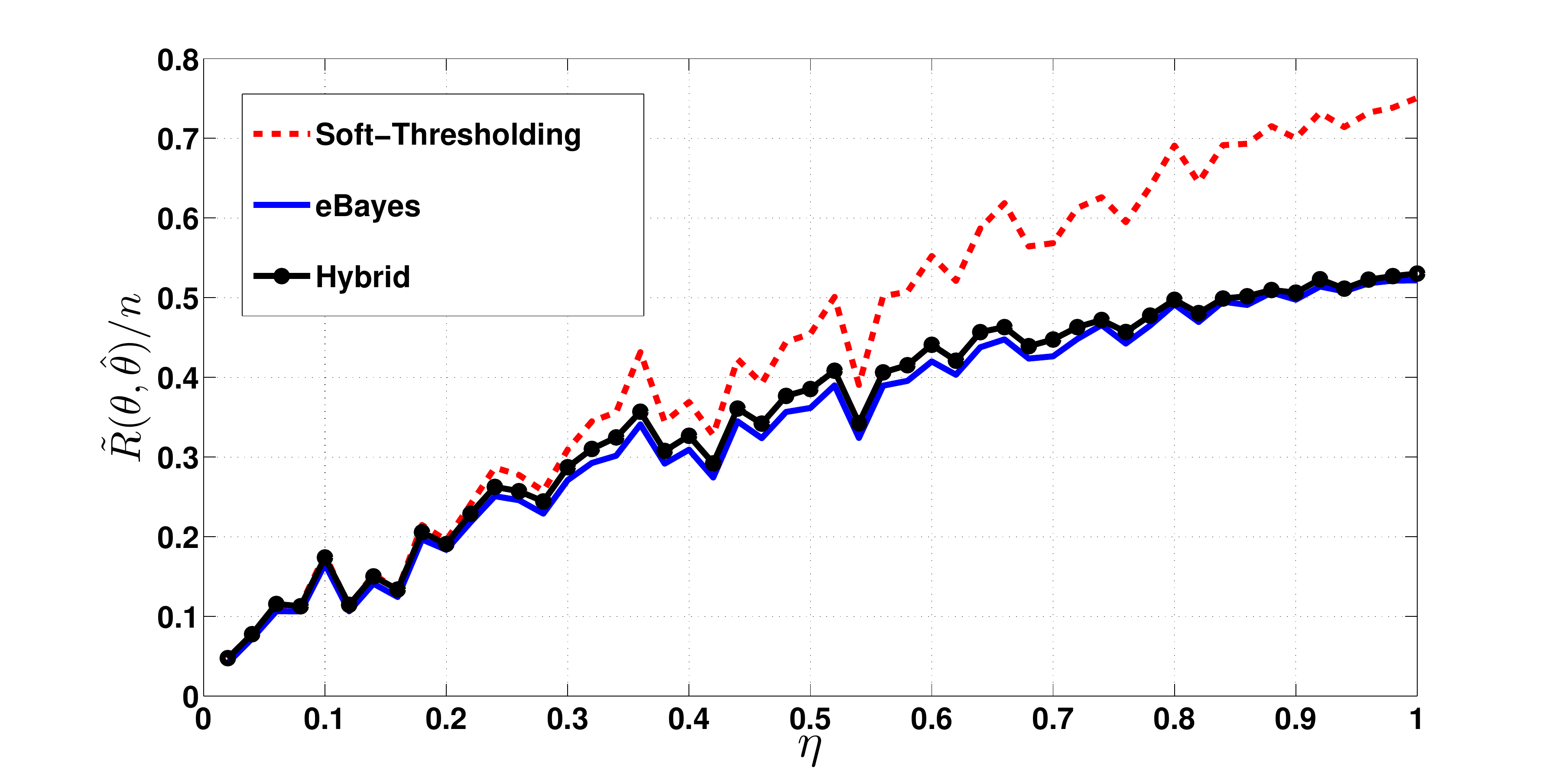}
       }
    \end{subfloat} 
    \quad
    \begin{subfloat}[\label{fig8_1}]{
       \includegraphics[trim= 0.9in 0 1.22in 0, clip=true,width=2.6in,height=2.2in]{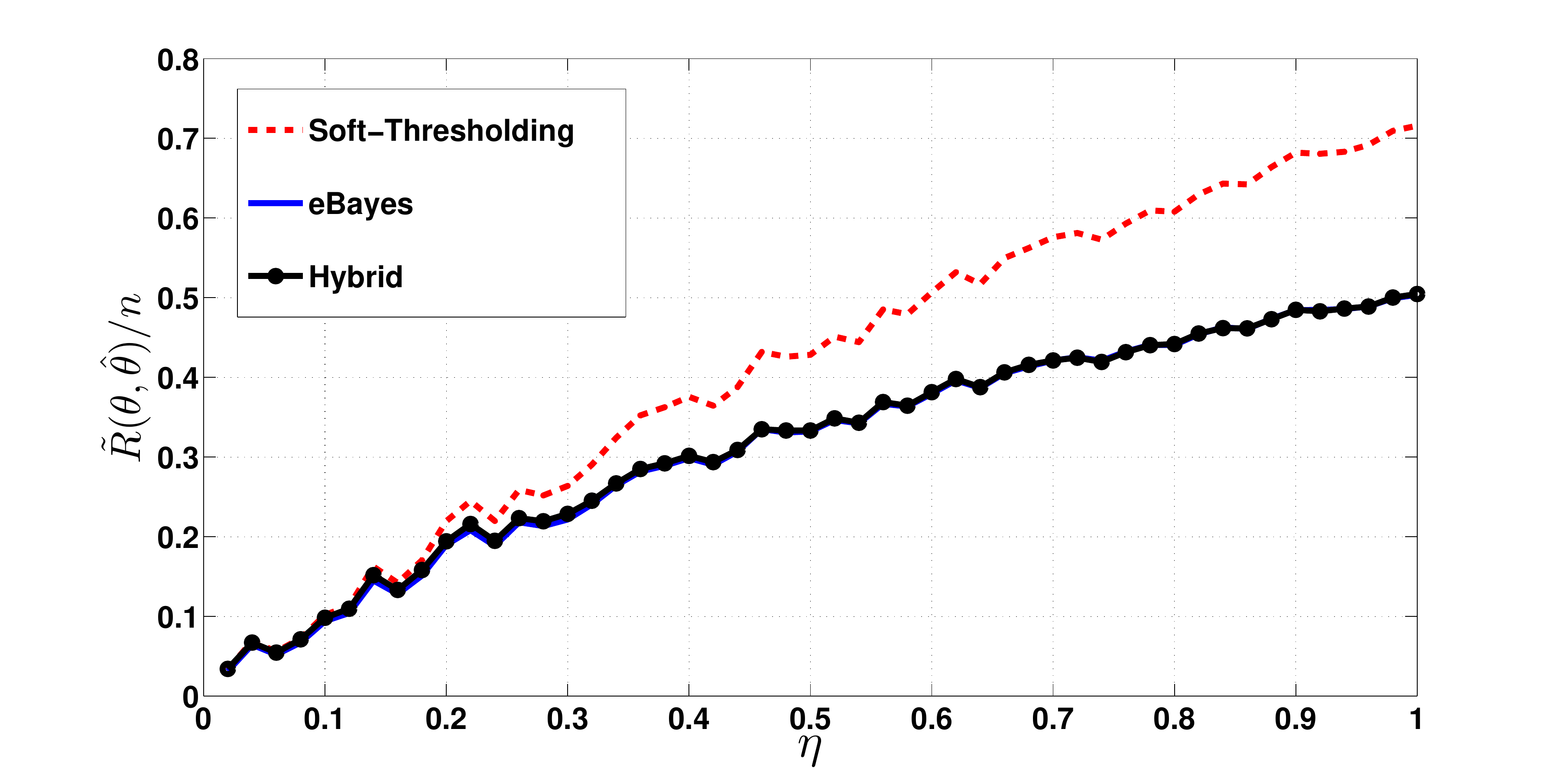}
       }
    \end{subfloat}   
    \caption{\small Average normalized loss $\tilde{R}(\bst,\bsth)/n$ with the non-zero entries drawn from the Rademacher distribution for the following cases: a) $n = 50$ b) $n = 100$ c) $n = 200$ d) $n = 500$.}
    \label{fig_5678_2}
\end{figure*}

For the eBayes estimator, we propose to find the best value of $\e$ in \eqref{eq:eBayes} by first discretizing the interval $(0,1]$ to define a discrete set $\mathcal{D}$, and choosing the sparsity parameter as 
\begin{equation}\label{eq_opt_eps}
 \e^* = \argmin_{\e \in \mathcal{D}} {\hat{R}(\bst,\bsth_{EB}(\by);\e)}/ n.
\end{equation}
Here $\hat{R}(\bst,\bsth_{EB}(\by);\e)/n$ is as  in \eqref{sure_eb}, with suitable modifications to account for non-zero $\hat{\mu}$. The hybrid estimator then chooses the estimator with the lower value of SURE, i.e., by comparing
$\hat{R}(\bst,\bsth_{ST}(\by);\lambda^*)$ versus  $\hat{R}(\bst,\bsth_{EB}(\by);\e^*)$.

Fig. \ref{fig_5678} shows the performance of the hybrid estimator at different sparsity levels for four choices for the distribution of the non-zero entries of $\bst$: Gaussian  (Fig. \ref{fig5}),  Laplacian (Fig. \ref{fig6}), Rademacher (equiprobable $\pm 1$) (Fig. \ref{fig7}), and uniform  (Fig. \ref{fig8}). 
We assume that the actual sparsity factor $\eta$ is unknown and use SURE to find the best sparsity parameters $\lambda^*$ and $\e^*$ for $\bsth_{ST}$ and $\bsth_{EB}$, respectively. The optimization is performed over the discrete sets   $\mathcal{S} = \{ 0.1i, i \in [\lceil 10\sqrt{2\log n}\rceil] \}$ and $\mathcal{D} = \{ 0.02i, i \in [ 50 ] \}$. In all the plots, $n = 1000$.  The plots suggest that for a wide range of $\bst$, $\bsth_{EB}$ is at least as good as $\bsth_{ST}$  for all values of the sparsity factor $\eta$, and better in most cases. 

Fig. \ref{fig_5678_2} illustrates the performance of the hybrid estimator as a function of $n$.  It shows the average normalized losses of the three estimators for different values of $n$ when the non-zero entries of $\bst$ take values from the Rademacher  distribution. These plots indicate that the proposed hybrid estimator performs very well even for relatively small values of $n$.  

\subsection{Application to Compressed Sensing}\label{subsec_AMP}
Given a measurement matrix $\mathbf{A} \in \mathbb{R}^{m \times n}$, the  goal  in compressed sensing \cite{CandesTaoLP, donohoCS, candesRT06}  is to estimate a sparse vector $\bst \in \mathbb{R}^n$ from a  noisy linear measurement $\by \in \mathbb{R}^m$. In particular,  consider the measurement model 
\begin{align*}
 \by = \mathbf{A}\bst + \bw,
\end{align*}
where  $\mathbf{A}$ is an  $m \times n$ random matrix with i.i.d. sub-Gaussian entries (normalized so that its columns have Euclidean norm concentrated around $1$),  and the noise vector  $\bw \sim \mathcal{N}(\mathsf{0}, \sigma^2 \mathbf{I})$.  The undersampling ratio is denoted by $\delta \vcentcolon = m/n <1$. 

For this linear model, Approximate message passing (AMP) \cite{DonMalMont09, BayMont11, montanari_graphical_models, Rangan11,BayatiLM15} is a class of low-complexity iterative algorithms to estimate  $\bst$ from $\by$. Starting with the initial conditions $\bst_0 = \mathsf{0}$, $\mathbf{z}_0 = \by$, 
 AMP iteratively produces estimates $\{\bst_t \}$, for $t \geq 1$  as follows \cite{BayMont11}:
\begin{align}  
 \bst_{t} &=  f_t\left(\mathbf{A}^T \mathbf{z}_{t-1} + \bst_{t-1}\right) \label{eq_amp_algo_theta} \\
   \mathbf{z}_t &= \by - \mathbf{A}\bst_t + \frac{1}{\delta}\mathbf{z}_{t-1}\left\langle f_{t}^\prime\left( \mathbf{A}^T \mathbf{z}_{t-1} + \bst_{t-1} \right)\right\rangle.  \label{eq_amp_algo_z}
\end{align}
Here for each $t$, $f_t : \mathbb{R} \to \mathbb{R}$ is a ``denoising" function, and   $f^\prime_t$ denotes its derivative. For a vector input $\mathbf{u} \in \mathbb{R}^n$, both $f_t$ and $f^{\prime}_t$  operate component-wise on $\mbf{u}$. Further, for  $\mathbf{u} \in \mathbb{R}^n$, $\langle \mathbf{u} \rangle \vcentcolon= \frac{1}{n}\sum_{t=1}^n u_t$ denotes the average of its entries. 

The AMP update \eqref{eq_amp_algo_theta} is underpinned by the following key property of the effective observation vector $(\mbf{A}^T\mathbf{z}_t + \bst^t)$: for large $n$, after each iteration $t$, $(\mbf{A}^T \mathbf{z}_t + \bst^t)$ is approximately distributed as $\bst + \tau_t \mbf{Z}$, where $\mbf{Z} \in \mathbb{R}^n$ is an i.i.d.\ $\mc{N}(0,1)$ random vector independent of $\bst$. The effective noise variance $\tau_t^2$ is determined (in the large system limit)  by a  scalar recursion called state evolution \cite{BayMont11}, \cite{montanari_graphical_models}. For our purposes, it suffices to note that for each $t$, a good estimate of $\tau_t^2$ is given by $\widehat{\tau}_t^2 \vcentcolon = \frac{\norm{\mathbf{z}_t}^2}{m}$ (see, for example,  \cite[pp. 14,21]{montanari_graphical_models}, also \cite{bayMontLASSO,cindy_ramji}). 

Thus, the function $f_t$ estimates the sparse vector $\bst$ from an observation in Gaussian noise of variance approximately $\widehat{\tau}_{t-1}^2=\frac{\norm{z^{t-1}}^2}{m}$.  Therefore, in each iteration, the AMP  provides a platform to compare the performance of soft-thresholding and the eBayes estimator (and hence the hybrid estimator) as choices for $f_t$. We note that while soft-thresholding operates on a vector component-wise, the eBayes estimator doesn't. However, for sufficiently large values of $m$ and $n$, both $\hat{\mu}$ and $\widehat{\xi^2}$ in \eqref{eq:mean_estimate}-\eqref{eq:eBayes} are close to deterministic values in which case the eBayes estimator also approximately acts component-wise on a vector.  We remark that if we use soft-thresholding with the threshold in each iteration tuned to the noise-level  $\tau_t$,  the fixed points of the AMP algorithm coincide with that of the LASSO \cite{montanari_graphical_models,bayMontLASSO}.

\begin{figure*}[!htp]
    \centering
    \begin{subfloat}[\label{fig9}]{
       \includegraphics[trim= 0.6in 0 1.22in 0, clip=true,width=2.5in,height=2.2in]{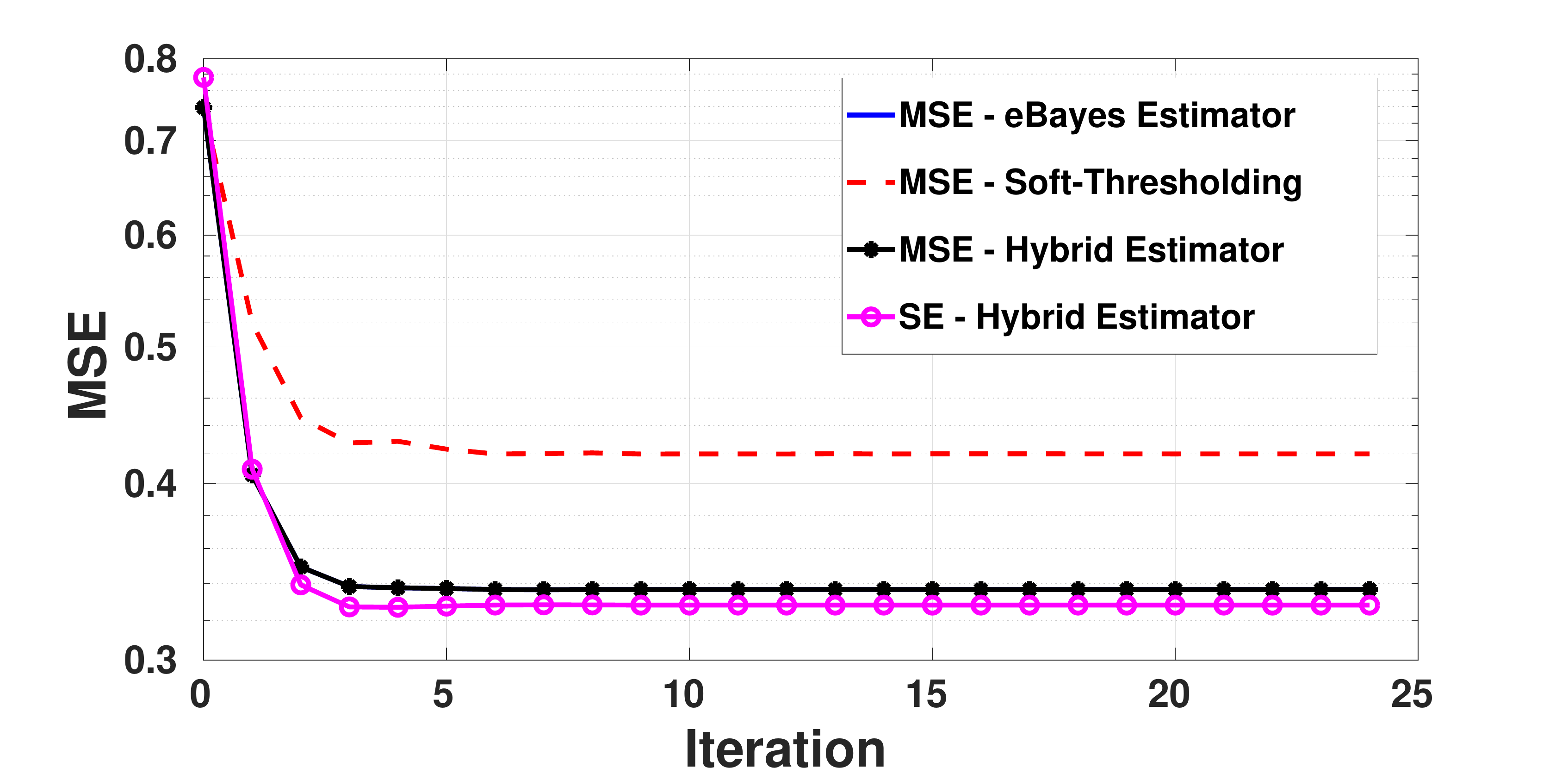}   
       }
    \end{subfloat}
  \quad
    \begin{subfloat}[\label{fig10}]{
       \includegraphics[trim= 0.6in 0 1.22in 0, clip=true,width=2.5in,height=2.2in]{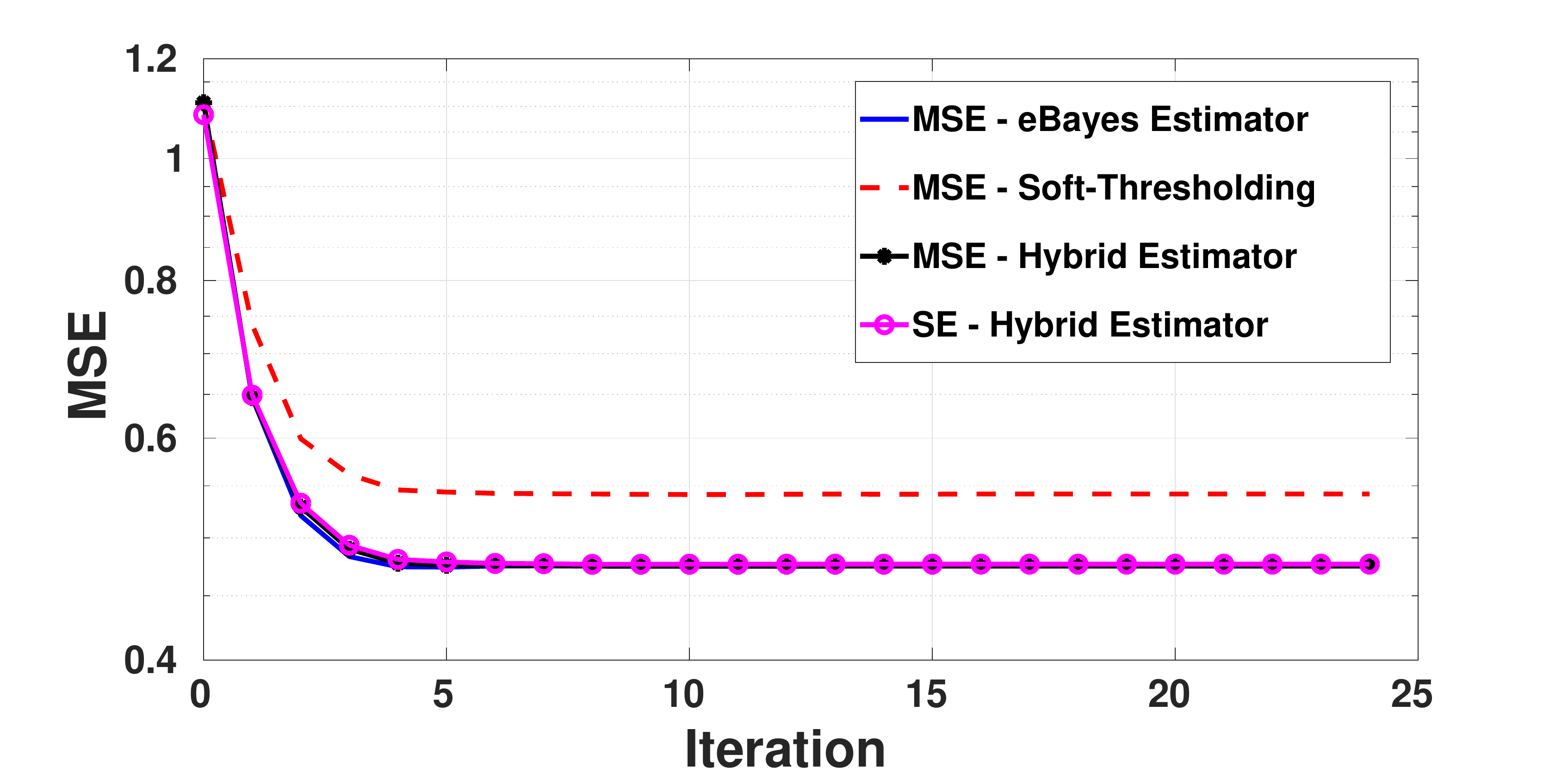}
       }
    \end{subfloat} 
    \quad
    \begin{subfloat}[\label{fig11}]{
       \includegraphics[trim= 0.21in 0 1.22in 0, clip=true,width=2.5in,height=2.2in]{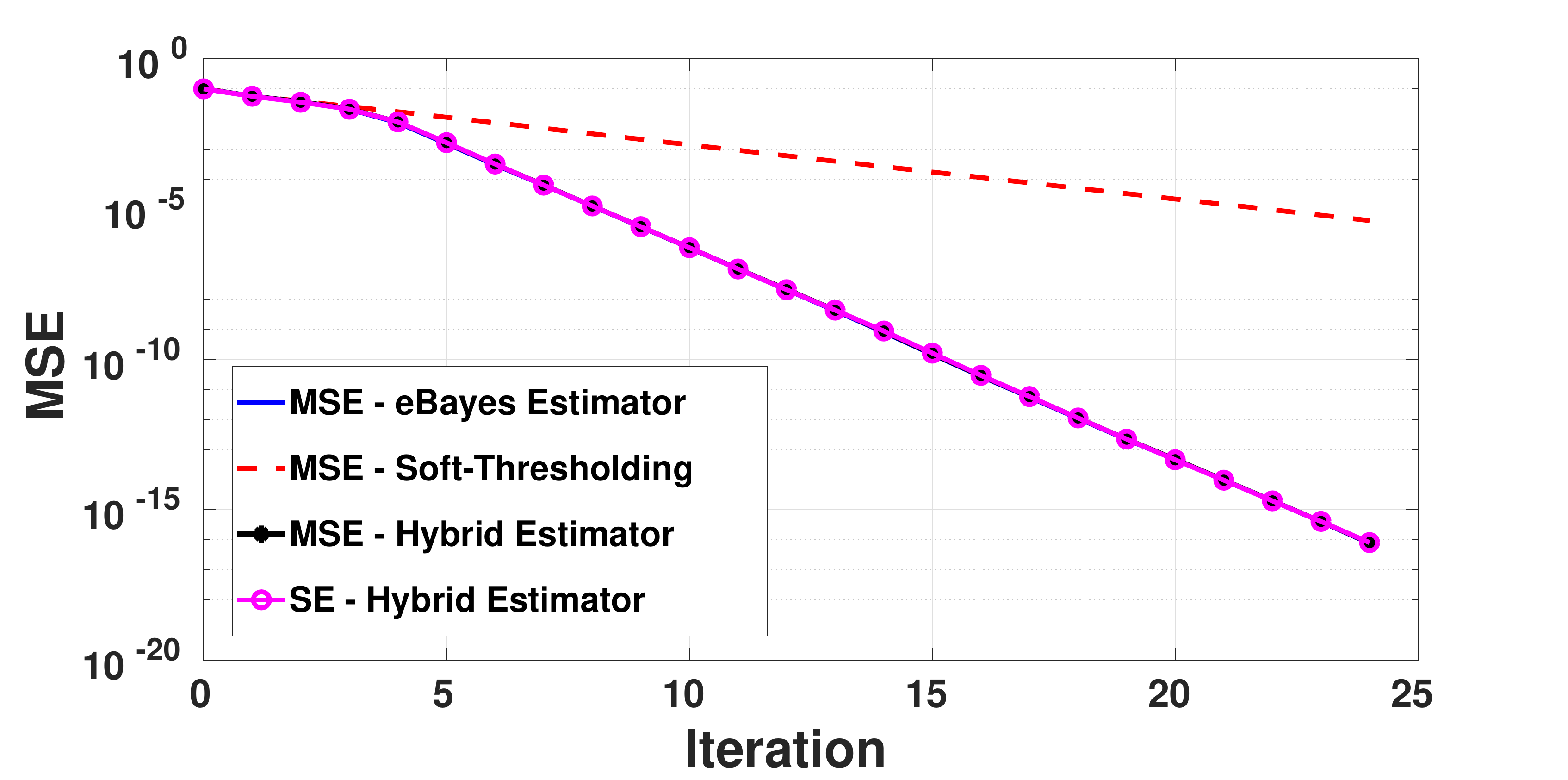}
       }
    \end{subfloat} 
    \quad
    \begin{subfloat}[\label{fig12}]{
       \includegraphics[trim= 0.4in 0 1.22in 0, clip=true,width=2.5in,height=2.2in]{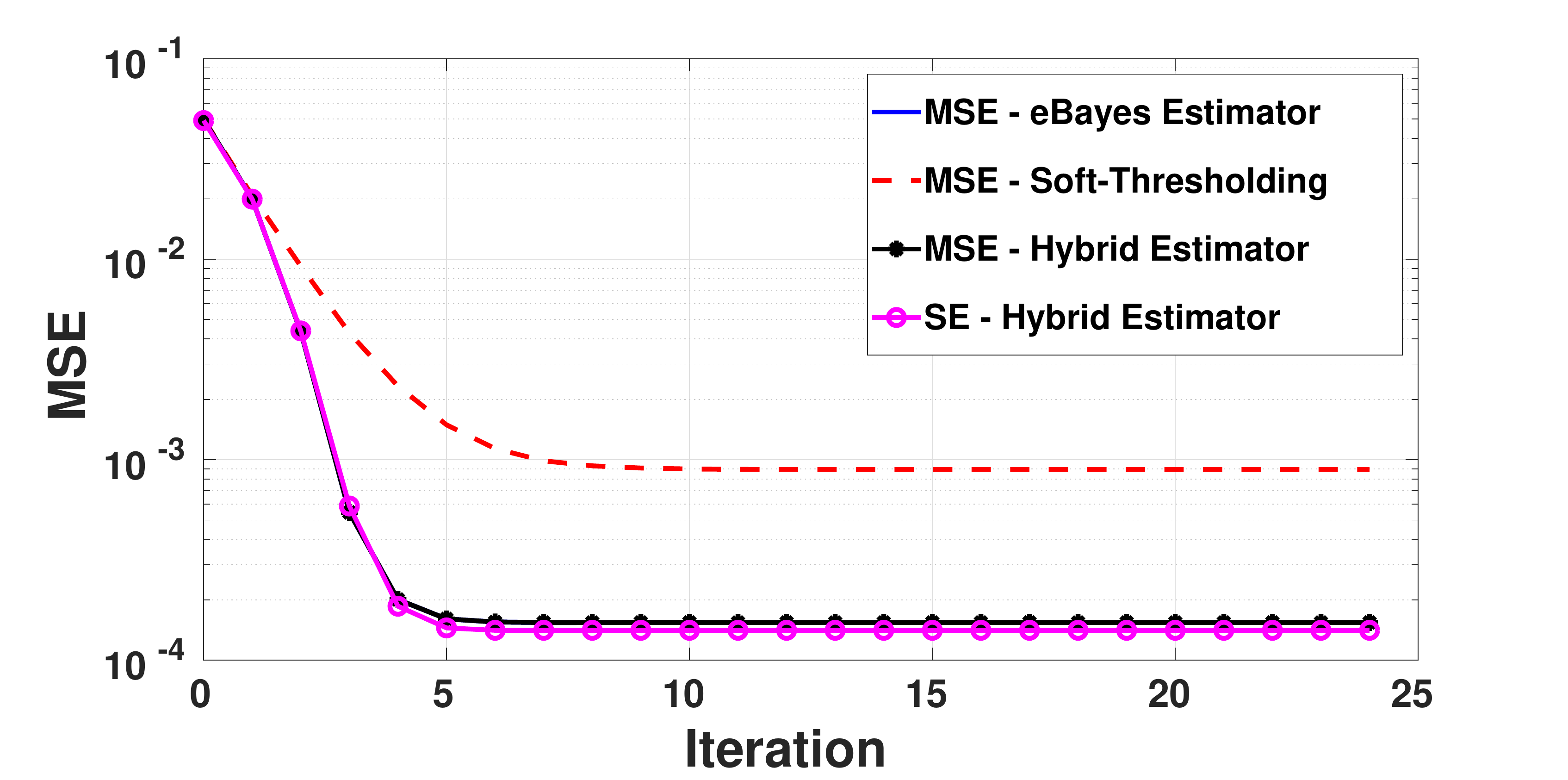}
       }
    \end{subfloat}   
    \caption{\small Plots of the mean squared error $\norm{\bst_t-\bst}^2/n$ as a function of the iteration number $t$ for the following cases, with $n=10,000$: a) $\delta = 0.65$, $\eta = 0.13$, $\sigma = 1$, the non-zero entries of $\bst$ are drawn from $\mathcal{N}(0,5)$. b) $\delta = 0.65$, $\eta = 0.13$, $\sigma = 1$, the non-zero entries are drawn from the uniform distribution between $[-5,5]$. c) $\delta = 0.5$, $\eta = 0.1$, $\sigma = 0$, the non-zero entries are drawn from the Rademacher distribution. d) $\delta = 0.5$, $\eta = 0.05$, $\sigma = 0.05$, the non-zero entries are drawn from the Rademacher distribution. The state evolution (SE) prediction of the MSE for the hybrid estimator is also shown in the plots.}
    \label{fig_9to12}
\end{figure*}

The simulation plots in Fig. \ref{fig_9to12} show the performances of the three estimators (soft-thresholding, the eBayes estimator, and the hybrid estimator) when used in the AMP algorithm. Throughout, we fix $n = 10000$ but consider various values of the undersampling ratio $\delta = m/n$, the sparsity factor $\eta = \norm{\bst}_0/n$, the noise variance $\sigma^2$, and the non-zero values of $\bst$. We choose such a large $n$ because the claim  that $\mbf{A}^T\mathbf{z}_t + \bst^t \stackrel{\text{d}}{=} \bst + \tau_t \mbf{Z}$ in every iteration $t$ holds in the large system limit. The measurement matrix $ \mathbf{A}$ is chosen with its entries i.i.d. $\sim \mathcal{N}(0,1/m)$, and the sparsity factor $\eta$ is assumed to be unknown. 

In each step of the algorithm, a suitable threshold $\lambda_t^*$ (for soft-thresholding), and a suitable sparsity parameter $\e_t^*$ (for the eBayes estimator) are chosen as described in \eqref{eq_opt_th} and \eqref{eq_opt_eps} with the only difference being that the risk estimates are now based on $\norm{\mathbf{z}_t}^2/m$ and not on SURE. To be precise, the updates in iteration $t$ for each case are generated as follows:
\begin{enumerate}
\item \emph{Soft-thresholding}:  Let $\mathcal{S} \vcentcolon=  \{ 0.1j, j \in [\lceil 10\sqrt{2\log n}\rceil] \}$. Then,  for each $\lambda \in \mathcal{S}$, compute:
\begin{align*}
\bst_{t}(\lambda) &=  \hat{\bst}_{ST}\left(\mathbf{A}^T \mathbf{z}_{t-1} + \bst_{t-1};\lambda \widehat{\tau}_{t-1} \right), \ \  \text{ where } 
\widehat{\tau}_{t-1} = \norm{ \mathbf{z}_{t-1}}/\sqrt{m}, \\
   \mathbf{z}_t(\lambda) &= \by - \mathbf{A}\bst_t(\lambda) + \frac{1}{\delta}\mathbf{z}_{t-1}\left\langle f_{t}^\prime\left( \mathbf{A}^T \mathbf{z}_{t-1} + \bst_{t-1};\lambda  \widehat{\tau}_{t-1} \right)\right\rangle.
\end{align*} 
Then choose $\lambda_t^*  = \argmin_{\lambda \in \mathcal{S}} \, \norm{\mathbf{z}_t(\lambda)}^2/m$, and generate the updated estimates 
\begin{align*}
 \bst_t & = \bst_{t}(\lambda_t^*),  ~~~~ \mathbf{z}_t = \mathbf{z}_t(\lambda_t^*).
\end{align*}
\item \emph{eBayes}: Let $\mathcal{D} \vcentcolon=  \{ 0.02j, j \in [ 50 ] \}$. Then, for each $\e \in \mathcal{D}$, compute:
\begin{align*}
\bst_{t}(\e) &= \hat{\bst}_{EB}\left(\mathbf{A}^T \mathbf{z}_{t-1} + \bst_{t-1};\e\right)  \   \text{ where } \hat{\bst}_{EB} \text{ is  modified for noise level } 
\widehat{\tau}_{t-1} = \frac{\norm{ \mathbf{z}_{t-1}}}{\sqrt{m}}, \\
   \mathbf{z}_t(\e) &= \by - \mathbf{A}\bst_t(\e) + \frac{1}{\delta}\mathbf{z}_{t-1}\left\langle f_{t}^\prime\left( \mathbf{A}^T \mathbf{z}_{t-1} + \bst_{t-1};\e \right)\right\rangle.
\end{align*} 
Then choose $ \e_t^*  = \argmin_{\e \in \mathcal{D}} \ \norm{\mathbf{z}_t(\e)}^2/m$, and generate the updated estimates
\begin{align*}
 \bst_t & = \bst_{t}(\e_t^*),  ~~~~ \mathbf{z}_t = \mathbf{z}_t(\e_t^*). 
 \end{align*}
\item \emph{Hybrid estimator}: In iteration $t$, $\bst_t$ is set to either $\bst_{t}(\lambda_t^*)$  or $\bst_{t}(\epsilon_t^*)$ depending on which of $\norm{\mathbf{z}_t(\lambda_t^*)}^2$ and $\norm{\mathbf{z}_t(\e_t^*)}^2$ is smaller. 
\end{enumerate}

The plots in Fig. \ref{fig_9to12} show the progression of the  mean squared error (MSE) $\norm{\bst_t - \bst}^2/n$  with the AMP iteration number $t$ for the three estimators when applied in the AMP algorithm for compressed sensing. We have also plotted the state evolution (SE) prediction of the MSE for the hybrid estimator (labelled ``SE - Hybrid Estimator'' on the plots), which for iteration $t$ is estimated as 
\begin{align*}
 \textrm{MSE}_{SE}(t) = \delta\left(\widehat{\tau}^2_{t} - \sigma^2\right)_+, ~~~~ \text{ where }\widehat{\tau}^2_{t} = \frac{\min\left( \norm{\mathbf{z}_t(\lambda_t^*)}^2, \norm{\mathbf{z}_t(\e_t^*)}^2\right)}{m}.
\end{align*}

\noindent   It can be inferred that the eBayes estimator provides a strong alternative to soft-thresholding in the AMP framework.

\begin{rem}
It was observed in numerical experiments that the optimal values $\lambda_t^*$ and $\epsilon_t^*$ do not vary much with the iteration index $t$. So, to reduce the computational load, one can compute $\lambda_t^*$ and $\epsilon_t^*$ in the first iteration (or the first few) alone, and then retain them for the rest of the steps.
\end{rem}

\begin{rem}
 As evident in Fig. \ref{fig_9to12}, the distinction between the MSEs for soft-thresholding and the eBayes estimator becomes clearer after the first few iterations. It has been observed in the experiments that for large values of $n$, the hybrid estimator picks the better estimator (eBayes estimator in the cases under consideration) with very high probability after around ten iterations. So, to further reduce the computational load, after a certain number of iterations (e.g., ten), one can continue with the most recent choice for the hybrid estimator. 
 \end{rem}
 
\begin{note}
The idea of using SURE to tune the parameters of the AMP denoising function $f_t$ has been previously  used  in \cite{MAR13,mousavi2018consistent,GD14}. In \cite{MAR13,mousavi2018consistent}, the authors tune the parameter of the soft-thresholding  denoiser by using a gradient descent based algorithm to optimize an objective function defined via SURE.  It is shown that the parameter estimates produced by this method converge to the asymptotically optimal values as the dimension grows. The paper \cite{GD14} uses SURE to tune the parameters of AMP denoising functions chosen from certain parametric kernel families, but does not provide theoretical guarantees on the performance of the proposed approach.

Beyond the context of AMP, SURE has been used to design denoising functions for images in \cite{FTM07, Blu_Lu}. In particular, \cite{FTM07} proposes   a hybrid estimator which is a mixture of derivatives of Gaussians while \cite{Blu_Lu} proposes an estimator that is a  mixture of more general exponential functions. 

\end{note}


\section{Proofs} \label{sec_proofs}

\subsection{Mathematical Preliminaries}
We list some lemmas that are used in the proofs of the theorems.

\begin{lemma}\label{lem_1a}

(a) (Hoeffding's lemma) \cite[Lemma 2.2]{boucheron2}: Let $X$ be a  random variable such that $\ex X = 0$ and $a \leq X \leq b$ almost surely. Then, for all $s \in \mathbb{R}$,
 $\ex \left[e^{s X}\right]\leq \exp\left({\frac {s ^{2}(b-a)^{2}}{8}}\right)$. 
Hence, $X$ is sub-Gaussian with variance factor $(b-a)^2/4$, and for any $t > 0$,
\begin{align*}
 \prb \left( \vert X \vert \geq t \right) \leq 2e^{-\frac{2t^2}{(b-a)^2}}.
\end{align*}

(b) (Hoeffding's inequality) \cite[Theorem 2.8]{boucheron2}  Let $X_1, \ldots, X_n$ be independent  random variables such that $X_i \in [a_i,b_i]$ almost surely for $i \in [n]$. Then for any $t > 0$,
\begin{align*}
 \prb \left( \abs{\sum_{i=1}^n X_i - \ex X_i} \geq t \right) \leq 2e^{-\frac{2t^2}{ \sum_i (b_i-a_i)^2}}.
\end{align*}
\end{lemma}

\begin{lemma}\label{lem_chung1} \cite[Thm. 8,9]{chung}
 Suppose $X_i$ are independent random variables satisfying $X_i \geq 0$, $\mathbb{E}[X_i^2] < \infty$, $\forall i \in [n]$. Let $X = \sum_{i=1}^n X_i$. Then, we have for any $ t > 0$, 
 \begin{equation*}
  \prb(X - \mathbb{E}[X] \leq -t) \leq e^{-\frac{t^2}{2\sum_{i=1}^n \mathbb{E}[X_i^2]} } .
 \end{equation*}
 On the other hand, if  $X_i \leq 0$ with $\mathbb{E}[X_i^2]< \infty$, $\forall i \in [n]$, then, we have for any $ t > 0$, 
 \begin{equation*}
  \prb(X - \mathbb{E}[X] \geq t) \leq e^{-\frac{t^2}{2\sum_{i=1}^n \mathbb{E}[X_i^2]} } .
 \end{equation*}
\end{lemma}

\begin{lemma}\label{lem_gau_lipschitz}
(Gaussian concentration inequality) \cite[Thm 5.6]{boucheron2}: Let $\mathbf{x} \sim \mathcal{N}(\boldsymbol{0},\mathbf{I})$  and let $f : \mathbb{R}^n \to \mathbb{R}$ denote an $L$-Lipschitz function, i.e., $\forall \mathbf{x}_1, \mathbf{x}_2 \in \mathbb{R}^n$, $\vert f(\mathbf{x}_1) - f(\mathbf{x}_2) \vert \leq L \Vert \mathbf{x}_1 - \mathbf{x}_2 \Vert$. Then, for all $t > 0$,
\begin{align*}
\prb\left(\vert f(\mathbf{x}) - \ex f(\mathbf{x}) \vert \geq t \right) \leq 2e^{-\frac{t^2}{2L^2}}. 
\end{align*} 
\end{lemma}

\begin{lemma}\label{lem1}
 Let $y_i = \theta_i + w_i$, $w_i \sim \mathcal{N}(0,1)$ and $\theta_i$ are deterministic constants, for $i=1,2,\cdots,n$. Let $f(\mathbf{y}) \vcentcolon= \frac{1}{n}\sum_{i=1}^n\frac{y_i^2}{\left(1 + ce^{-ay_i^2}\right)^p} $, where $a,c$ are positive constants and $p$ is a positive integer. Then for any $ t > 0$, we have
 \begin{align*}
 \mathbb{P}\left( f(\mathbf{y}) - \mathbb{E}\left[ f(\mathbf{y})\right]  \leq -t \right) & \leq e^{-\frac{nk_1t^2}{\sum_i \theta_i^4/n}}, \qquad
  \mathbb{P}\left(  f(\mathbf{y}) - \mathbb{E}\left[ f(\mathbf{y})\right]   \geq t \right)  \leq e^{-\frac{nk_2\min(t,t^2)}{\max(1,\Vert\bst\Vert^2/n)}},
 \end{align*}
where $k_1$ and $k_2$ are absolute positive constants.
\end{lemma}
\begin{proof} 
See Appendix \ref{sec_appendix}.
\end{proof}

\begin{lemma}\label{lem1a}
  Let $y_i = \theta_i + w_i$, $w_i \sim \mathcal{N}(0,1)$, $i=1,2,\cdots,n$, and let $f(\mathbf{y}) \vcentcolon= \frac{1}{n}\sum_{i=1}^n\frac{1}{1 + ce^{-ay_i^2}} $ for any positive constants $a$ and $c$. Then, we have for any $ t > 0$,
 \begin{align}
 \mathbb{P}\left( \vert f(\mathbf{y}) - \mathbb{E}\left[ f(\mathbf{y})\right] \vert \geq t \right) & \leq 2e^{-2n(1+c^2)t^2/c^2}.
 \end{align}
\end{lemma}
\begin{proof}
 This is a straightforward application of Hoeffding's inequality (Lemma \ref{lem_1a}(b)) after noting that $ \frac{1}{1 + ce^{-ay_i^2}} \in [\frac{1}{1+c},1)$.
\end{proof}

\begin{lemma}\label{lem2}
 Let $y_i = \theta_i + w_i$, $w_i \sim \mathcal{N}(0,1)$, $i=1,2,\cdots,n$, and let 
 $f(\mathbf{y}) \vcentcolon= \frac{1}{n}\sum_{i=1}^n\frac{y_i^2e^{-a_1y_i^2}}{\left(1 + ce^{-a_2y_i^2}\right)^2} $ for any positive constants $c$, $a_1$ and $a_2$. Then, we have for any $ t > 0$,
 \begin{align}
 \mathbb{P}\left( f(\mathbf{y}) - \mathbb{E}\left[ f(\mathbf{y})\right]  \leq -t \right) & \leq e^{-\frac{nk_1t^2}{\sum \theta_i^4/n}},\\
  \mathbb{P}\left(  f(\mathbf{y}) - \mathbb{E}\left[ f(\mathbf{y})\right]   \geq t \right) & \leq e^{-\frac{nk_2\min(t,t^2)}{\sum_i  \theta_i^2/n}},
 \end{align}
where $k_1$ and $k_2$ are absolute positive constants.
\end{lemma}
\begin{proof}
 The proof is along the lines of that for Lemma \ref{lem1}. Let $Z_i \vcentcolon= \frac{y_i^2 e^{-a_1y_i^2}}{n\left(1 + ce^{-a_2y_i^2}\right)^p}$.  Since $Z_i$'s
 are non-negative, the lower tail bound follows from Lemma \ref{lem_chung1}. As in  Lemma \ref{lem1}, the proof for the upper tail involves showing that $\norm{\nabla g(\by) }^2$ is bounded, where
 $g(\by) =\sqrt{f(\by)}$.
 \end{proof}
 
 \begin{lemma}\label{lem2a}
 Let $y_i = \theta_i + w_i$, where $w_i \sim \mathcal{N}(0,1)$ and $\theta_i$, $i=1,2,\cdots,n$, are deterministic. Let $f_1(\mathbf{y}) \vcentcolon= \frac{1}{n}\sum_{i=1}^n\frac{\theta_i y_i\mathsf{1}_{\{\theta_iy_i \geq 0 \}}}{1 + ce^{-ay_i^2}} $ and $f_2(\mathbf{y}) \vcentcolon= \frac{1}{n}\sum_{i=1}^n\frac{\theta_i y_i\mathsf{1}_{\{\theta_iy_i \leq 0 \}}}{1 + ce^{-ay_i^2}} $, where $a,c$ are positive constants. Then for any $ t > 0$, we have
 \begin{align} \label{eq:lem2a_st1}
 \mathbb{P}\left( f_1(\mathbf{y}) - \mathbb{E}\left[ f_1(\mathbf{y})\right]  \geq t \right) & \leq e^{-\frac{n t^2}{2 (1+c)^2\sum_i \theta_i^2/n}},\\ \label{eq:lem2a_st2}
 \mathbb{P}\left( f_2(\mathbf{y}) - \mathbb{E}\left[ f_2(\mathbf{y})\right]  \leq -t \right) & \leq e^{-\frac{n t^2}{2 (1+c)^2\sum_i \theta_i^2/n}},
 \end{align}
where $k$ is an absolute positive constant.
\end{lemma}

\begin{proof}
 We let $Z_i \vcentcolon= \frac{\theta_i y_i}{n\left(1 + ce^{-ay_i^2}\right)}$ and $f(\mathbf{y}) = \sum_{i=1}^n Z_i$. Now,
 \begin{align}\nonumber
  \norm{\nabla f(\by) }^2 &= \sum_{i=1}^n \left(\frac{\partial f(\by)}{\partial y_i} \right)^2 = \sum_{i=1}^n \frac{\theta_i^2}{n^2}\left[ \frac{1}{1 + ce^{-ay_i^2}} + \frac{acy_i^2e^{-ay_i^2}}{(1 + ce^{-ay_i^2})^2} \right]^2\\ \label{eq:lem2a_1}
  & \leq \sum_{i=1}^n \frac{\theta_i^2}{n^2}\left[ \frac{1}{1 + ce^{-ay_i^2}} + \frac{c}{(1 + ce^{-ay_i^2})^2} \right]^2 \leq \frac{(1+c)^2\norm{\bst}^2}{n^2}.
 \end{align}
Now, let $\mathcal{R} \vcentcolon= \{ \by \in \mathbb{R}^n ~\vert ~ \theta_i y_i \geq 0 ,\forall i \in [n]\}$. Let $h: \mathbb{R} \to \mathbb{R}$ be the function defined as 
\begin{align*}
 h(y_i) = \left\{ \begin{array}{cc}
                  y_i & \textrm{if }\theta_i y_i \geq 0,\\
                  0 & \textrm{otherwise}.\\
                 \end{array}\right.
\end{align*}
For any $\by \in \mathbb{R}^n$, let $h(\by ) \in \mathbb{R}^n$ be the vector obtained by  applying $h(\cdot )$ component-wise on the elements of $\by$. Since $f_1(\by) = f_1(h(\by))$, we have for any $\by_1,\by_2 \in \mathbb{R}^n$, 
\begin{align*}
 \vert f_1(\by_1) - f_1(\by_2) \vert = \vert f_1(h(\by_1)) - f_1(h(\by_2)) \vert & \overset{(a)}{=} \langle \nabla f(\mathbf{c}),h(\by_1) - h(\by_2) \rangle \\
 &\overset{(b)}{\leq} \norm{\nabla f(\mathbf{c})} \Vert h(\by_1) - h(\by_2) \Vert \leq  \norm{\nabla f(\mathbf{c})} \Vert \by_1 -\by_2 \Vert\\
 & \leq L_n \Vert \by_1 -\by_2 \Vert
\end{align*}
where step $(a)$ is due to the mean value theorem with $\mathbf{c} = h(\by_1) + c ( h(\by_2) - h(\by_1) )$ for some $c \in [0,1]$, step $(b)$ is due to the Cauchy-Schwarz inequality, and $L_n \vcentcolon= \sup_{\by \in \mathcal{R}} \{ \norm{\nabla f(\by)} \}  \leq (1+c)\norm{\bst}/n$ from \eqref{eq:lem2a_1}. Therefore, using Lemma \ref{lem_gau_lipschitz}, we get \eqref{eq:lem2a_st1} (note that we do not require the lower tail inequality for $f_1(\by)$ in this paper). Proceeding in the same manner, it is straightforward to obtain the proof of \eqref{eq:lem2a_st2}.
\end{proof}

 \begin{lemma} \label{lem3}
 Let $\{X_n(\bst),\bst \in \mathbb{R}^n\}_{n \geq 1}$ be a sequence of random variables such that for any $t > 0$,
 \begin{equation*}
  \mathbb{P}(\vert X_n(\bst) \vert \geq t) \leq Ke^{-nk\min(t,t^2)},
 \end{equation*}
where $K$ and $k$ are positive constants. Then 
\[ \ex \vert X_n(\bst) \vert \leq  \frac{c_1}{\sqrt{n}}\left(1 +  \frac{c_2}{\sqrt{n}} \right), \]
where $c_1 = \frac{K}{2}\sqrt{\frac{\pi}{k}}$, $c_2 = \frac{2}{\sqrt{k\pi}}$. 
\end{lemma}
\begin{proof}
We have
\begin{align*}
 \mathbb{E}\left[ \vert X_n \vert \right]& = \int_{0}^\infty \mathbb{P}\left( \vert X_n \vert > t\right)dt {~\leq} \int_{0}^1 Ke^{-nkt^2}dt + \int_{1}^\infty Ke^{-nkt}dt\\
 & < \int_{0}^\infty Ke^{-nkt^2}dt + \int_{0}^\infty Ke^{-nkt}dt\\
 & = \frac{K}{\sqrt{nk}}\int_{0}^{\infty}e^{-x^2}dx + \frac{K}{nk}\int_{0}^\infty e^{-x}dx  = \frac{c_1}{\sqrt{n}}\left(1 +  \frac{c_2}{\sqrt{n}} \right).
\end{align*} 
\end{proof}

\begin{lemma}\label{lem_plipschitz}
(Concentration for sum of pseudo-Lipschitz function of sub-Gaussians \cite[Lemma A.11]{cindy_ramji}). Let $f : \mathbb{R} \to \mathbb{R}$ be a pseudo-Lipschitz function \cite{BayMont11} of order 2 with pseudo-Lipschitz constant $L$, i.e., for any $x,y \in \mathbb{R}$,
$\vert f(x) - f(y) \vert \leq L(1 + \vert x \vert + \vert y \vert) \vert x-y \vert$. Let $\mathbf{z} \in \mathbb{R}^n$ be a random vector with entries i.i.d. sub-Gaussian random variables with variance factor $\nu$. Then, for any $t > 0$,
\begin{align*}
 \prb \left( \frac{1}{n}\left \vert \sum_{i=1}^n f(z_i) -  \ex[f(z_i)] \right \vert \geq t \right) \leq 2e^{-nk \min(t,t^2)}
\end{align*}
where $k$ is some absolute constant (inversely proportional to $L^2$).
\end{lemma}

\subsection{Proof of Theorem \ref{thm1}}\label{subsec_proof_2}

To complete the proof in Sec. \ref{subsec:eb_conc}, we need to prove Lemmas \ref{lem:un} and \ref{lem:fn}. We start with the latter.

\noindent \textbf{Proof of Lemma \ref{lem:fn}}:

We want to obtain a bound for $\prb(\abs{f_n} \geq t)$, where 
 \be f_n \vcentcolon= \frac{a^2_\by}{n}\sum_{i=1}^n\frac{y_i^2}{\left(1 + c_\by e^{-ay_i^2/2}\right)^2} 
 - \frac{a^2}{n}\sum_{i=1}^n\ex \left[\frac{y_i^2}{\left(1 + ce^{-ay_i^2/2}\right)^2}\right],  \label{eq:fdef2} \ee  
with  the deterministic values $a,c$ defined in \eqref{eq:bc_defs}. For brevity, we define $b \vcentcolon = \frac{\norm{\bst}^2}{n}$. Also define the event
\begin{align}\label{eq_E}
 \mathcal{E} \vcentcolon= \left \{ \by \left\vert b - u \leq \frac{\norm{\by}^2}{n} - 1 \leq b + u \right.  \right \},
\end{align}
where $u >0$ will be specified later. 
From \eqref{eq4_thm1}, we have $\prb \left( \mathcal{E}^c \right) \leq 2 e^{-nk\min(u,u^2)}$. Therefore,
\begin{equation} \label{eq_fn_bound}
\begin{split}
 \prb \left( \vert f_n \vert \geq t \right) & = \prb \left(\vert f_n \vert \geq t, \mathcal{E}  \right) + \prb \left(\vert f_n \vert \geq t, \mathcal{E}^c  \right)\\  
 & \leq \prb (\mathcal{E}) \prb \left(\vert f_n \vert  \geq t \vert \mathcal{E}  \right) + 2e^{-nk\min(u,u^2)} \\ 
 & =  \prb (\mathcal{E}) \prb \left( f_n \geq t \vert \mathcal{E}  \right) + \prb (\mathcal{E}) \prb \left( f_n \leq -t \vert \mathcal{E}  \right)+  2e^{-nk\min(u,u^2)}.
 \end{split}
\end{equation}
Now, when event $\mathcal{E}$ occurs, from the definition of $a_\by$ in \eqref{eq:aydy_def} we have 
\begin{align*}
  \left[1 - \frac{\e}{b - u + \e} \right]_+  \leq a_\by \leq 1 - \frac{\e}{b + u + \e} .
\end{align*}
We therefore have the following lower and upper bounds:
\begin{align}
a_\by & \geq  a_L  \vcentcolon \max \left\{ \frac{b-u}{b + \e}, \  0  \right\}, \label{eq_aL} \\
a_\by & \leq a_U  \vcentcolon = \min \left\{ \frac{b+u}{b + \e}, \ 1 \right\}. \label{eq_aU}
\end{align}
Similarly when $\mc{E}$ occurs, $\frac{1-\e}{\e^{3/2}}\sqrt{\e+[b - u]_+  } \leq c_\by \leq \frac{1-\e}{\e^{3/2}}\sqrt{b + u + \e}$, and we have the bounds
\begin{align}
\label{eq_cL}
c_\by  \geq c_L  & \vcentcolon = c - \kappa_1 \min(u,b),  \\
c_\by \leq c_U &  \vcentcolon = c + \kappa_1 u, \label{eq_cU}
\end{align}
where $\kappa_1 \vcentcolon =  \frac{c}{b+\e}= \frac{1-\e}{\e^{3/2} \sqrt{b+\e}}$.   Using these bounds in the definition of $f_n$ in \eqref{eq:fdef2}, we have
\begin{align}
 &\prb \left( f_n \geq t  \vert \mathcal{E}\right) \nonumber \\
 & \leq \prb\left( \left. \frac{a_U^2}{n}\sum_{i=1}^n\frac{y_i^2}{\left(1 + c_Le^{-a_Uy_i^2/2}\right)^2} - \frac{a^2}{n}\sum_{i=1}^n\ex \left[\frac{y_i^2}{\left(1 + ce^{-ay_i^2/2}\right)^2}\right] \geq t  \right\vert \mathcal{E}\right) \nonumber \\
 & =  \prb\left( \left. \frac{a_U^2}{n}\sum_{i=1}^n\frac{y_i^2}{\left(1 + c_Le^{-a_Uy_i^2/2}\right)^2} - \frac{a_U^2}{n}\sum_{i=1}^n\ex \left[\frac{y_i^2}{\left(1 + c_Le^{-a_Uy_i^2/2}\right)^2}\right] + \Delta_{1n} \geq t  \right\vert \mathcal{E}\right),  \label{eq3_thm1}
 \end{align}
where 
\be \Delta_{1n} \vcentcolon = \frac{a_U^2}{n}\sum_{i=1}^n\ex \left[\frac{y_i^2}{\left(1 + c_Le^{-a_Uy_i^2/2}\right)^2}\right] - \frac{a^2}{n}\sum_{i=1}^n\ex \left[\frac{y_i^2}{\left(1 + ce^{-ay_i^2/2}\right)^2}\right].  \label{eq:Delta1_def} \ee
Next, 
\begin{equation} 
\label{eq5_thm1}	
\begin{split}
 &\prb\left(\left. \frac{a_U^2}{n}\sum_{i=1}^n\frac{y_i^2}{\left(1 + c_Le^{-a_Uy_i^2/2}\right)^2} - \frac{a_U^2}{n}\sum_{i=1}^n\ex \left[\frac{y_i^2}{(1 + c_Le^{-a_Uy_i^2/2})^2}\right] \geq t \right \vert \mathcal{E} \right) \\
 & \leq \frac{1}{\prb(\mathcal{E})}{\prb\left( \frac{a_U^2}{n}\sum_{i=1}^n \left( \frac{y_i^2}{\left(1 + c_Le^{-a_Uy_i^2/2}\right)^2} - \ex \left[\frac{y_i^2}{(1 + c_Le^{-a_Uy_i^2/2})^2}\right] \right) \geq t \right)} 
 \leq \frac{e^{-nk\min(t,t^2)}}{\prb(\mathcal{E})},
 \end{split}
 \end{equation}
where the last inequality follows from Lemma \ref{lem1} after noting that $c_L \leq c$ and hence upper bounded, and $k$ is some absolute positive constant due to the assumption that $\sum_{i=1}^n \theta_i^4/n < \Lambda$.  Next, it is shown in Appendix \ref{app:Delta12} that $\Delta_{1n} \leq \kappa_3 u$, where $\kappa_3 \leq 2 \frac{(b + 1)}{b} (1 + a + \kappa_1 a^2 b)$ is an absolute positive constant. Using this bound on $\Delta_{1n}$ and \eqref{eq5_thm1} in \eqref{eq3_thm1} we obtain
\ben
\begin{split}
&\prb\left( \left. \frac{a_U^2}{n}\sum_{i=1}^n\frac{y_i^2}{\left(1 + c_Le^{-a_Uy_i^2/2}\right)^2} - \frac{a_U^2}{n}\sum_{i=1}^n\ex \left[\frac{y_i^2}{\left(1 + c_Le^{-a_Uy_i^2/2}\right)^2}\right] + \Delta_{1n} \geq t  +\kappa_3u  \right\vert \mathcal{E}\right) \\
 &\leq \frac{e^{-nk\min(t,t^2)}}{ \prb (\mathcal{E})}.
\end{split}
\een
 Choosing $u = t$, we finally have, from \eqref{eq3_thm1},
 \begin{align}\label{eq6_thm1}
 \prb (\mathcal{E}) \prb \left( f_n \geq t \right \vert \mathcal{E}) \leq e^{-nk\min(t,t^2)}
\end{align}
for a suitable absolute positive constant $k$. 

The lower tail bound is established as follows in a similar manner.
\begin{align}\label{eq7_thm1}
 & \prb \left( f_n \leq -t \right \vert \mathcal{E}) \\ \nonumber
 &\leq \prb\left(\left. \frac{a_L^2}{n}\sum_{i=1}^n\frac{y_i^2}{\left(1 + c_Ue^{-a_Ly_i^2/2}\right)^2} - \frac{a_L^2}{n}\sum_{i=1}^n\ex \left[\frac{y_i^2}{\left(1 + c_U e^{-ay_i^2/2}\right)^2}\right] - \Delta_{2n} \leq -t \right \vert \mathcal{E} \right),
\end{align}
where 
\be \Delta_{2n} \vcentcolon =  \frac{a^2}{n}\sum_{i=1}^n\ex \left[\frac{y_i^2}{\left(1 + ce^{-ay_i^2/2}\right)^2}\right]
- \frac{a_L^2}{n}\sum_{i=1}^n\ex \left[\frac{y_i^2}{\left(1 + c_U e^{-a_L y_i^2/2}\right)^2}\right].  \label{eq:Delta2_def} \ee
 Next, we have 
\begin{align} \nonumber
&\prb\left(\left.\frac{a_L^2}{n}\sum_{i=1}^n\frac{y_i^2}{\left(1 + c_Ue^{-a_Ly_i^2/2}\right)^2} - \frac{a_L^2}{n}\sum_{i=1}^n\ex \left[\frac{y_i^2}{\left(1 + c_Ue^{-a_Ly_i^2/2}\right)^2}\right] \leq -t \right \vert \mathcal{E} \right)\\
& \leq  \frac{1}{\prb(\mathcal{E})} \prb\left(\frac{a_L^2}{n}\sum_{i=1}^n \left( \frac{y_i^2}{(1 + c_Ue^{-a_Ly_i^2/2})^2} -\ex \left[\frac{y_i^2}{(1 + c_Ue^{-a_Ly_i^2/2})^2}\right] \right) \leq -t\right)  \leq \frac{e^{-nkt^2}}{\prb(\mathcal{E})},   \label{eq8_thm1}
\end{align}
where the last inequality uses Lemma \ref{lem1}, and $k$ is an absolute positive constant due to the assumption that $\sum_{i=1}^n \theta_i^4/n < \Lambda$.
 Next it is shown in Appendix \ref{app:Delta12} that $\Delta_{2n} \leq \kappa_3 u$, where $\kappa_3 \leq 2 \frac{(b + 1)}{b} (1 + a + \kappa_1 a^2 b)$ is an absolute positive constant. Using this bound on $\Delta_{2n}$ and \eqref{eq8_thm1} in \eqref{eq7_thm1} we obtain
\begin{align*}
& \prb\left(\left. \frac{a_L^2}{n}\sum_{i=1}^n\frac{y_i^2}{\left(1 + c_Ue^{-a_Ly_i^2/2}\right)^2} - \frac{a_L^2}{n}\sum_{i=1}^n\ex \left[\frac{y_i^2}{\left(1 + c_U e^{-ay_i^2/2}\right)^2}\right] - \Delta_{2n} \leq -t -\kappa_3 u \right \vert \mathcal{E} \right)\\
&  \leq e^{-nkt^2}/\prb (\mathcal{E}).
\end{align*}
  Choosing $u = t$, we finally have, from \eqref{eq7_thm1},
 \begin{align}\label{eq9_thm1}
 \prb (\mathcal{E})\prb \left( f_n \leq -t \right \vert \mathcal{E}) \leq e^{-nkt^2}
\end{align}
for some suitable absolute positive constant $k$. So, using \eqref{eq6_thm1} and \eqref{eq9_thm1} in \eqref{eq_fn_bound}, we arrive at 
\begin{align*}
 \prb \left( \vert f_n  \vert \geq t \right ) \leq 4e^{-nk\min(t,t^2)}.
\end{align*}

The proofs of the concentration inequalities for $g_n$, $h_n$, and $w_n$ are along similar lines to the steps from \eqref{eq3_thm1}-\eqref{eq9_thm1}. In particular, the concentration inequality for $g_n$ involves the application of Lemma \ref{lem1} as done in \eqref{eq5_thm1} and \eqref{eq8_thm1}, and that for $h_n$ involves the application of Lemma \ref{lem1a} to obtain inequalities of the form \eqref{eq5_thm1} and \eqref{eq8_thm1}. The concentration inequality for $w_n$ involves the application of Lemma \ref{lem2} and is also similar to the steps from \eqref{eq3_thm1}-\eqref{eq9_thm1} with two notable differences: 

\begin{enumerate}
 \item First, we establish that  both $(a_U^2 c_U - a^2c) \leq p_1u$ and $(a^2 c - a_L^2 c_L) \leq p_2 u$  for some positive constants $p_1, p_2$.
  This can be done using \eqref{eq_aL}--\eqref{eq_cU}. Indeed, for $u \geq  b$, $a_U = 1$, and so $a_U^2c_U - a^2c \leq (1-a^2)c + \kappa_1u \leq \frac{(1-a^2)cu}{b} + \kappa_1 u = p_1 u $ where $p_1 \vcentcolon= \frac{(1-a^2)c}{b} + \kappa_1$. For $0 \leq u < b$, $a_U^2c_U - a^2c$ is convex in $u$ and hence, it is clear that for any $u > 0$, $a_U^2c_U - a^2c \leq p_1 u$. Next, it is clear from \eqref{eq_cL} and \eqref{eq_cU} that for $0 \leq u < b$, $a^2 c - a_L^2 c_L$ is a concave function, and is bounded by $a^2c$ for $u \geq b$. Hence, $a^2 c - a_L^2 c_L \leq p_2 u$ for $u >0$, where $p_2$ is the derivative of $a^2 c - a_L^2 c_L$ at $u = 0$, i.e., $p_2 = 3(a+2)c/(b+\e)$.

\item  Next, we use the above bounds to  show that 
 \begin{align*}
 \frac{1}{n}\sum_{i=1}^n\ex \left[\frac{a_U^2c_U^2 y_i^2e^{-a_Ly_i^2/2}}{\left(1 + c_Le^{-a_Uy_i^2/2}\right)^2} - \frac{a^2 c^2 y_i^2e^{-ay_i^2/2}}{\left(1 + ce^{-ay_i^2/2}\right)^2}\right] \leq q_1 u\\
  \frac{1}{n}\sum_{i=1}^n\ex \left[\frac{a^2 c^2 y_i^2e^{-ay_i^2/2}}{\left(1 + ce^{-ay_i^2/2}\right)^2} - \frac{a_L^2 c_L^2 y_i^2e^{-a_Uy_i^2/2}}{\left(1 + c_Ue^{-a_Ly_i^2/2}\right)^2}\right] \leq q_2 u
 \end{align*}
for some positive constants $q_1$ and $q_2$. This is done using steps similar to those used to bound $\Delta_{1n}, \Delta_{2n}$ in Appendix \ref{app:Delta12}.
\end{enumerate}

\noindent \textbf{Proof of Lemma \ref{lem:un}}:

We show that $u_n, v_n, x_n$ are each bounded by order $1/n$ quantities, and then apply Lemma \ref{lem_1a}(a) to obtain the concentration result.

\textbf{Concentration for $u_n$}: As $b_i(\by) \geq 1$ for all $i$, we have 
\begin{align*}
 \frac{1}{n} \sum_{i=1}^n \frac{y_i^2}{d_\by^2b_i(\by)} \mathsf{1}_{\{\norm{\by}^2 > n\}}  & \leq \frac{1}{n} \sum_{i=1}^n \frac{y_i^2}{d_\by^2} \mathsf{1}_{\{\norm{\by}^2 > n\}}  = \frac{\e^2\left(\norm{\by}^2/n\right)}{\left[\left(\norm{\by}^2/n\right) - 1 + \e  \right]^2  }\mathsf{1}_{\{\norm{\by}^2 > n\}} \leq 1.
\end{align*}
Therefore,
\begin{align*}
  \frac{4}{\e n^2} \sum_{i=1}^n \frac{y_i^2}{d_\by^2b_i(\by)} \mathsf{1}_{\{\norm{\by}^2 > n\}} \in \left[0,4/(n\e)\right].
\end{align*}
Applying Lemma \ref{lem_1a}(a), we obtain that for any $t > 0$, 
$ \prb \left( \vert u_n \vert \geq t \right) \leq 2 e^{-n^2 kt^2}$
for a suitable positive constant $k$. 

\noindent \textbf{Concentration for $v_n$}:  As $d_\by ,b_i(\by)\geq 1$, $1 \leq i \leq n$, we have 
\begin{align*}
 \frac{1}{n}\sum_{i=1}^n\frac{ a_\by y_i^4  e^{-\frac{a_\by y_i^2}{2}}}{d_\by^{3/2}b_i^2(\by)}  \leq  \frac{1}{n}\sum_{i=1}^n a_\by y_i^4  e^{-\frac{a_\by y_i^2}{2}} .
\end{align*}
Next, we show that 
\begin{align}\label{eq_hard_bound}
 f(\mbf{y}):= \frac{1}{n} \sum_{i=1}^n { a_\by  y_i^4  e^{-\frac{a_\by y_i^2}{2}}}  \leq K,
\end{align}
 where $K$ is a constant to be determined. Let $C>0$ be a constant to be fixed later. There are two cases:
\begin{enumerate}
\item $\norm{\by}^2/n \leq 1+C$: In this case, use the bound $e^{-x} < \frac{1}{x}$ for all $x >0$. Using this we have
\begin{align*}
 f(\mbf{y}) :=    \frac{1}{n} \sum_{i=1}^n { a_\by  y_i^4  e^{-\frac{a_\by y_i^2}{2}}}  
  \, \leq  \,    \mathsf{1}_{\{\norm{\by}^2 > n\}} \frac{2}{n} \sum_{i=1}^n  y_i^2 
 \, \leq \, 2 (1+C)
\end{align*}  
by assumption.

\item  $\norm{\by}^2/n > 1+C$: In this case, from the definition of $a_\by$ in \eqref{eq:aydy_def} note that $a_\by > C/(C+\e)$. Now use the bound $e^{-x} < \frac{1}{x^2}$ for all $x >0$ to obtain 
\begin{align*}
  f(\mbf{y}):=   \frac{1}{n} \sum_{i=1}^n { a_\by  y_i^4  e^{-\frac{a_\by y_i^2}{2}}}  
 \leq   \frac{4}{n} \sum_{i=1}^n  \frac{1}{a_\by}  < \frac{4(C+\e)}{C}.
\end{align*}
 \end{enumerate}

\noindent Choosing $C= (1 + \sqrt{1 +8\e})/2$ to make the two bounds equal yields $K= 3 +\sqrt{1 +8\e}$. Therefore,
\begin{align*}
 \frac{2(1-\e )}{\e^2 n^2} \sum_{i=1}^n\frac{ a_\by y_i^4  e^{-\frac{a_\by y_i^2}{2}}}{d_\by^{3/2}b_i^2(\by)}   \in \left[0 , \frac{2(1-\e)K}{n\e^2} \right]. 
\end{align*}
Applying Lemma \ref{lem_1a}(a) yields, for any $t > 0$,
\begin{align*}
 \prb \left( \vert v_n \vert \geq t \right) \leq 2 e^{-n^2 kt^2}
\end{align*}
for a suitable positive constant $k$. 

\vspace{10pt}
\noindent \textbf{Concentration for $x_n$}:  Since $xe^{-x} \leq 1/e$ for $x > 0$, 
\begin{align*}
 \frac{2(1-\e)}{n^2\e^2}  \sum_{i=1}^n \frac{ a_\by y_i^2  e^{-\frac{a_\by y_i^2}{2}}}{ \sqrt{d_\by} b_i^2(\by)} \leq \frac{4(1-\e)}{n\e^2 e}.
\end{align*}
A direct application of Lemma \ref{lem_1a}(a) results in $\prb \left( \vert x_n \vert \geq t \right) \leq 2 e^{-n^2 kt^2}$ for any $t > 0$ and some positive constant $k$. \qedhere

\subsection{Proof of \eqref{eq_conc_loss_thm1a} for Theorem \ref{thm1a}} \label{sec_proof_thm1a}

The goal is to obtain a bound for $\prb(\abs{s_n} \geq t)$, where 
 \be s_n \vcentcolon= \frac{a_\by}{n}\sum_{i=1}^n\frac{\theta_i y_i}{1 + c_\by e^{-ay_i^2/2}} 
 - \frac{a}{n}\sum_{i=1}^n\ex \left[\frac{\theta_i y_i}{1 + ce^{-ay_i^2/2}}\right],  \label{eq:sdef2} \ee  
with  the deterministic values $a,c$ defined in \eqref{eq:bc_defs}. Since the summands of the random term in \eqref{eq:sdef2} can take both positive and negative values, we employ the following approach to obtain the concentration inequality.

Let $s_n = s_{n}^+ + s_{n}^-$ where 
\begin{align*}
s_{n}^+ &\vcentcolon= \frac{a_\by}{n}\sum_{i=1}^n\frac{\theta_i y_i \mathsf{1}_{\{\theta_i y_i \geq 0 \}}}{1 + c_\by e^{-ay_i^2/2}} 
 - \frac{a}{n}\sum_{i=1}^n\ex \left[\frac{\theta_i y_i\mathsf{1}_{\{\theta_i y_i \geq 0 \}}}{1 + ce^{-ay_i^2/2}}\right],\\
 s_{n}^- &\vcentcolon= \frac{a_\by}{n}\sum_{i=1}^n\frac{\theta_i y_i\mathsf{1}_{\{\theta_i y_i \leq 0 \}}}{1 + c_\by e^{-ay_i^2/2}} 
 - \frac{a}{n}\sum_{i=1}^n\ex \left[\frac{\theta_i y_i\mathsf{1}_{\{\theta_i y_i \leq 0 \}}}{1 + ce^{-ay_i^2/2}}\right].
\end{align*}

\noindent Using \eqref{eq_E} and proceeding along the lines of \eqref{eq_fn_bound}, we obtain
\begin{align} \nonumber
 \prb \left( \vert s_n \vert \geq t \right) &=  \prb \left( \vert s_n \vert \geq t , \mathcal{E}\right) +  \prb \left( \vert s_n \vert \geq t , \mathcal{E}^c\right) \\  \nonumber
 &  \leq  \prb(\mathcal{E})\prb \left( \vert s_n \vert \geq t \vert \mathcal{E}\right) +  \prb \left( \mathcal{E}^c\right) \\ \label{eq_sn_bound}
 & \leq \prb(\mathcal{E})\prb \left( \vert s_n^+ \vert \geq t/2 \vert \mathcal{E}\right) + \prb(\mathcal{E})\prb \left( \vert s_n^- \vert \geq t/2 \vert \mathcal{E}\right)+  2e^{-nk\min(u,u^2)} 
\end{align}
where $u > 0$ will be specified later. Now, with $a_U$ and $c_L$ as respectively defined in \eqref{eq_aU} and \eqref{eq_cL}, we have
\begin{align} \nonumber
 \prb \left( s_n^+ \geq t/2  \vert \mathcal{E}\right)& \leq \prb\left( \left. \frac{a_{U}}{n}\sum_{i=1}^n\frac{\theta_i y_i\mathsf{1}_{\{\theta_i y_i \geq 0 \}}}{1 + c_{L}e^{-a_Uy_i^2/2}} - \frac{a}{n}\sum_{i=1}^n\ex \left[\frac{\theta_i y_i\mathsf{1}_{\{\theta_i y_i \geq 0 \}}}{1 + ce^{-ay_i^2/2}}\right] \geq t/2  \right\vert \mathcal{E}\right) \\ \label{eq3_thms1}
 & =  \prb\left( \left. \frac{a_{U}}{n}\sum_{i=1}^n\frac{\theta_i y_i\mathsf{1}_{\{\theta_i y_i \geq 0 \}}}{1 + c_{L}e^{-a_Uy_i^2/2}} - \frac{a_{U}}{n}\sum_{i=1}^n\ex \left[\frac{\theta_i y_i\mathsf{1}_{\{\theta_i y_i \geq 0 \}}}{1 + c_{L}e^{-a_Uy_i^2/2}}\right] + \Delta_{3n} \geq t/2  \right\vert \mathcal{E}\right),  
 \end{align}
where 
\be \Delta_{3n} \vcentcolon = \frac{a_{U}}{n}\sum_{i=1}^n\ex \left[\frac{\theta_i y_i\mathsf{1}_{\{\theta_i y_i \geq 0 \}}}{1 + c_{L}e^{-a_Uy_i^2/2}}\right] - \frac{a}{n}\sum_{i=1}^n\ex \left[\frac{\theta_i y_i\mathsf{1}_{\{\theta_i y_i \geq 0 \}}}{1 + ce^{-ay_i^2/2}}\right].  \label{eq:Deltas1_def} \ee
Next, 
\begin{equation} 
\label{eq5_thms1}	
\begin{split}
 &\prb\left(\left. \frac{a_{U}}{n}\sum_{i=1}^n\frac{\theta_i y_i\mathsf{1}_{\{\theta_i y_i \geq 0 \}}}{1 + c_{L}e^{-a_Uy_i^2/2}} - \frac{a_{U}}{n}\sum_{i=1}^n\ex \left[\frac{\theta_i y_i\mathsf{1}_{\{\theta_i y_i \geq 0 \}}}{1 + c_{L}e^{-a_Uy_i^2/2}}\right] \geq t/2 \right \vert \mathcal{E} \right) \\
 & \leq \frac{1}{\prb(\mathcal{E})}{\prb\left( \frac{a_{U}}{n}\sum_{i=1}^n \left( \frac{\theta_i y_i\mathsf{1}_{\{\theta_i y_i \geq 0 \}}}{1 + c_Le^{-a_Uy_i^2/2}} - \ex \left[\frac{\theta_i y_i\mathsf{1}_{\{\theta_i y_i \geq 0 \}}}{1 + c_{L}e^{-a_Uy_i^2/2}}\right] \right) \geq t/2 \right)} 
 \leq \frac{e^{-n \kappa t^2 } }{\prb(\mathcal{E})},
 \end{split}
 \end{equation}
where the last inequality follows from Lemma \ref{lem2a}, and $\kappa = (8 a_U^2 (1+c_L)^2 \norm{\bst}^2/n )^{-1}$. We note that   $\kappa$ is bounded from below by an absolute positive constant due to Assumption A and the fact that $a_U \leq 1, c_L \leq c$ and hence bounded.  Next, it is shown in Appendix \ref{app:Delta12} that $\Delta_{3n} \leq \kappa_4 u$, where $\kappa_4 \leq \sqrt{\frac{b+1}{b}}\left(2 + \kappa_1 ab  \right)$ is an absolute positive constant. Using this bound on $\Delta_{3n}$ and \eqref{eq5_thms1} in \eqref{eq3_thms1} we obtain
\ben
\begin{split}
&\prb\left( \left. \frac{1}{n}\sum_{i=1}^n\frac{a_{U}\theta_i y_i}{1 + c_{L}e^{-a_Uy_i^2/2}} - \frac{1}{n}\sum_{i=1}^n\ex \left[\frac{a_{U}\theta_i y_i}{1 + c_{L}e^{-a_Uy_i^2/2}}\right] + \Delta_{3n} \geq t/2  +\kappa_4u  \right\vert \mathcal{E}\right) \\
 &\leq \frac{e^{-n \kappa t^2 }}{ \prb (\mathcal{E})}.
\end{split}
\een
 Choosing $u = t$, we finally have, from \eqref{eq3_thms1},
 \begin{align}\label{eq4_thms1}
 \prb (\mathcal{E}) \prb \left( s_n^+ \geq t/2 \right \vert \mathcal{E}) \leq e^{-nk t^2}
\end{align}
for a suitable absolute positive constant $k$. 

To establish the lower tail, we proceed as follows. With $a_L$ and $c_U$ as respectively defined in \eqref{eq_aL} and \eqref{eq_cU}, we have
\begin{align} \nonumber
 \prb \left( s_n^+ \leq -t/2  \vert \mathcal{E}\right)& \leq \prb\left( \left. \frac{a_{L}}{n}\sum_{i=1}^n\frac{\theta_i y_i\mathsf{1}_{\{\theta_i y_i \geq 0 \}}}{1 + c_{U}e^{-a_Ly_i^2/2}} - \frac{a}{n}\sum_{i=1}^n\ex \left[\frac{\theta_i y_i\mathsf{1}_{\{\theta_i y_i \geq 0 \}}}{1 + ce^{-ay_i^2/2}}\right] \leq -t/2  \right\vert \mathcal{E}\right) \\ \label{eq6_thms1}
 & =  \prb\left( \left. \frac{a_{L}}{n}\sum_{i=1}^n\frac{\theta_i y_i\mathsf{1}_{\{\theta_i y_i \geq 0 \}}}{1 + c_{U}e^{-a_Ly_i^2/2}} - \frac{a_{L}}{n}\sum_{i=1}^n\ex \left[\frac{\theta_i y_i\mathsf{1}_{\{\theta_i y_i \geq 0 \}}}{1 + c_{U}e^{-a_Ly_i^2/2}}\right] - \Delta_{4n} \leq -t/2  \right\vert \mathcal{E}\right),  
 \end{align}
where 
\be \Delta_{4n} \vcentcolon =  \frac{a}{n}\sum_{i=1}^n\ex \left[\frac{\theta_i y_i\mathsf{1}_{\{\theta_i y_i \geq 0 \}}}{1 + ce^{-ay_i^2/2}}\right]-\frac{a_L}{n}\sum_{i=1}^n\ex \left[\frac{\theta_i y_i\mathsf{1}_{\{\theta_i y_i \geq 0 \}}}{1 + c_{U}e^{-a_Ly_i^2/2}}\right].  \label{eq:Deltas2_def} \ee
Next, 
\begin{equation} 
\label{eq7_thms1}	
\begin{split}
 &\prb\left(\left. \frac{a_{L}}{n}\sum_{i=1}^n\frac{\theta_i y_i\mathsf{1}_{\{\theta_i y_i \geq 0 \}}}{1 + c_{U}e^{-a_Ly_i^2/2}} - \frac{a_{L}}{n}\sum_{i=1}^n\ex \left[\frac{\theta_i y_i\mathsf{1}_{\{\theta_i y_i \geq 0 \}}}{1 + c_{U}e^{-a_Ly_i^2/2}}\right] \leq -t/2 \right \vert \mathcal{E} \right) \\
 & \leq \frac{1}{\prb(\mathcal{E})}{\prb\left( \frac{a_{L}}{n}\sum_{i=1}^n \left( \frac{\theta_i y_i\mathsf{1}_{\{\theta_i y_i \geq 0 \}}}{1 + c_Ue^{-a_Ly_i^2/2}} - \ex \left[\frac{\theta_i y_i\mathsf{1}_{\{\theta_i y_i \geq 0 \}}}{1 + c_{U}e^{-a_Ly_i^2/2}}\right] \right) \leq -t/2 \right)} 
 \leq \frac{e^{-nkt^2}}{\prb(\mathcal{E})},
 \end{split}
 \end{equation}
where the last inequality follows from Lemma \ref{lem_chung1}, and $k$ is some absolute positive constant due to Assumption A. Note that to obtain an inequality of the form \eqref{eq7_thms1}, instead of Lemma \ref{lem_chung1}, we cannot use a lower tail inequality result that is the counterpart of \eqref{eq:lem2a_st1}  because  the constant $k$ would be proportional to $c_U^{-1} = (c + \kappa_1 u)^{-1}$ which cannot be bounded from below unlike $c_L^{-1}$.

Next, it is shown in Appendix \ref{app:Delta12} that $\Delta_{4n} \leq \kappa_4 u$. Using this bound on $\Delta_{4n}$ and \eqref{eq7_thms1} in \eqref{eq6_thms1} we obtain
\ben
\begin{split}
&\prb\left( \left. \frac{1}{n}\sum_{i=1}^n\frac{a_{L}\theta_i y_i}{1 + c_{U}e^{-a_Ly_i^2/2}} - \frac{1}{n}\sum_{i=1}^n\ex \left[\frac{a_{L}\theta_i y_i}{1 + c_{U}e^{-a_Ly_i^2/2}}\right] - \Delta_{4n} \leq -t/2  - \kappa_4u  \right\vert \mathcal{E}\right) \leq \frac{e^{-nkt^2}}{ \prb (\mathcal{E})}.
\end{split}
\een
 Choosing $u = t$, we finally have, from \eqref{eq6_thms1},
 \begin{align}\label{eq8_thms1}
 \prb (\mathcal{E}) \prb \left( s_n^+ \leq -t/2 \right \vert \mathcal{E}) \leq e^{-nkt^2}
\end{align}
for a suitable absolute positive constant $k$. 

 Hence, using \eqref{eq4_thms1} and \eqref{eq8_thms1}, we arrive at 
\begin{align} \label{eq:thm1a_eq1}
 \prb (\mathcal{E})  \prb \left( \vert s_n^+  \vert \geq t \vert \mathcal{E} \right) \leq 2e^{-nk t^2}.
\end{align}

\noindent In a similar manner, it is straightforward to obtain
\begin{align}\label{eq:thm1a_eq2}
 \prb (\mathcal{E})  \prb \left( \vert s_n^-  \vert \geq t \vert \mathcal{E} \right) \leq 2e^{-nk t^2}.
\end{align}

\noindent Using \eqref{eq:thm1a_eq1} and \eqref{eq:thm1a_eq2} in \eqref{eq_sn_bound} and recalling that $u=t$, we finally obtain
\begin{equation*}
 \prb \left( \vert s_n  \vert \geq t  \right) \leq 6e^{-nk\min(t,t^2)}
\end{equation*}
for some suitable absolute positive constant $k$.

\subsection{Proof of Theorem  \ref{thm1b}} \label{sec_proof_thm1b}

 For $i \in [n]$, let
 \begin{align*}
      h(w_i) \vcentcolon=  \hat{\theta}_{ST,i}- \theta_i = \left\{ \begin{array}{cc}
                                                                    - \theta_i,  &     -\lambda - \theta_i <  w_i < \lambda - \theta_i \\
                                                                    w_i - \lambda, &  w_i \geq \lambda - \theta_i \\
                                                                    w_i + \lambda, & w_i \leq -\lambda - \theta_i\\
                                                                   \end{array} \right.
 \end{align*} 
 where $w_i \sim \mathcal{N}(0,1)$. Let $g(w_i) \vcentcolon=  (\hat{\theta}_{ST,i}- \theta_i)^2$. We show that $g$ is pseudo-Lipschitz of order $2$ with pseudo-Lipschitz constant $\max(1,2\lambda)$, and then apply Lemma \ref{lem_plipschitz} to arrive at \eqref{eq_conc_thm1b}. It is straightforward to note that $h$ is Lipschitz with Lipschitz constant $1$, i.e., $\vert h(x) - h(y) \vert \leq \vert x- y \vert$,  $\forall x, y \in \mathbb{R}$. Now, for any $w_{i,1}, w_{i,2} \in \mathbb{R}$,
 \begin{align*}
  \vert g(w_{i,1}) - g(w_{i,2}) \vert & =  \vert h(w_{i,1}) + h(w_{i,2}) \vert  \vert h(w_{i,1})- h(w_{i,2}) \vert \\
  & \leq \vert h(w_{i,1}) + h(w_{i,2}) \vert \vert w_{i,1}- w_{i,2} \vert  \\
  & \leq (\vert h(w_{i,1}) \vert + \vert h(w_{i,2}) \vert) \vert w_{i,1}- w_{i,2} \vert  \\
  & \leq \left(2\lambda + \vert w_{i,1} \vert + \vert w_{i,2} \vert \right)  \vert w_{i,1}- w_{i,2} \vert \\
  & \leq \max(1,2\lambda)\left(1 + \vert w_{i,1} \vert + \vert w_{i,2} \vert \right)  \vert w_{i,1}- w_{i,2} \vert.
 \end{align*}
Hence, we obtain the pseudo-Lipschitz constant $L$  to be $\max(1,2\lambda)$.  Therefore, we can apply  Lemma \ref{lem_plipschitz} 
to obtain the desired concentration result \eqref{eq_conc_thm1b} for  $\frac{1}{n} \sum_{i=1}^n g(w_i)= \frac{1}{n}  \sum_{i=1}^n  (\hat{\theta}_{ST,i}- \theta_i)^2$.

 \subsection{Proof of Theorem \ref{thm2}} \label{sec_hybrid_proof}
 
 Without loss of generality, for a chosen $t > 0$, we can assume that 
 \be \frac{1}{n} L_{sep}(\bst,\by)  > t + \kappa_n  \label{eq:lsep_assume} \ee 
 because otherwise, it is clear that 
\[   \frac{1}{n} L(\bst,\hat{\bst}_H(\by)) \leq  \frac{1}{n} L_{min}(\bst,\by) +  t + \kappa_n \]
and \eqref{eq:hybrid_loss} trivially holds. In what follows, we use $K$ and $k$ as generic universal constants that appear in the concentration inequalities. These constants are independent of $n$, but their values  change as we proceed through the proof. 

 Let us first suppose that $\hat{\bst}_{EB}$ is the better estimator of $\bst$ for the given realization $\by$.  Then, recalling the definition of $\gamma_{\mbf{y}}$ in \eqref{eq_gamma}, the desired probability can be bounded as 
 \begin{align}
 \nonumber
 & \mathbb{P}\left( \frac{1}{n} L(\bst,\hat{\bst}_H(\by))  \geq  \frac{1}{n}L(\bst,\hat{\bst}_{EB}(\by)) + t + \kappa_n \right)  \\
 &  \leq \mathbb{P}(\gamma_{\mbf{y}} = 0) 
  = \mathbb{P}\left(\frac{1}{n}\hat{R}(\bst,  \hat{\bst}_{EB}(\by)) - \frac{1}{n}\hat{R}(\bst,  \hat{\bst}_{ST}(\by)) > 0 \right).
 \label{eq:desired_prob}
 \end{align}
The RHS of \eqref{eq:desired_prob} is bounded as follows. Using the triangle inequality, we have for any $u >0$,
\begin{align}
& \prb  \left( \left \vert \frac{1}{n} \hat{R}(\bst,\hat{\bst}_{EB}(\by)) -  \frac{1}{n} L(\bst,\hat{\bst}_{EB}(\by)) \right \vert \geq u + \kappa_n  \right)  \nonumber  \\
& \leq  \prb \left( \frac{1}{n} \abs{L(\bst,\hat{\bst}_{EB}(\by)) - R_2(\bst, \bsth_{EB})}  + \frac{1}{n} \abs{\hat{R}(\bst,\hat{\bst}_{EB}(\by)) - R_1(\bst, \bsth_{EB})} \right. \nonumber \\
& \qquad \qquad  \left. + \frac{1}{n} \abs{R_2(\bst, \bsth_{EB}) - R_1(\bst, \bsth_{EB})} \geq u + \kappa_n  \right)  \nonumber  \\
& \stackrel{(a)}{\leq} \prb \left( \frac{1}{n} \abs{L(\bst,\hat{\bst}_{EB}(\by)) - R_2(\bst, \bsth_{EB})} \geq \frac{u}{2} \right) + 
\prb \left( \frac{1}{n} \abs{\hat{R}(\bst,\hat{\bst}_{EB}(\by)) - R_1(\bst, \bsth_{EB})} \geq \frac{u}{2}  \right) \nonumber \\
& \stackrel{(b)}{\leq} Ke^{-nk\min(u,u^2)},   \label{eq_1_est_risk}
\end{align}
where inequality $(a)$ uses the definition of $\kappa_n$ in \eqref{eq:kappa_n}, and inequality $(b)$ is obtained using \eqref{eq_conc_thm1} and \eqref{eq_conc_thm1a}. Similarly, using \eqref{eq_RST_conc} and \eqref{eq_conc_thm1b}, we obtain for any $u >0$,
\begin{align}
 \label{eq_2_est_risk}
  \mathbb{P}\left( \frac{1}{n}\left \vert L(\bst,\hat{\bst}_{ST}(\by)) - \hat{R}(\bst,\hat{\bst}_{ST}(\by)) \right \vert \geq u \right) \leq 4e^{-nk\min(u,u^2)}.
 \end{align}
Combining \eqref{eq_1_est_risk} and \eqref{eq_2_est_risk}, and using the definition of $L_{sep}$ in \eqref{eq:lsep_def}, we have for $u >0$
\begin{align}
\mathbb{P}\left(\frac{1}{n}\hat{R}(\bst,  \hat{\bst}_{EB}(\by)) - \frac{1}{n}\hat{R}(\bst,  \hat{\bst}_{ST}(\by)) +  \frac{L_{sep}(\bst,\by)}{n}  > 2u +\kappa_n \right)
 & \leq Ke^{-nk\min(u,u^2)}.
 \label{eq_3_est_risk}
\end{align}
Now, choosing $2u = \frac{L_{sep}(\bst,\by)}{n} - \kappa_n $ (which is at least $t$ by the assumption \eqref{eq:lsep_assume}) yields 
\ben
\mathbb{P}\left(\frac{1}{n}\hat{R}(\bst,  \hat{\bst}_{EB}(\by)) - \frac{1}{n}\hat{R}(\bst,  \hat{\bst}_{ST}(\by)) > 0 \right) \leq Ke^{-nk\min(t,t^2)}.
\een
Using this in \eqref{eq:desired_prob}, and noting that the other case ($\hat{\bst}_{ST}$ is the better estimator) can be similarly analysed, we arrive at the concentration result \eqref{eq:hybrid_loss}.

To prove  \eqref{eq:hybrid_risk},  let $X_n \vcentcolon=   \frac{1}{n}\left [ L(\bst,\hat{\bst}_H(\by)) - L_{min}( \boldsymbol{\theta},\by) \right ] \geq 0$. We have,
 \begin{align*}
  \ex \left[ X_n  \right] & = \int_{0}^{\infty}\prb\left( X_n  \geq u \right) du = \int_{0}^{ \kappa_n }\prb\left(  X_n  \geq u \right) du + \int_{\kappa_n }^{\infty}\prb\left( X_n  \geq u \right) du \\
  &\leq \int_{0}^{\kappa_n} du + \int_{\kappa_n}^{\infty}\prb\left(  X_n  \geq u \right) du  = \kappa_n + \int_{0}^{\infty}\prb\left(  X_n  \geq t + \kappa_n \right) dt.
 \end{align*}
Note that $\kappa_n$ is an $\mathcal{O}(1/\sqrt{n})$ term. So, using \eqref{eq:hybrid_loss} and the steps of the proof of Lemma \ref{lem3}, we obtain
 \begin{align*}
 \ex \left[ X_n  \right] \leq \kappa_n  + \int_{0}^{\infty}Ke^{-nk\min(t,t^2)} dt \leq  \kappa_n  + \frac{C}{\sqrt{n}}\left(1 + \frac{1}{\sqrt{n}}\right)   
 \end{align*}
for some positive constant $C$. So,
\begin{align}
 \frac{ R(\bst,\hat{\bst}_H)}{n} \leq  \frac{\ex \left[L_{min}( \boldsymbol{\theta},\by)\right]}{n} + \kappa_n +  \frac{C}{\sqrt{n}}\left(1 + \frac{1}{\sqrt{n}}\right). 
 \label{eq:RbstH_bnd}
\end{align}
It trivially follows that $\ex \left[ L_{min}( \boldsymbol{\theta},\by)\right] \leq \min \{R(\bst,\hat{\bst}_{EB}), R(\bst,\hat{\bst}_{ST})  \}$. Using this in \eqref{eq:RbstH_bnd} and noting that $\kappa_n=\mc{O}(1/\sqrt{n})$ completes the proof of \eqref{eq:hybrid_risk}.


\section*{Funding}
This work was supported  by the European Commission [Marie Curie Career Integration Grant, Agreement Number 631489]; the Isaac Newton Trust; and the Engineering and Physical Sciences Research Council [EP/N013999/1].

\appendix
\vspace{10pt}


\section{Proof of Lemma \ref{lem1}} \label{sec_appendix}
 Let $Z_i \vcentcolon= \frac{y_i^2}{n\left(1 + ce^{-ay_i^2}\right)^p}$. Since $Z_i \geq 0$, from Lemma \ref{lem_chung1}, we have
 \begin{align}\label{eq1_lem1}
   \mathbb{P}\left(f(\mathbf{y}) - \mathbb{E}\left[ f(\mathbf{y})\right]  \leq -t \right) \leq e^{-nt^2/\mathsf{k}}
 \end{align}
where $\mathsf{k} = {2n \sum_{i=1}^n\mathbb{E}[Z_i^2]} \leq {\frac{2}{n}\sum_{i=1}^n\mathbb{E}[y_i^4]} = \frac{2}{n}\sum_{i=1}^n \left(\theta_i^4 + 6\theta_i^2 + 6 \right)$. To prove the upper tail inequality we proceed as follows. For $g : \mathbb{R}^n \to \mathbb{R}$ which is differentiable, we have \cite[Sec. 2.3]{mjwain}
\begin{align}
 \ex\left[ e^{s\{g(\mathbf{y}) - \mathbb{E}[g(\mathbf{y})]\}} \right] \leq \ex \left[e^{{s^2\pi^2 \norm{\nabla g(\mathbf{y})}^2}/{8}} \right].
 \label{eq:gy_mgf}
\end{align}
We take $g(\mathbf{y}) \vcentcolon= \sqrt{f(\mathbf{y})}$, and  obtain an upper tail bound  for $g$ by showing that $\norm{\nabla g(\mathbf{y})}^2$ is bounded. 
Now,
\begin{align*}
 \frac{\partial g(\mathbf{y})}{\partial w_i}  = \frac{ y_i \left[1+ce^{-ay_i^2} +  pacy_i^2e^{-ay_i^2}\right]}{\sqrt{n} \left(1+ce^{-ay_i^2}\right)^{{p}+1}\sqrt{\sum_{j=1}^n{y_j^2}/\left(1+ce^{-ay_j^2}\right)^{p}}} .
\end{align*}
Hence, 
\begin{align*}
 n\norm{\nabla g(\mathbf{y})}^2 &  = \frac{1}{{\sum_{j=1}^n{y_j^2}/\left(1+ce^{-ay_j^2}\right)^{p}}}\sum_{i=1}^n \frac{ y_i^2 \left[1+ce^{-ay_i^2} +  pc\, ay_i^2e^{-ay_i^2}\right]^2}{\left(1+ce^{-ay_i^2}\right)^{{2p}+2}} \\
  & = \frac{1}{{\sum_{j=1}^n{y_j^2}/\left(1+ce^{-ay_j^2}\right)^{p}}} \sum_{i=1}^n \frac{ y_i^2}{\left(1+ce^{-ay_i^2}\right)^{p}}  \frac{\left[1+ce^{-ay_i^2} +  pc \, ay_i^2e^{-ay_i^2}\right]^2}{\left(1+ce^{-ay_i^2}\right)^{{p}+2}} \\
 & \leq \frac{1}{{\sum_{j=1}^n{y_j^2}/\left(1+ce^{-ay_j^2}\right)^{p}}} \sum_{i=1}^n \frac{ y_i^2}{\left(1+ce^{-ay_i^2}\right)^{p}} \left[1 + \frac{ pc\, a y_i^2e^{-ay_i^2}}{\left(1+ce^{-ay_i^2}\right)^{{p/2}+1}}\right]^2.
\end{align*}
Using the bound $a y_i^2e^{-ay_i^2} \leq 1/e$, we obtain
\begin{align*}
 n\norm{\nabla g(\mathbf{y})}^2 \leq C
\end{align*}
where $C \vcentcolon= \left(1 + \frac{pc}{e} \right)^2$. Hence,
\begin{align*}
 \ex\left[ e^{s\{g(\mathbf{y}) - \mathbb{E}[g(\mathbf{y})]\}} \right] \leq e^{{s^2\pi^2C}/{8n}} .
\end{align*}
Hence, $g(\mathbf{y})$ is sub-Gaussian \cite[Sec. 2.3]{boucheron2} with variance factor at most $\frac{\pi^2C}{4n}$. Therefore, using the Cram{\'e}r-Chernoff bound, we 
 have for $t >0$:
\begin{align}
 \mathbb{P}\left(g(\mathbf{y}) \geq \mathbb{E}\left[ g(\mathbf{y})\right] +  t \right) \leq e^{-{2nt^2}/{\pi^2C}}.
 \label{eq:gy_bound}
\end{align}
Recalling that $g(\mbf{y}) = \sqrt{f(\mathbf{y})}$ and using Jensen's inequality, we have $\mathbb{E}\left[f(\mathbf{y}) \right] = \mathbb{E}\left[g^2(\mathbf{y}) \right] \geq \left(\mathbb{E}\left[g(\mathbf{y})\right]\right )^2$. We therefore have 
\begin{align} \nonumber
 \mathbb{P}\left(f(\mathbf{y}) \geq \mathbb{E}\left[f(\mathbf{y})\right] +  t^2 + 2t \mathbb{E}\left[g(\mathbf{y})\right] \right)
 &  \leq \mathbb{P}\left(f(\mathbf{y}) \geq \left(\mathbb{E}\left[g(\mathbf{y})\right]\right )^2 +  t^2 + 2\mathbb{E}\left[g(\mathbf{y})\right]t \right) \\ \nonumber
&  = \mathbb{P}\left( g(\mathbf{y}) \geq \mathbb{E}\left[ g(\mathbf{y})\right] +  t  \right) \\ \label{eq:temp}
&  \leq e^{-{2nt^2}/{\pi^2C}},
 \end{align}
 where the last inequality follows from \eqref{eq:gy_bound}.
Note that 
\begin{align*}
 t^2 + 2t \mathbb{E}[g(\mathbf{y})] \leq \left \{ \begin{array}{cc}
                                                   t^2(1 + 2 \mathbb{E}[g(\mathbf{y})]) & \textrm{when } t \geq 1,\\
                                                   t(1 + 2 \mathbb{E}[g(\mathbf{y})]) & \textrm{otherwise},\\
                                                  \end{array}
 \right.
\end{align*}
 and recall  that $(\mathbb{E}\left[g(\mathbf{y}) \right])^2 \leq \ex f(\mathbf{y}) \leq \Vert\bst\Vert^2/n +1$. Now, setting  $u = \max(t,t^2)(1 + 2\sqrt{1 + \Vert\bst\Vert^2/n})$, from \eqref{eq:temp} we obtain, 
 \begin{align*}
   \mathbb{P}\left(f(\mathbf{y}) \geq \mathbb{E}\left[f(\mathbf{y})\right] + u \right) \leq \mathbb{P}\left(f(\mathbf{y}) \geq \mathbb{E}\left[f(\mathbf{y})\right] +  t^2 + 2t \mathbb{E}\left[g(\mathbf{y})\right] \right) \leq e^{-{2nt^2}/{\pi^2C}}.
 \end{align*}

 \noindent Therefore, we have, for every $u > 0$:
\begin{align}\label{eq2_lem1}
\prb\left(f(\mathbf{y}) - \mathbb{E}\left[ f(\mathbf{y})\right]  \geq u \right) \leq e^{-nk\min(u,u^2)/\max(1,\Vert\bst\Vert^2/n)} 
\end{align}
for a suitable absolute positive constant $k$. Combining \eqref{eq1_lem1} and \eqref{eq2_lem1} completes the proof. \qedhere


\section{Bounds on $\Delta_{1n}, \Delta_{2n}$, $\Delta_{3n}$, $\Delta_{4n}$} \label{app:Delta12}

\textbf{Bound on $\Delta_{1n}$}:
Recall that $a=\frac{b}{b + \e}$ and $a_U= \min \{ \frac{b+u}{b+\e}, 1\}$. Hence, 
$a_U - a \leq u/b$, and it follows that $a_U^2 - a^2 = (a_U - a)(a_U + a) \leq \frac{2}{b}u$.  We therefore have
\begin{align}  \nonumber
 & \Delta_{1n}  = \frac{a_U^2}{n}\sum_{i=1}^n\ex \left[\frac{y_i^2}{\left(1 + c_Le^{-a_Uy_i^2/2}\right)^2}\right] - \frac{a^2}{n}\sum_{i=1}^n\ex \left[\frac{y_i^2}{\left(1 + ce^{-ay_i^2/2}\right)^2}\right] \\ \nonumber 
 & \leq  \frac{(a^2 + (2/b)u)}{n}\sum_{i=1}^n\ex \left[\frac{y_i^2}{\left(1 + c_Le^{-a_Uy_i^2/2}\right)^2}\right] - \frac{a^2}{n}\sum_{i=1}^n\ex \left[\frac{y_i^2}{\left(1 + ce^{-ay_i^2/2}\right)^2}\right]\\ \nonumber
 & = \frac{a^2}{n}\sum_{i=1}^n\ex \left[\frac{y_i^2}{\left(1 + c_Le^{-a_Uy_i^2/2}\right)^2} - \frac{y_i^2}{\left(1 + ce^{-ay_i^2/2}\right)^2}\right] + \frac{(2/b)u}{n}\sum_{i=1}^n\ex \left[\frac{y_i^2}{\left(1 + c_Le^{-a_Uy_i^2/2}\right)^2}\right]
\end{align}
\begin{align}
\nonumber 
 & \leq \frac{a^2}{n}\sum_{i=1}^n\ex \left[y_i^2\left( \frac{\left(ce^{-ay_i^2/2} -c_Le^{-a_Uy_i^2/2}\right) \left( 2 + c_Le^{-a_Uy_i^2/2} + ce^{-ay_i^2/2} \right) }{\left(1 + c_Le^{-a_Uy_i^2/2}\right)^2\left(1 + ce^{-ay_i^2/2}\right)^2}\right)\right] + \frac{2u}{b} \ex \left[\frac{\norm{\by}^2}{n} \right] \\ 
  & \leq \frac{a^2}{n}\sum_{i=1}^n\ex \left[2y_i^2\left( \frac{ce^{-ay_i^2/2}\left(1 -e^{-(a_U-a)y_i^2/2}\right) + \kappa_1 u e^{-a_Uy_i^2/2} }{\left(1 + c_Le^{-a_Uy_i^2/2}\right)\left(1 + ce^{-ay_i^2/2}\right)}\right)\right] + \frac{2(b + 1)u}{b},
 \label{eq10_thm1}
\end{align}
where the last inequality holds  because the definition of $c_L$ in \eqref{eq_cL} implies $c - c_L \leq \kappa_1 u$ . Now, $ce^{-ay_i^2/2}\left(1 -e^{-(a_U-a)y_i^2/2}\right)$ has a maximum when $ e^{-(a_U-a)y_i^2/2} = a/a_U$, and so,
\begin{align} 
 \frac{ce^{-ay_i^2/2}\left(1 -e^{-(a_U-a)y_i^2/2}\right) + \kappa_1 u e^{-a_Uy_i^2/2} }{\left(1 + c_Le^{-a_Uy_i^2/2}\right)\left(1 + ce^{-ay_i^2/2}\right)} & \leq 1 - \frac{a }{a_U} + \kappa_1 u   \leq \frac{1}{ab} u + \kappa_1 u. \label{eq_kappa2}
\end{align}
Using this in \eqref{eq10_thm1}, we obtain 
\begin{align}
\Delta_{1n} \leq \kappa_3 u, 
 \label{eq:au_a_close}
\end{align}
 where $\kappa_3 = 2 \frac{(b + 1)}{b} (1 + a + \kappa_1 a^2 b)u$.
 
 \textbf{Bound on $\Delta_{2n}$}: Recalling that $a=\frac{b}{b + \e}$ and $a_L= \max \{ \frac{b - u}{b+\e}, 0\}$,
it follows that $a - a_L\leq u/b$ and hence $a^2 - a_L^2 = (a - a_L)(a + a_L) \leq 2 u/b$. We therefore have
\begin{align} \nonumber
 & \frac{a^2}{n}\sum_{i=1}^n\ex \left[\frac{y_i^2}{\left(1 + ce^{-ay_i^2/2}\right)^2}\right] -  \frac{a_L^2}{n}\sum_{i=1}^n\ex \left[\frac{y_i^2}{\left(1 + c_Ue^{-a_Ly_i^2/2}\right)^2}\right] \\ \nonumber
 & \leq  \frac{a^2}{n}\sum_{i=1}^n\ex \left[\frac{y_i^2}{\left(1 + ce^{-ay_i^2/2}\right)^2}\right] - \frac{(a^2 - 2u/b)}{n}\sum_{i=1}^n\ex \left[\frac{y_i^2}{\left(1 + c_Ue^{-a_Ly_i^2/2}\right)^2}\right] \\ \nonumber
 & = \frac{a^2}{n}\sum_{i=1}^n\ex \left[\frac{y_i^2}{\left(1 + ce^{-ay_i^2/2}\right)^2} - \frac{y_i^2}{\left(1 + c_Ue^{-a_Ly_i^2/2}\right)^2}\right]  + \frac{2 u/b}{n}\sum_{i=1}^n\ex \left[\frac{y_i^2}{\left(1 + c_Ue^{-a_Ly_i^2/2}\right)^2}\right] \\ \nonumber
 & \leq \frac{a^2}{n}\sum_{i=1}^n\ex \left[y_i^2\left( \frac{\left(c_Ue^{-a_Ly_i^2/2} -ce^{-ay_i^2/2}\right)\left( 2 + c_Ue^{-a_Ly_i^2/2} + ce^{-ay_i^2/2}  \right) }{\left(1 + c_Ue^{-a_Ly_i^2/2}\right)^2\left(1 + ce^{-ay_i^2/2}\right)^2}\right)\right] + \frac{2 u}{b} \ex \left[\frac{\norm{\by}^2}{n} \right] \\ 
  & \leq \frac{a^2}{n}\sum_{i=1}^n\ex \left[2y_i^2\left( \frac{ce^{-a_Ly_i^2/2}\left(1-e^{-(a-a_L)y_i^2/2}\right) + \kappa_1 u e^{-a_Ly_i^2/2} }{\left(1 + c_Ue^{-a_Ly_i^2/2}\right)\left(1 + ce^{-ay_i^2/2}\right)}\right)\right]  + \frac{2 u}{b} \ex \left[\frac{\norm{\by}^2}{n} \right]   \label{eq11_thm1}
\end{align}
where the last inequality holds  because $c_U - c \leq \kappa_1 u$, from the definition of $c_U$ in \eqref{eq_cU}. Now, $ce^{-a_Ly_i^2/2}\left(1-e^{-(a-a_L)y_i^2/2}\right)$ has a maximum when $ e^{-(a-a_L)y_i^2/2} = a_L/a$, and so,
\begin{align*}
\frac{ce^{-a_Ly_i^2/2}\left(1-e^{-(a-a_L)y_i^2/2}\right) + \kappa_1 u e^{-a_Ly_i^2/2} }{\left(1 + c_Ue^{-a_Ly_i^2/2}\right)\left(1 + ce^{-ay_i^2/2}\right)} \leq \left(1- \frac{a_L}{a} \right) + \kappa_1 u \leq \frac{u}{ab} + \kappa_1 u.
\end{align*}
Using this in \eqref{eq11_thm1}, we obtain 
\begin{align}
\Delta_{2n} \leq \kappa_3 u, 
 \label{eq:al_a_close}
\end{align}
 where $\kappa_3 = 2 \frac{(b + 1)}{b} (1 + a + \kappa_1 a^2 b)u$.  Equations \eqref{eq:au_a_close} and \eqref{eq:al_a_close} give the required bounds on $\Delta_{1n}$ and $\Delta_{2n}$, respectively. 
 
 \textbf{Bound on $\Delta_{3n}$ and $\Delta_{4n}$}:
 We have
\begin{align}  \nonumber
 & \Delta_{3n} = \frac{a_{U}}{n}\sum_{i=1}^n\ex \left[\frac{\theta_i y_i\mathsf{1}_{\{\theta_i y_i \geq 0 \}}}{1 + c_{L}e^{-a_Uy_i^2/2}}\right] - \frac{a}{n}\sum_{i=1}^n\ex \left[\frac{\theta_i y_i\mathsf{1}_{\{\theta_i y_i \geq 0 \}}}{1 + ce^{-ay_i^2/2}}\right] \\ \nonumber 
 & \leq  \frac{(a + (1/b)u)}{n}\sum_{i=1}^n\ex \left[\frac{\theta_i y_i\mathsf{1}_{\{\theta_i y_i \geq 0 \}}}{1 + c_{L}e^{-a_Uy_i^2/2}}\right] - \frac{a}{n}\sum_{i=1}^n\ex \left[\frac{\theta_i y_i\mathsf{1}_{\{\theta_i y_i \geq 0 \}}}{1 + ce^{-ay_i^2/2}}\right]  \\ \nonumber 
 & = \frac{a}{n}\sum_{i=1}^n\ex \left[\frac{\theta_i y_i\mathsf{1}_{\{\theta_i y_i \geq 0 \}}}{1 + c_{L}e^{-a_Uy_i^2/2}} - \frac{\theta_i y_i\mathsf{1}_{\{\theta_i y_i \geq 0 \}}}{1 + ce^{-ay_i^2/2}}\right] + \frac{u}{b n}\sum_{i=1}^n\ex \left[\frac{\theta_i y_i\mathsf{1}_{\{\theta_i y_i \geq 0 \}}}{1 + c_{L}e^{-a_Uy_i^2/2}}\right]\\  \nonumber
 &\leq \frac{a}{n}\sum_{i=1}^n\ex \left[\theta_i y_i\mathsf{1}_{\{\theta_i y_i \geq 0 \}} \left( \frac{ce^{-ay_i^2/2} -c_Le^{-a_Uy_i^2/2}}{\left(1 + c_Le^{-a_Uy_i^2/2}\right)\left(1 + ce^{-ay_i^2/2}\right)}\right)\right] + \frac{u}{b n} \sum_{i=1}^n\ex \left[\theta_i y_i\mathsf{1}_{\{\theta_i y_i \geq 0 \}}\right]\\ \nonumber 
 & \stackrel{(1)}{\leq} \frac{a}{n}\sum_{i=1}^n\ex \left[ \theta_i y_i\mathsf{1}_{\{\theta_i y_i \geq 0 \}} \left( \frac{ce^{-ay_i^2/2}\left(1 -e^{-(a_U-a)y_i^2/2}\right) + \kappa_1 u e^{-a_Uy_i^2/2} }{\left(1 + c_Le^{-a_Uy_i^2/2}\right)\left(1 + ce^{-ay_i^2/2}\right)}\right)\right] + \frac{u}{bn}\left(\Vert \bst \Vert \sqrt{\ex \Vert \by \Vert^2}\right)\\ \nonumber
 & \stackrel{(2)}{\leq} \frac{a}{n}\left(\frac{1}{ab} + \kappa_1 \right)  u \left(\Vert \bst \Vert \sqrt{\ex \Vert \by \Vert^2}\right) + \frac{u}{bn} \left(\Vert \bst \Vert \sqrt{\ex \Vert \by \Vert^2}\right) =  \kappa_4 u,
\end{align}
where 
\begin{align*} 
 \kappa_4 \vcentcolon =  \sqrt{\frac{b+1}{b}}\left(2 + \kappa_1 ab  \right).
\end{align*}
In the above bound, inequality $(1)$ is obtained  using  $c - c_L \leq \kappa_1 u$, and inequality $(2)$ using \eqref{eq_kappa2}. Finally, the expression for $\kappa_4$  is obtained by recalling that $b = \Vert \bst \Vert^2/n$.

The bound for $\Delta_{4n}$ is straightforward to obtain in a similar manner and is therefore not detailed here.
\qedhere


\bibliographystyle{ieeetr}
{\small
\bibliography{References}
}

\end{document}